\documentclass[12pt]{article}

\usepackage{graphicx}
\usepackage{wrapfig}
\usepackage{amssymb,amsfonts,amsmath,amsthm,amscd}
\usepackage{color}
\usepackage[english]{babel}
\usepackage[latin1]{inputenc}
\usepackage{times}
\usepackage{geometry}
\usepackage[breaklinks,pdfstartview=FitH,colorlinks,linktocpage,bookmarks=false]{hyperref}
\usepackage{newtxtext,newtxmath}
\usepackage{pgf,tikz,circuitikz}\usetikzlibrary{arrows}
\usepackage{mathrsfs}
\usepackage{booktabs}
\usepackage{lscape}

\newtheorem{theorem}{Theorem}[subsection]
\newtheorem{proposition}[theorem]{Proposition}
\newtheorem{lemma}[theorem]{Lemma}

\newtheorem{corollary}[theorem]{Corollary}

\theoremstyle{remark}

\newtheorem{remark}[theorem]{Remark}

\theoremstyle{definition}

\newtheorem{definition}[theorem]{Definition}
\newcommand{\spleen}{\stackrel{*}{\frown}}

\newcommand{\mscomm}[1]{}
\newcommand{\varedit}[3]{#2}
\newcommand{\red}[1]{#1}
\newcommand{\edit}[2]{#2}
\newcommand{\remove}[2]{}
\newcommand{\move}[2]{#2}
\newcommand{\fix}[1]{#1}
\newcommand{\gram}[1]{#1}
\newcommand{\clarity}[1]{#1}

\long\def\comment#1\endcomment{}

\long\def\arxiv#1\endarxiv{#1}
\newcommand{\ifarxiv}[2]{#1}

\begin{document}


\title{Discrete field theory: symmetries and conservation laws}

\author{M. Skopenkov}
\date{}

\maketitle

\begin{abstract}
We present a general algorithm constructing a discretization of a classical field theory from a Lagrangian. We prove a new discrete Noether theorem relating symmetries to conservation laws and an energy conservation theorem not based on any symmetry. This gives exact conservation laws for several \varedit{R12P3}{theories, e.g., lattice}{-0.4} electrodynamics and gauge theory.\remove{R12P3}{} In particular, we construct a conserved discrete energy-momentum tensor, approximating the continuum one at least for free fields. The theory is stated in topological terms, such as coboundary and products of cochains.
\smallskip



\arxiv \vspace{-1cm} \endarxiv
\end{abstract}

\footnotetext[0]{The publication was prepared within the framework of the Academic Fund Program at the National Research University Higher School of Economics (HSE) in 2018-2019 (grant N18-01-0023) and by the Russian Academic Excellence Project ``5-100''. The author has also received support from the Simons--IUM fellowship.
}

{
\tableofcontents
}

\section{Introduction}\label{s:intro}


This work is a try to build a general \emph{discrete field theory}. This has the following motivation:
\begin{itemize}
\item
getting effective numeric algorithms for field theory;
\item
putting field theory to a mathematically rigorous basis;
\item
creating an alternative candidate for a fundamental field theory.
\end{itemize}

Numerous discretizations of particular field theories are known \cite{Arnold-etal-10,Courant-friedrichs-Lewy-28,Creuz-70,
Dimakis-etal-94,Desbrun-etal-08,Gawlik-Mullen-Pavlov-Marsden-Desbrun-10,
Gross-Kotiuga-04,Kron-44}. Our aim is \emph{not} to invent new discretizations but to extract and study the best among the known ones. Discretizations exhibiting exact (not just approximate) conservation laws have been proved to be most successful for computational purposes \cite{Gawlik-Mullen-Pavlov-Marsden-Desbrun-10}.
This leads us to the following \emph{principles of discretization}:
\begin{itemize}
\item
keep approximation of continuum theory;
\item
keep conservation laws exact;
\item
drop spatial symmetries easily.
\end{itemize}


These principles have a built-in difficulty: we have to drop most continuous symmetries, but usually, conservation laws are obtained just from such symmetries using the Noether theorem. We 
develop a new 
method to get discrete conservation laws. \remove{clarity}{} 
Compared to 
\cite{Dorodnitsyn-04,Gawlik-Mullen-Pavlov-Marsden-Desbrun-10,Hydon-Mansfield-04,
Kraus-15,Mardsen-etal-98},
it allows to write 
the conservation laws
explicitly as one-line formulae
(using 
standard topological notation) in numerous
examples.

The following basic warm-up results of discrete field theory are obtained in the present paper:
\begin{itemize}
\item
discretization of several field theories in a similar fashion keeping conservation laws exact (\S\ref{sec-examples});
\item
a new discrete Noether theorem relating symmetries to conservation laws (Theorems~\ref{th-Noether},\ref{th-Noether-covar});
\item
a new discrete energy conservation theorem not based on symmetry (Theorems~\ref{th-energy-conservation} and~\ref{cor-free}). \remove{grammar}{}
\end{itemize}

\vspace{-0.4cm}

%
%
%
%
%

\subsection{Quick start} \label{ssec-quick}

We start with an elementary and \remove{clarity}{}
informal description of one result (Theorem~\ref{cor-free}), in the simplest unknown particular case. It is an energy conservation theorem for lattice electrodynamics 
in $2$ spatial and $1$ time dimensions. For these small dimensions, we just \emph{draw} everything. The more realistic case of $3$ spatial and $1$ time dimensions is analogous; see~\S\ref{ssec-Electrodynamics}, where we state the result precisely.

\begin{wrapfigure}{r}{2.6cm}
\vspace{-0.4cm}
\includegraphics[width=2.6cm]{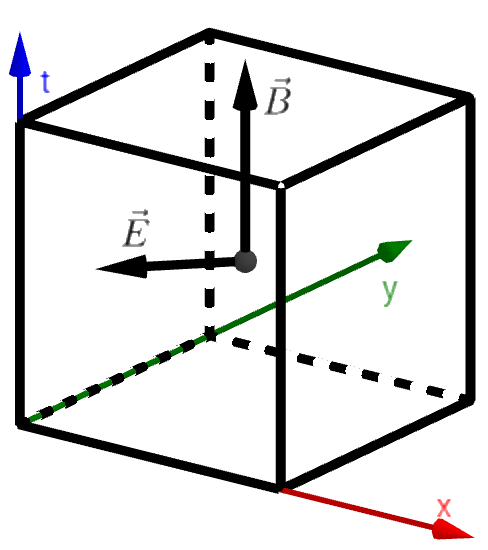}
\caption{Cube}
\label{fig-coordinates}
\vspace{-0.4cm}
\end{wrapfigure}
Recall briefly the energy conservation theorem in \emph{continuum} electrodynamics (the Poynting theorem). Let $x,y,t$ be the Cartesian coordinates in space; see Figure~\ref{fig-coordinates}. \emph{Electric} and  \emph{magnetic fields} are arbitrary smooth vector-valued functions $\vec {\mathrm{E}}(x,y,t)$ and $\vec{\mathrm{B}}(x,y,t)$ 
such that $\vec{\mathrm{E}}\perp Ot$ and $\vec {\mathrm{B}}\parallel Ot$. The \emph{energy density} and the \emph{energy flux} (\emph{the Poynting vector}) are the functions $\tfrac{1}{2}(\vec{\mathrm{E}}^2+\vec{\mathrm{B}}^2)$ and $\vec {\mathrm{E}}\times\vec{\mathrm{B}}$. 
The Poynting theorem asserts that under \emph{Maxwell's equations} (where $\vec{\mathrm{E}}=:(0,\mathrm{E}_x,\mathrm{E}_y)$ and $\vec{\mathrm{B}}=:(\mathrm{B}_t,0,0)$)
\begin{align*}
 \frac{\partial \mathrm{B}_t}{\partial t}
+\frac{\partial \mathrm{E}_y}{\partial x}
-\frac{\partial \mathrm{E}_x}{\partial y}
&=0;
&\frac{\partial \mathrm{E}_x}{\partial x}
+\frac{\partial \mathrm{E}_y}{\partial y}
&=0;
&\frac{\partial \mathrm{B}_t}{\partial x}
+\frac{\partial \mathrm{E}_y}{\partial t}
&=0;
&\frac{\partial \mathrm{B}_t}{\partial y}
-\frac{\partial \mathrm{E}_x}{\partial t}
&=0;
\end{align*}
the following identity holds for each cube with the edges parallel to the coordinate axes:
$$
  \int\limits_{\includegraphics[width=0.6cm]{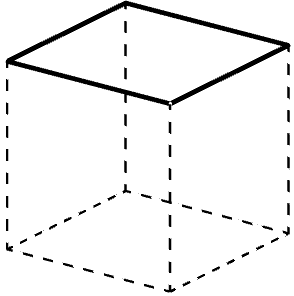}}
  \frac{\vec{\mathrm{E}}^2+\vec{\mathrm{B}}^2}{2}\,\mathrm{dA} -\int\limits_{\includegraphics[width=0.6cm]{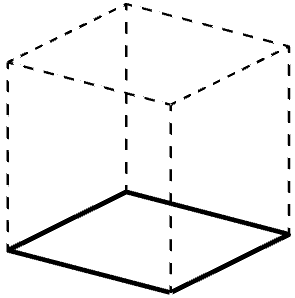}}
  \frac{\vec{\mathrm{E}}^2+\vec{\mathrm{B}}^2}{2}\,\mathrm{dA}=
  \int\limits_{\includegraphics[width=0.6cm]{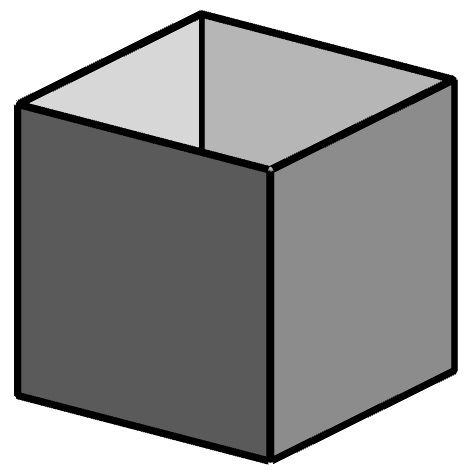}} \vec {\mathrm{E}}\times\vec{\mathrm{B}}\,\mathrm{d}\vec{\mathrm{n}}.
$$
Here the cube is shown by dotted lines, and the faces which a particular integral is taken over are in bold.  The first two integrals mean the total energy contained in the same square in the $Oxy$ plane at two different moments of time $t$. The third integral means the total inward energy flux through the boundary between these two moments. Thus the equation means energy conservation. 

Let us discretize. Dissect the unit cube into $N\times N\times N$ equal cubes. \clarity{Throughout this subsection}
by \emph{cubes} we mean the latter 
cubes, by \emph{faces} and \emph{edges} --- their faces and edges. A discrete \emph{electromagnetic field} $F$ is any real-valued function on the set of faces. Informally, its values
\clarity{$F(\includegraphics[width=0.6cm]{0+0+1+2s.png})$, $F(\includegraphics[width=0.6cm]{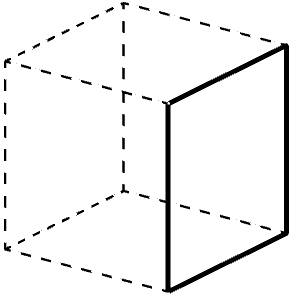})$,
$F(\includegraphics[width=0.6cm]{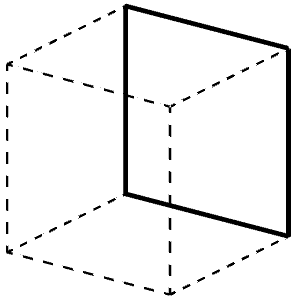})$}
discretize $-\mathrm{B}_t$, $\mathrm{E}_y$, $\mathrm{E}_x$ respectively, depending on 
face direction \edit{R2P2}{(for exterior-calculus fans: $F$ itself discretizes the
\emph{electromagnetic field}
$\mathrm{F}=-\mathrm{B}_t\,\mathrm{dx\wedge dy}+\mathrm{E}_x\,\mathrm{dt\wedge dx}+
\mathrm{E}_y\,\mathrm{dt\wedge dy}$).} \remove{clarity}{}
Hereafter a particular face \gram{at} which the function is evaluated
is in bold, and one of the adjacent 
cubes is shown by dotted lines to identify the face position.
The well-known discrete 
\emph{Maxwell's equations} are \varedit{R2P3}{}{2}
\begin{equation} \label{eq-Maxwell-1}
\hspace{-0.8cm}
\begin{aligned}
F(\includegraphics[width=0.6cm]{1+2s.png})- F(\includegraphics[width=0.6cm]{0+0+1+2s.png})-
F(\includegraphics[width=0.6cm]{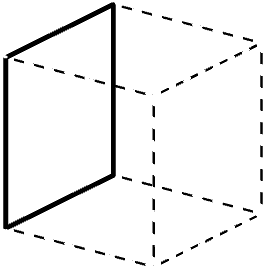})+ F(\includegraphics[width=0.6cm]{0+1+1+2s.png})+
F(\includegraphics[width=0.6cm]{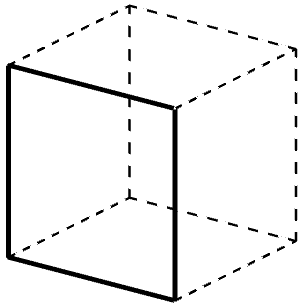})- F(\includegraphics[width=0.6cm]{0+1+2+2s.png})&=0;
\\
F(\hspace{-0.3cm}\begin{tabular}{c}\vspace{-0.0cm}
       \includegraphics[width=0.9cm]{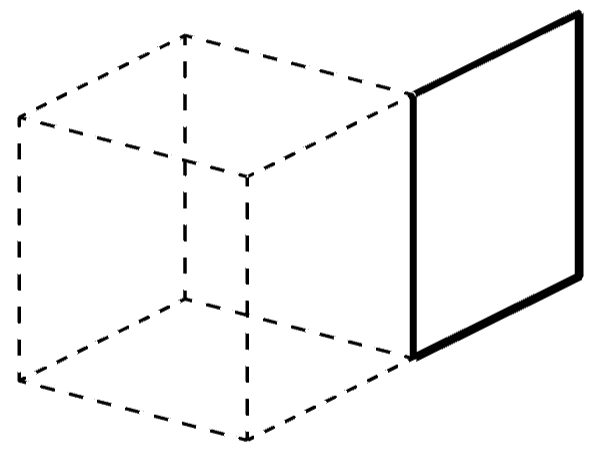}
       \end{tabular}\hspace{-0.2cm})-
F(\includegraphics[width=0.6cm]{0+1+1+2s.png})+
F(\hspace{-0.3cm}\begin{tabular}{c}\vspace{-0.0cm}
       \includegraphics[width=0.9cm]{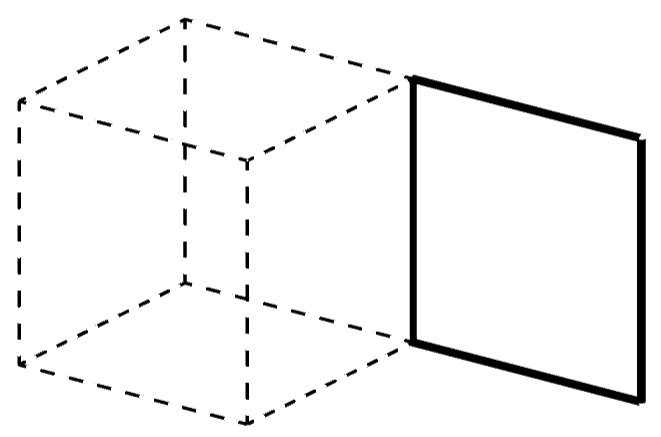}
       \end{tabular}\hspace{-0.2cm})-
F(\includegraphics[width=0.6cm]{0+1+2+2s.png}) &=0;
\\
F(\includegraphics[width=0.6cm]{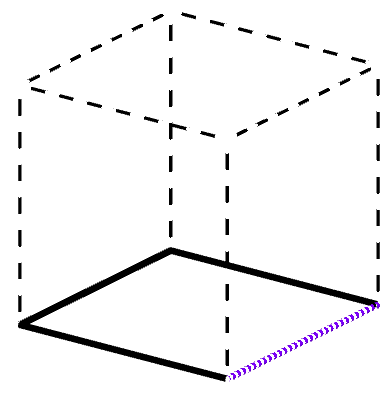})-
F(\hspace{-0.3cm}\begin{tabular}{c}\vspace{-0.0cm}
       \includegraphics[width=0.9cm]{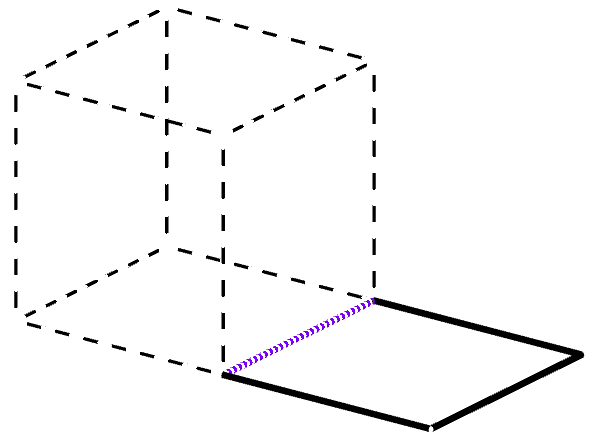}
       \end{tabular}\hspace{-0.2cm})-
F(\hspace{-0.3cm}\begin{tabular}{c}\vspace{-0.2cm}
       \includegraphics[width=0.6cm]{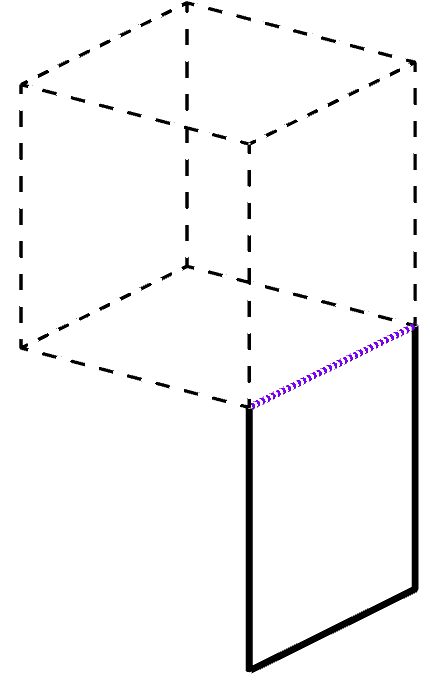}
       \end{tabular}\hspace{-0.2cm})+
F(\includegraphics[width=0.6cm]{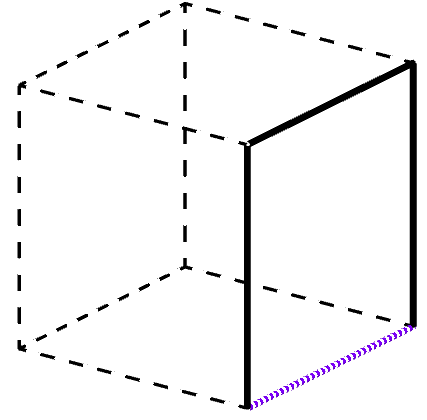})
&=0;
\\
F(\includegraphics[width=0.6cm]{1+2s.png})-
  F(\hspace{-0.3cm}\begin{tabular}{c}\vspace{-0.0cm}
       \includegraphics[width=0.9cm]{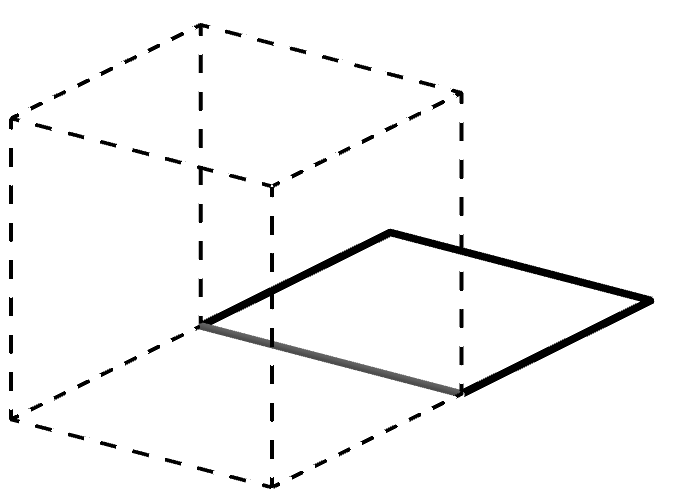}
       \end{tabular}\hspace{-0.2cm})+
  F(\hspace{-0.3cm}\begin{tabular}{c}\vspace{-0.2cm}
       \includegraphics[width=0.6cm]{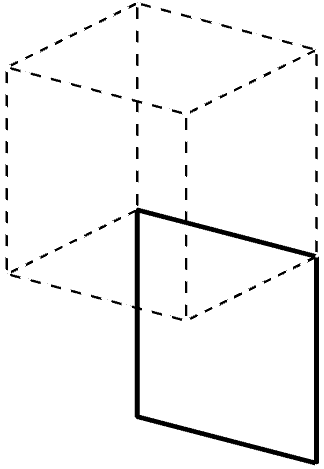}
       \end{tabular}\hspace{-0.2cm})-
  F(\includegraphics[width=0.6cm]{0+1+2+2s.png})&=0.
\end{aligned}
\end{equation}
Here we sum the values of $F$ at the faces of a particular cube (in the first equation) or the faces containing a particular edge (in the other equations), with appropriate signs. \remove{grammar}{}
We write one equation per cube and one per \gram{non-boundary} edge and impose no boundary conditions.

It's time for our new definition. Let $T$ be the function on the set of \red{non-boundary} faces given by
\begin{align*}
T(\includegraphics[width=0.6cm]{1+2s.png})
&=
\frac{1}{2}\left[
F(\includegraphics[width=0.6cm]{1+2s.png})\cdot F(\includegraphics[width=0.6cm]{1+2s.png})+
F(\hspace{-0.3cm}\begin{tabular}{c}\vspace{-0.1cm}
\includegraphics[width=0.6cm]{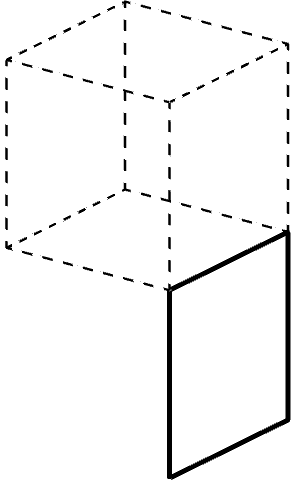}
\end{tabular}\hspace{-0.2cm})\cdot F(\includegraphics[width=0.6cm]{0+1+1+2s.png})+
F(\hspace{-0.2cm}\begin{tabular}{c}
\includegraphics[width=0.6cm]{_0+1+2+2s.png}\end{tabular}\hspace{-0.2cm})\cdot F(\includegraphics[width=0.6cm]{0+1+2+2s.png})\right]\red{;}
\\
T(\includegraphics[width=0.6cm]{0+2s.png})
&=
\frac{1}{2}\left[
F(\includegraphics[width=0.6cm]{1+2s.png})\cdot F(\includegraphics[width=0.6cm]{0+2s.png})+
F(\includegraphics[width=0.6cm]{0+0+1+2s.png})\cdot F(\includegraphics[width=0.6cm]{0+2s.png})\right]\red{;}
\\
T(\includegraphics[width=0.6cm]{0+1s.png})
&=
\frac{1}{2}\left[
F(\includegraphics[width=0.6cm]{1+2s.png})\cdot F(\includegraphics[width=0.6cm]{0+1s.png})+
F(\includegraphics[width=0.6cm]{0+0+1+2s.png})\cdot F(\includegraphics[width=0.6cm]{0+1s.png})\right]\red{.}\\[-1.0cm]
\end{align*}
For instance, the latter equality expresses the value of $T$ at a vertical face parallel to the $x$-axis through the values of $F$ at the same face and the two horizontal faces right behind it.
The value of $T$ at a horizontal (respectively, vertical) face discretizes energy density (respectively, flux). 
Proposition~\ref{th-Maxwell-approximation} below asserts that under a natural choice of $F$ we have uniform convergence as $N\to\infty$:
\begin{equation}\label{eq-quick-start-approximation}
T(\includegraphics[width=0.6cm]{1+2s.png})
      \rightrightarrows N^2\int\limits_{\includegraphics[width=0.6cm]{1+2s.png}}
      \tfrac{1}{2}(\vec{\mathrm{E}}^2+\vec {\mathrm{B}}^2)\,\mathrm{dA},\quad
      T(\includegraphics[width=0.6cm]{0+2s.png})
      \rightrightarrows -N^2\int\limits_{\includegraphics[width=0.6cm]{0+2s.png}}
      \vec{\mathrm{E}}\times\vec {\mathrm{B}}\,\mathrm{d}\vec{\mathrm{n}},\quad
      T(\includegraphics[width=0.6cm]{0+1s.png})
      \rightrightarrows N^2\int\limits_{\includegraphics[width=0.6cm]{0+1s.png}}
      \vec{\mathrm{E}}\times\vec {\mathrm{B}}\,\mathrm{d}\vec{\mathrm{n}}.
\end{equation}
\vspace{-0.4cm}

The desired discrete Poynting theorem (particular case of Theorem~\ref{cor-free} below) asserts that assuming only Maxwell's equations~\eqref{eq-Maxwell-1}, we have the following identity for each \gram{non-boundary} cube:
\begin{equation}\label{eq-discrete-Poynting}
T(\includegraphics[width=0.6cm]{0+0+1+2s.png})-
T(\includegraphics[width=0.6cm]{1+2s.png})-
T(\includegraphics[width=0.6cm]{0+1+1+2s.png})+
T(\includegraphics[width=0.6cm]{0+2s.png})+
T(\includegraphics[width=0.6cm]{0+1+2+2s.png})-
T(\includegraphics[width=0.6cm]{0+1s.png})=0.\\[-0.2cm]
\end{equation}
Properties~\eqref{eq-quick-start-approximation}--\eqref{eq-discrete-Poynting} are exactly what one requests from a discretization of energy density and flux according to the above discretization principles; it is nontrivial to satisfy both properties simultaneously. A proof \emph{in pictures} is in \S\ref{ssec-proofs-global}. 
And we proceed to a systematic discussion of discrete field theory.

\subsection{Background}

\remove{grammar}{}
Discrete field theory is actually at least as old as the continuum one. In 1847 G.~Kirchhoff stated the laws of an electrical network, which is 
the simplest model of the theory; see~\S\ref{ssec-Networks}.
In the continuum limit, the laws approximate the Laplace equation; thus the model perfectly serves for \gram{the} numerical solution of the latter. Remarkable approximation theorems were proved by L.~Lusternik \cite{Lusternik-26}, R.~Courant--K.~Friedrichs--H.~Lewy \cite{Courant-friedrichs-Lewy-28} in 1920s  and later generalized, e.g., in \edit{R12P1}{} \cite{Ciarlet-78,Chelkak-Smirnov-08,Wilson-08,Bobenko-Skopenkov-13,Werness-14}. Planar networks lead to the discretization of complex analysis having applications in statistical physics (e.g., obtained in \gram{the} 2010s by S.Smirnov\,et\,al.\,\cite{Chelkak-Smirnov-08}) and even computer graphics~\cite{Desbrun-etal-08}.

Discrete field theory was closely related to topology from the youth of both subjects. The Kirchhoff laws are naturally stated in terms of the \emph{boundary} and the \emph{coboundary} operators; see \S\ref{ssec-Networks} for an elementary introduction. 
Such formulation is usually attributed to H.~Weyl; 
see \cite[\S1F, p.~31]{Gross-Kotiuga-04} for an elaborate historical survey. In \gram{the} 1930s G.~de~Rham established \gram{a} correspondence between these operators and the exterior derivative and its dual;
see \cite{Arnold-etal-10} for a survey and \cite{Tonti-13} for general philosophy.
This lead to the above discrete Maxwell equations~\eqref{eq-Maxwell-1}; see also  \S\ref{ssec-Electrodynamics} and~
\cite{Bossavit-03,Gross-Kotiuga-04,Kron-44,Teixeira-13}.

The next major step was done by A.~Kolmogorov and J.~Alexander in \gram{the} 1930s, who invented a product discretizing the exterior product in a sense. Kolmogorov commented that such discretization was his original motivation. The construction was soon modified by H.~Whitney \cite{Whitney-38} and others to give the famous \red{\emph{cup product.}}\move{grammar}{} The original product was anticommutative, whereas the \red{cup product} was associative. One cannot get both properties simultaneously (this fact is crucial for rational homotopy theory). This reflects a general phenomenon that not all properties survive under discretization. We choose the associative \red{cup product} as a discretization of the exterior product, in contrast to~\edit{R12P1}{}\cite{Desbrun-etal-08, Wilson-08, Berbatov-22}. This requires \emph{vertices ordering} in 
\red{discrete field} theory,
a structure introduced for the first time.

Later there appeared discrete models for other classical fields: e.g.,  \emph{Feynman checkerboard} from \gram{the} 1940s and \emph{Regge calculus} from \red{the} 1960s for the Dirac and the gravitational field respectively; see \cite{Skopenkov-Ustinov-20} for an elementary introduction and survey of the former model.

In \gram{the} 1970s F.~Wegner and K.~Wilson introduced \emph{lattice gauge theory} as a computational tool for gauge theory describing all known interactions; 
see \cite{Maldacena-16} or \S\ref{ssec-gauge}
for an elementary introduction and \cite{Creuz-70} for details.
This culminated in determining the proton mass theoretically with an error $<2\%$. 

In \gram{the} 1980s A.~Connes developed a formalism, dealing (to some extent) uniformly with continuous and discrete 
geometries \cite{Connes}.
Using it, A.~Dimakis et al.~discretized the Yang-Mills equations \cite[Eq.~(4.15)]{Dimakis-etal-94}. Corollary~\ref{cor-gauge} extends their result by adding sources and the crucial unitarity constraint. Compare \gram{this} with the efforts put to achieve the gauge covariance in the remarkable survey \cite[\S9]{bender-etal-94}.

\edit{R12P1}{In the 1990s R.~Forman \cite{Forman-02} introduced a different discretization of differential forms, using cochain maps rather than cochains. 
S.~Wilson \cite{Wilson-08} and K.~Berbatov et al.~\cite{Berbatov-22} equipped them with the cup product (again, non-associative) to discretize Riemann surfaces and diffusion processes respectively. Although their setup is different, generalized Forman's forms appear in our energy conservation law.}

In \gram{the late} 1990s J.~Marsden et al.~discretized basic general theorems of field theory: the Euler--Lagrange equations and the Noether theorem on a $2$-dimensional grid; see \cite[Eq.~(5.2) and~(5.7)]{Mardsen-etal-98}, 
cf.~\cite[Eq.~(60) and~\red{(72)}]{Kraus-15}, 
\cite[Theorem~5.2.37]{Hydon-Mansfield-04},
\edit{R12P1}{\cite[Theorem~5.5]{Mansfield-etal-17}},
\cite[\red{Theorems~7.1 and~8.1 in Ch.~III}]{Dorodnitsyn-04}. %
These results \gram{extended} the ones obtained earlier for $1$-dimensional difference equations; see \cite{Hydon-Mansfield-04} for references. The discrete Euler--Lagrange equations in~\S\ref{ssec-statements} are straightforward generalizations of the known ones. \red{B}ut Discrete Noether Theorem~\ref{th-Noether} is different:
\edit{R12P1}{we construct a \emph{conserved current}
through edges just as in the Kirchhoff law, whereas
in previous works,
the current was defined on vertices \cite
{Dorodnitsyn-04},
pairs (triangle, its vertex) \cite
{Mardsen-etal-98}, and pairs (square, its vertex) \cite
{Kraus-15}. In
~\cite{Mansfield-etal-17}, the conservation law was stated in a global form, bypassing a construction of a current. This all led to rather technical statements of the conservation laws, 
to our knowledge, never applied to a particular field theory; see~\cite{Mansfield-etal-17} for a survey.} 
%
For the first time we use vertices ordering and \gram{\emph{cap product}}, making the statement, proof, and applications of the \red{discrete Noether} theorem particularly simple.

In \gram{the} 2010s M.~Kraus et al.~have stepped beyond the Lagrangian formulation \cite{Kraus-15}.
A discretization of hydrodynamics was introduced by E.~Gawlik et al. 
\cite[\S4]{Gawlik-Mullen-Pavlov-Marsden-Desbrun-10}. They derived general Euler--Poincare equations and Kelvin--Noether theorem \cite[\S3]{Gawlik-Mullen-Pavlov-Marsden-Desbrun-10}. Their approach was based on \red{the} discretization of the diffeomorphism group, thus \red{applied to a} rather specific class of models.
\remove{R12P1}{}

There was a folklore belief that no conserved discrete energy-momentum tensor exists in this framework. E.g., in 2016 D.~Chelkak, A.~Glazman, and S.~Smirnov introduced a ``halfway'' conserved tensor \cite[Corollary~2.12(1)]{Chelkak-Smirnov-16}; cf.~\cite{Suzuki-16}.
Even the notion of a rank~2 symmetric tensor itself is hard to discretize \cite[\S7]{Arnold-etal-10}. But in \gram{the} 2000s V.~Dorodnitsyn \red{discretized} energy and momentum conservation in some particular cases \cite[Example in \S8 \red{of Ch.~III}]{Dorodnitsyn-04}. \edit{R1P7}{His construction, like the other known ones, was based on moving the points of a $2$-dimensional lattice.} \edit{R12P1}{We extend 
it using a new approach not relying on any continuous motion or symmetry. As a result,
in 
Theorem~\ref{cor-free}} we construct an exactly conserved discrete energy-momentum tensor, approximating the continuum one at least for free fields.


\gram{The} great success of discrete models forces \red{us} to search for a general discretization method and even to build the whole field theory starting from discrete rather than continuous space and time \cite{bender-etal-94}.
\vspace{-0.2cm}

\subsection{Main idea}
\label{ssec-main-tools}

We propose the following discretization algorithm for field theories:
\begin{enumerate}
  \item
%
Take a continuum Lagrangian written in terms of exterior calculus operations from Table~\ref{tab-translation}.
\item
Replace the exterior calculus operations \gram{with} cochain operations using Table~\ref{tab-translation} \emph{literally}.
\item
Get equations of motions/conservation laws from discrete Euler--Lagrange/Noether theorems.
\end{enumerate}
%
This idea is well-known but \gram{the} realization is new. In 
\clarity{Tables~\ref{tab-translation}--\ref{tab-discretization},}
in contrast to the rest of the paper, we assume familiarity with the basics of exterior calculus and continuum field theory.

\edit{R12P2}{}\gram{}
\begin{table}[!t]
  \caption{Correspondence between continuum and discrete notions}
  \label{tab-translation}
  \hspace{-0.3cm}
  \begin{tabular}{lllll}
    \toprule
    \multicolumn{2}{c}{\textbf{Continuum}}&
    \multicolumn{2}{c}{\textbf{Discrete}}&
    \hspace{-0.5cm}\textbf{Definition}\\
    \midrule
    \multicolumn{4}{c}{\emph{Algorithmic part I. Replacement in Lagrangian and action}:}& \\[0.3cm]
    differentiable manifold (spacetime) & $\mathrm{M}$ &
    simplicial or cubical complex & ${M}$ & \ref{def-cochain},
    \\[-0.1cm]
    & & with a fixed vertices ordering & & \ref{def-spacetime} 
    \\
    $k$-form, $\mathbb{R}$- or $\mathbb{C}^{m\times n}$-valued & $\phiup$ & $k$-cochain, $\mathbb{R}$- or $\mathbb{C}^{m\times n}$-valued  & $\phi$ & \ref{def-cochain} \\
    exterior derivative & $\mathrm{d}$ &
    coboundary & $\delta$ & \ref{def-coboundary} \\
    exterior product & $\wedge$ &
    \red{cup product} & $\smile$ & \ref{def-covar-coboundary-gauge} 
    \\
    interior product & $\lrcorner$ &
    \red{cap product} & $\frown$ & \ref{def-covar-coboundary-gauge} 
    \\
    connection\,1-form,Lie-algebra-valued & $\mathrm{A}$ &
    connection,not\,Lie-algebra-valued & $A$ & \ref{def-connection} \\
    curvature 2-form, Lie-algebra-valued & $\mathrm{F}$ &
    curvature, not Lie-algebra-valued & $F$ & \ref{def-curvature} \\
    covariant exterior derivative & $\mathrm{D}_\mathrm{A}$ &
    covariant coboundary & $D_A$ & 
    \ref{def-covar-coboundary-gauge} \\
    raising all indices & $\sharp$ &
    sharp-operator (new notion) & $\#$ & \ref{def-sharp} 
    \\
    function on $\mathbb{R}$ or $\mathbb{C}^{m\times n}$ (e.g., $\log$ or $\mathrm{Tr}$) & $\mathrm{f}$ &
    the same function on $\mathbb{R}$ or $\mathbb{C}^{m\times n}$ & $\mathrm{f}$ & --- \\
    spacetime integration of a 0-form & $\int_{\mathrm{M}}\!\mathrm{dV}\cdot$ & sum of the values of a 0-chain & $\epsilon$ & \ref{def-epsilon} \\[0.3cm]
    \multicolumn{4}{c}{\emph{Informal part II. Correspondence in equations of motion and conservation laws}:}& \\[0.3cm]
    codifferential, $\sharp$-conjugated & $\sharp\deltaup\flat$ &
    boundary & $\partial$ & \ref{def-boundary} \\
    covariant\,codifferential,$\sharp$-conjugated & \hspace{-0.1cm} $\sharp \mathrm{D}^*\!\!\!_\mathrm{A}\flat\!$ &
    covariant boundary & $D^*_A$ & 
    \ref{def-covar-boundary-gauge} \\
    interior product & $\llcorner$ & \red{cop product} (new notion) & $\spleen$ & \ref{def-cap-general} \\
    tensor product over $C^\infty(M)$ & $\otimes$ & chain-cochain \red{cross product} & $\times$ & \ref{def-tensor} \\
    type $(1,1)$ tensor & $\mathrm{T}$ & type $(1,1)$ tensor (new notion) & $T$ & \ref{def-tensor} \\
    integration of its $k$-th component & $\int_{\red{\mathrm{h}}}\!\mathrm{T}_          k$ &
    flux (new notion) & \hspace{-0.2cm}$\langle T,\red{h}\rangle_k\hspace{-0.2cm}$ & \ref{def-flow}
    \\[0.1cm]
    integration of a $k$-form & $\int_{\red{\mathrm{h}}}\!\phiup$ &
    pairing & \hspace{-0.2cm}$\langle \phi,\red{h}\rangle\hspace{-0.2cm}$ & \ref{def-pairing} \\[0.1cm]
    \bottomrule
  \end{tabular}
\end{table}

\move{clarity}{} We stress that Part I of Table~\ref{tab-translation} gives an \emph{algorithm}, not just an analogy (as Part II).
The algorithm provides conservation laws only for symmetries \gram{that} are preserved by the discretization. Thus we usually guarantee charge conservation (based on the automatically preserved gauge symmetry) and energy-momentum conservation (not based on any symmetry in our setup).

Results of applying the algorithm to basic field theories are \edit{R12P3}{summarized in Table~\ref{tab-discretization} and}
discussed in \S\ref{sec-examples}. The output discrete theories are usually simpler than the input continuum ones; knowledge of the latter is not required for understanding the former. All the output theories of~\S\ref{sec-examples} are known, but some obtained conservation laws are new. As a tool, we use discrete covariant differentiation \edit{R2P1}{} (see \S\ref{ssec-general-connections} and~\cite{Dimakis-etal-94}) and build a new discretization of tensor calculus involving non-antisymmetric tensors (see \S\ref{ssec-statements}). This is done in terms of cochain operations from Table~\ref{tab-translation}, which appear naturally in examples. 
\edit{R12P3}{A reader looking for a zero-knowledge introduction can now proceed directly to~\S\ref{sec-examples}.}

\begin{landscape}
\begin{table}[hbt]
  \caption{\red{Discretization of basic field theories}}
  \label{tab-discretization}
  \begin{tabular}{lccccccc}
    \toprule
    &\multicolumn{2}{c}{\textbf{Continuum}}&\multicolumn{4}{c}{\textbf{Discrete}}&\\
    \cmidrule(rl){2-3}                     \cmidrule(rl){4-7}
\textbf{Field theory}
    &Field &Lagrangian &Lagrangian &Equation 
    &Conserved &Energy-momentum &\textbf{Reference}\\
    &      &           &           &of motion 
    &current   &tensor  &\\
   \midrule
Electric network & $\mathbb{R}$-valued 
& $\tfrac{1}{2}\mathrm{d}\phiup\lrcorner \mathrm{d}\phiup - \mathrm{s}\lrcorner\phiup$
& $\tfrac{1}{2}\delta\phi\frown \delta\phi - s\frown \phi$
& $\partial\delta\phi=s$
& $\delta\phi$
& $\delta\phi\times\delta\phi$
& \S\ref{ssec-Networks}
\\
& $0$-form $\phiup$
&
&
&
&
&
&
\\[0.3cm]
Electrodynamics  &  $\mathbb{R}$-valued 
& $-\tfrac{1}{2}\sharp\mathrm{dA}\lrcorner \mathrm{dA} - j\lrcorner\mathrm{A}$
& $-\tfrac{1}{2}\#\delta A\frown \delta A - j \frown A$
& $-\partial\#\delta A=j$
& $j$
& $-\#\delta A\times\delta A$
& \S\ref{ssec-Electrodynamics}
\\
&  $1$-form $\mathrm{A}$
&
&
&
&
&
&
\\[0.3cm]
Gauge theory  &  connection 
& $-\mathrm{Re}\mathrm{Tr}[\tfrac{1}{2}\sharp\mathrm{F}^*\lrcorner \mathrm{F} + j\lrcorner\mathrm{A}]$,
& $-\mathrm{Re}\mathrm{Tr}[\tfrac{1}{2}\# F^*\frown F + j \frown A]$,
& $\mathrm{Pr}_{T_U G}\,D^*_A\# F$
& $j$
& $-\mathrm{Re}\mathrm{Tr}[\# F\times F]$
& \S\ref{ssec-gauge}
\\
             &  $1$-form $\mathrm{A}$ 
& $F=\mathrm{dA}+\mathrm{A}\wedge\mathrm{A}$
& $F=\delta A+ A\smile A$
& $=-\mathrm{Pr}_{T_U G}\,j$
&
&
&
\\[0.3cm]
Klein--Gordon & $\mathbb{C}$-valued 
& $\sharp\mathrm{d}\phiup\lrcorner \mathrm{d}\phiup^* - m^2\phiup\lrcorner\phiup^*$
& $\#\delta\phi\frown \delta\phi^* - m^2\phi\frown \phi^*$
& $\partial\#\delta\phi=m^2\phi$
& $-2\mathrm{Im}[\#\delta\phi^*\frown\phi]$
& $2\mathrm{Re}[\#\delta\phi^*\times\delta\phi$
& \ifarxiv{\S\ref{ssec-Klein-Gordon}}{\cite[\S A.2]{Skopenkov-preprint}}
\\
& $0$-form $\phiup$
&
&
&
&
& $-m^2\phi^*\times\phi]$
&
\\[0.3cm]
Klein--Gordon & $\mathbb{C}^{1\times n}$-valued 
& $\sharp\mathrm{D_A}\phiup\lrcorner (\mathrm{D_A}\phiup)^*$
& $\#D_A\phi\frown (D_A\phi)^*$
& $D^*_A\#D_A\phi$
& $-2\phi^*\smile \#D_A\phi$
& unknown
& \ifarxiv{\S\ref{ssec-Klein-Gordon}}{\cite[\S A.2]{Skopenkov-preprint}}
\\
in a gauge field & $0$-form $\phiup$
&  $- m^2\phiup\lrcorner\phiup^*$ 
&  $- m^2\phi\frown \phi^*$ 
&  $=m^2\phi$
&
&
&
\\
    \bottomrule
  \end{tabular}
\end{table}
\end{landscape}

\begin{remark}
  In Table~\ref{tab-translation} we intentionally include no discretization for the Hodge star or products other than exterior, interior, 
  tensor products. In all the examples, we have succeeded to avoid them. 

  Continuum and discrete notations fit not that well. But both are commonly used in their contexts (except \gram{for} a few new discrete objects, for which we keep the continuum notation in a different~font).

  \move{R12P3}{} Putting a continuum Lagrangian to the required input form is not always possible and can be ambiguous:
  \edit{R12P3}{For instance, in Table~\ref{tab-discretization}, electrodynamics can be also viewed as a gauge theory with a $\mathrm{u}(1)$-valued connection. This leads to the same continuum Lagrangian but different discretizations.}
\end{remark}

\subsection{Statements} \label{ssec-statements}

Let us state the main new results precisely in their simplest form. \edit{R12P3}{This subsection is a technical summary. 
The introduced notions are all motivated in~\S\ref{sec-examples}, where they appear little by little in examples.} Further generalizations are postponed until~\S\ref{sec-general}.

\begin{definition} \label{def-cochain}
\move{R12P2}{} Dissect the hypercube $0\le x_0,x_1\dots, x_{d-1}\le N$ in $\mathbb{R}^d$ into $N^d$ 
unit hypercubes\edit{R12P2}{; see Figure~\ref{fig-grid9}. By \emph{$k$-dimensional faces} we mean the $k$-dimensional faces of those unit hypercubes. The collection of all those faces is called the \emph{$d$-dimensional grid} $I^d_N$. In what follows we denote $M=I^d_N$, unless the values of $d$ and $N$ need to be shown explicitly (this is convenient for generalizations).}

A \emph{$k$-dimensional field} or \emph{$k$-cochain} \clarity{or \emph{function on $k$-dimensional faces}} is a real-valued function defined on the set of $k$-dimensional faces of~${M}$. 
Denote by $C^k({M};\mathbb{R})=C_k({M};\mathbb{R})$ the set of all $k$-dimensional fields; see~Remark~\ref{rem-compare-definition} for comparison with the other definitions in literature.

A \emph{Lagrangian} is a function $\mathcal{L}\colon C^k({M};\mathbb{R})\to C_0({M};\mathbb{R})$.
The \emph{action} $\mathcal{S}\colon C^k({M};\mathbb{R})\to \mathbb{R}$ is the sum of the values of the Lagrangian over all the vertices. A field $\phi\in C^k({M};\mathbb{R})$ is
\edit{R1P4}{an \emph{extremal} or a \emph{critical point} or}
\emph{stationary} for the action functional,
\red{if}
$\left.\frac{\partial}
  {\partial t}\mathcal{S}[\phi+ t\Delta]\right|_{ t=0}=0$ for each $\Delta\in C^k({M};\mathbb{R})$.
\end{definition}

\begin{figure}[htb]
  \centering
  \vspace{-0.4cm}
  \includegraphics[width=3.5cm]{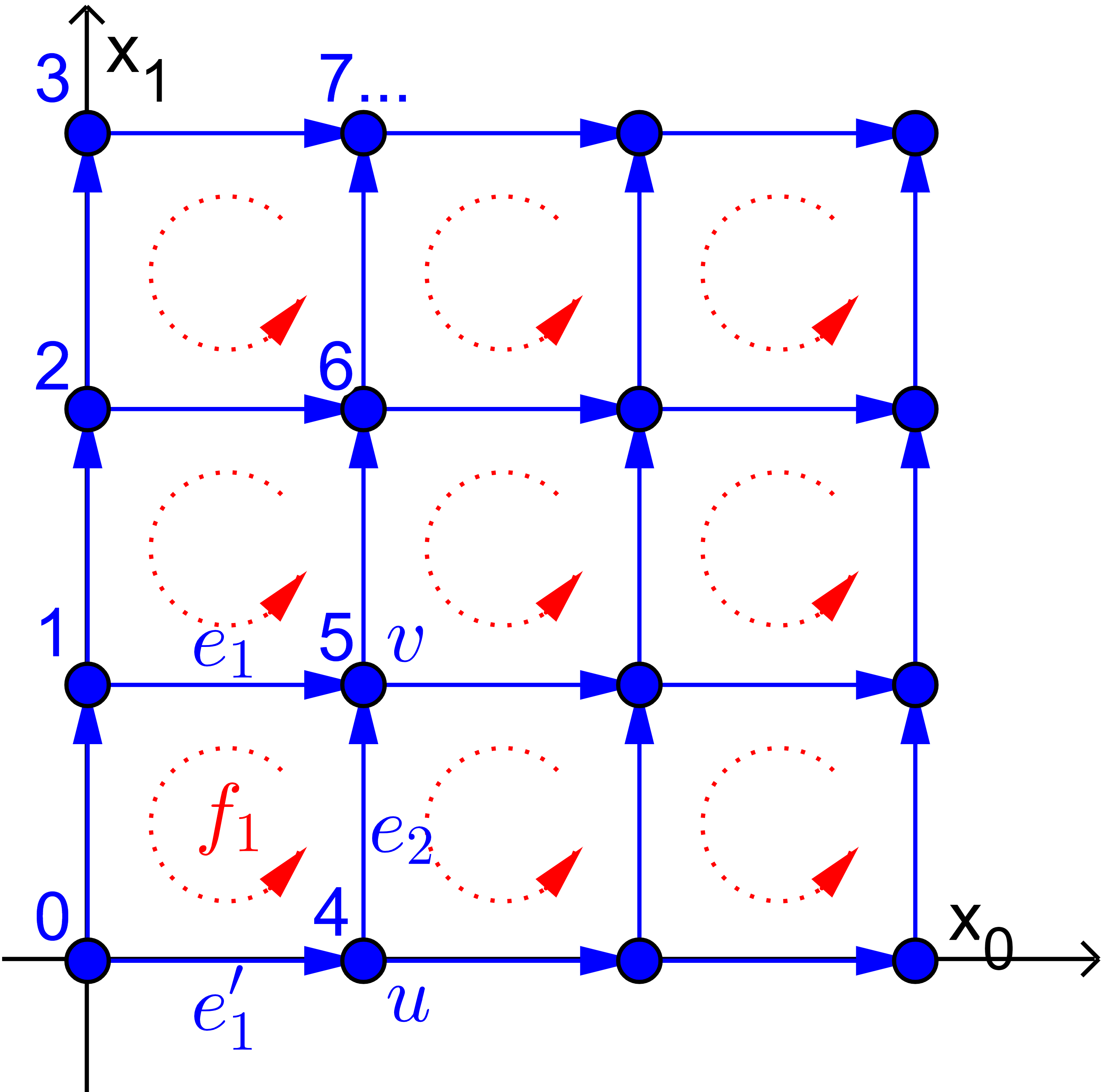}
  \vspace{-0.2cm}
  \caption{\red{The} $3\times 3$ grid \red{$I^2_3$} with the dictionary \red{order} of vertices; the vertices are enumerated in increasing order. \red{Orientation of $1$- and $2$-dimensional faces.}  The $1$-dimensional ($e_1$ and $e_2$) and $2$-dimensional ($f_1$) faces with the maximal vertex~$v$. \red{The $1$-dimensional face $e'_1$ of $f_1$ containing the vertex $u$ and not containing the vertex $v$.}
  }\label{fig-grid9}
\end{figure}
\clarity{}

\edit{R12P3}{Discrete field theory studies extremals of action functionals. Now we introduce additional structure: the \emph{dictionary order of vertices}, the natural \emph{orientation of faces}, the \emph{boundary} and \emph{coboundary} operators. The latter is an analog of the (exterior) derivative.}

\begin{definition} \label{def-boundary} \label{def-coboundary}
\move{R12P3}{} Fix the \emph{dictionary order} of the vertices \red{of the grid $I^d_N$}: set $(x_0,x_1\dots, x_{d-1})<(y_0,y_1\dots, y_{d-1})$ if and only if $x_0=y_0$, \dots, $x_{k-1}=y_{k-1}$, and $x_{k}<y_{k}$ for some $0\le k\le d-1$.

Denote by $\max f$ ($\min f$) the maximal (minimal) vertex  of a face $f$ \edit{R12P2}{of $M$}. (On the grid, it is the vertex with the maximal (minimal) sum of the coordinates).

Fix the following orientation of $k$-dimensional faces of \edit{R12P2}{$M$}. A \emph{positively oriented} basis in a face is formed by the $k$ vectors starting at the minimal vertex of the face, going along the edges of the face, and listed in the order opposite to the order of the endpoints. E.g., 
a positively oriented basis in a $d$-dimensional face of $I^d_N$ is $(1/N,0,\dots,0),(0,1/N,\dots,0),\dots,(0,0,\dots,1/N)$, because $(1/N,0,\dots,0)>(0,1/N,\dots,0)>\dots>(0,0,\dots,1/N)$.
A \clarity{$(k+1)$}-dimensional face $f$ and a \red{$k$}-dimensional face $e\subset f$ are \emph{cooriented} (respectively, \emph{opposite oriented}), if
the ordered set consisting of the outer normal to $e$ in $f$ and a positive basis in $e$ is a positive (respectively, negative) basis in $f$.
\remove{clarity}{}

The \emph{boundary} $\partial \red{\phi}$ and the \emph{coboundary} $\delta \red{\phi}$ of a function $\red{\phi}$ on $k$-dimensional faces $e$ are the functions on  $(k-1)$- and $(k+1)$-dimensional faces $v$ and $f$ respectively given by (see Figures~\ref{fig-3D} and~\ref{fig-pipeline2})
\begin{align*}
[\partial \red{\phi}](v)&=\sum_{e\supset v\textrm{ cooriented with }v}\red{\phi}(e)-\sum_{e\supset v\textrm{ oriented opposite to }v}\red{\phi}(e),\\
[\delta \red{\phi}](f)&=\sum_{e\subset f\textrm{ cooriented with } f}\red{\phi}(e)-\sum_{e\subset f\textrm{ oriented opposite to } f}\red{\phi}(e).
\end{align*}
Hereafter an empty sum is set to be $0$, 
and $C^k(\red{M};\mathbb{R}):=\{0\}$ for $k<0$ or $k>\red{\dim M}$. \move{R12P3}{}
\end{definition}

\begin{figure}[htbp]
\begin{tabular}{cl}
\begin{tabular}{l}
\includegraphics[width=2cm]{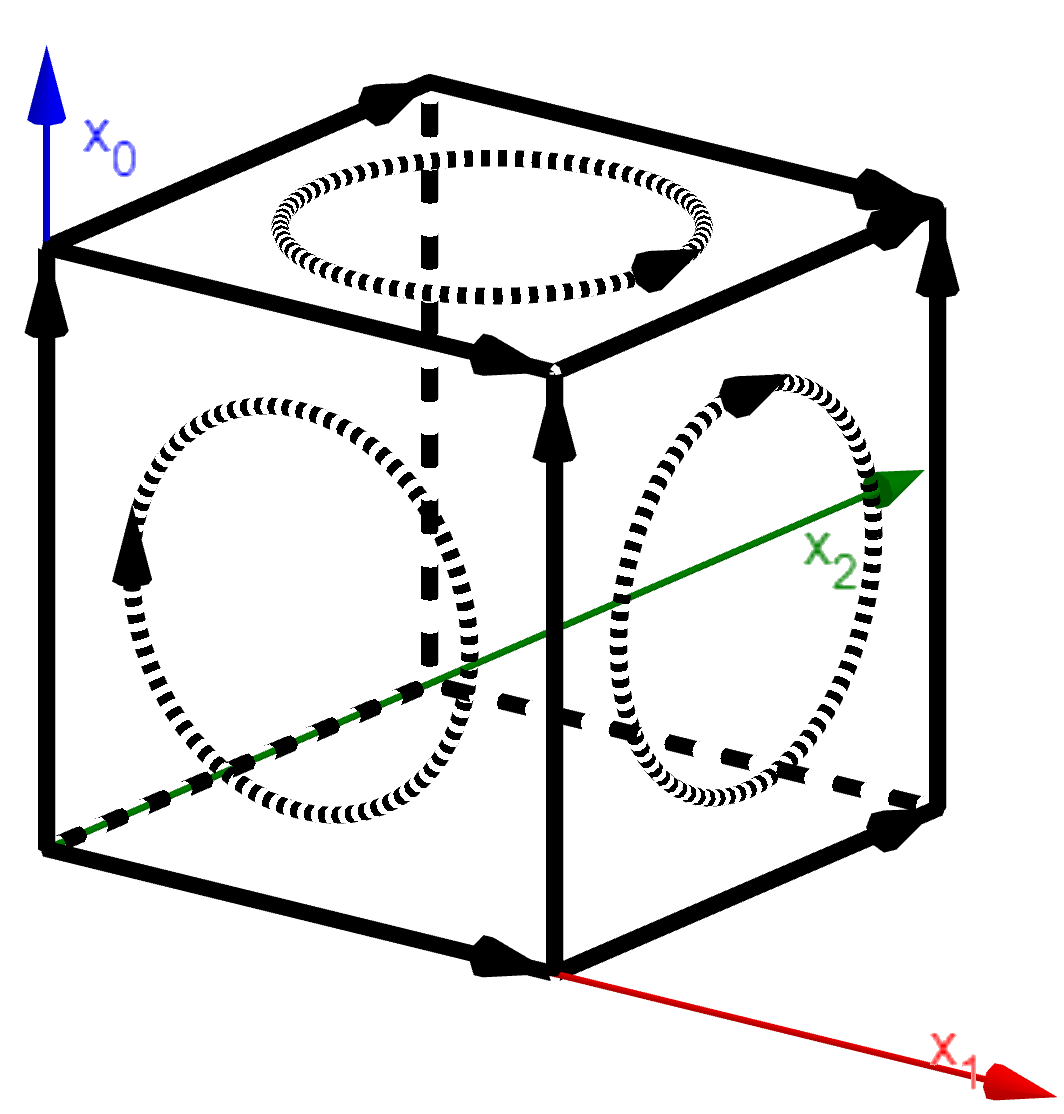}
\end{tabular}
&
\begin{tabular}{l}
$[\partial \red{\phi}](\includegraphics[width=0.6cm]{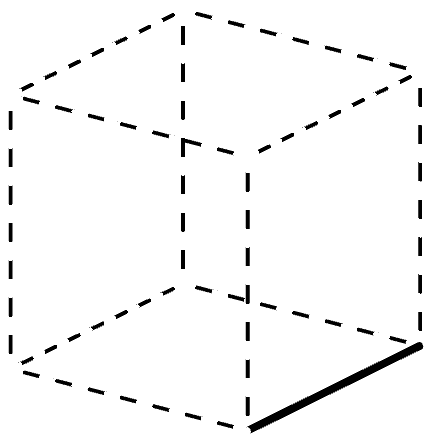})
=
\red{\phi}(\includegraphics[width=0.6cm]{1+2s.png})-
\red{\phi}(\hspace{-0.3cm}\begin{tabular}{c}\vspace{-0.0cm}
       \includegraphics[width=0.9cm]{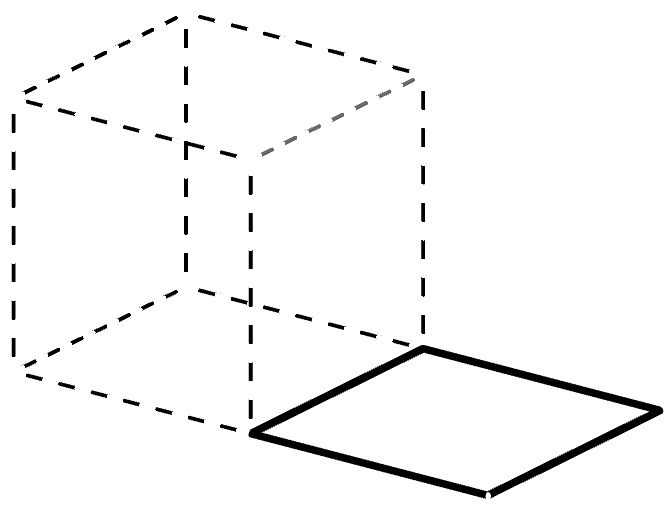}
       \end{tabular}\hspace{-0.2cm})+
\red{\phi}(\hspace{-0.3cm}\begin{tabular}{c}\vspace{-0.2cm}
       \includegraphics[width=0.6cm]{_0+1+1+2s.png}
       \end{tabular}\hspace{-0.2cm})-
\red{\phi}(\includegraphics[width=0.6cm]{0+1+1+2s.png})$
\\[0.4cm]
$[\delta \red{\phi}](\includegraphics[width=0.6cm]{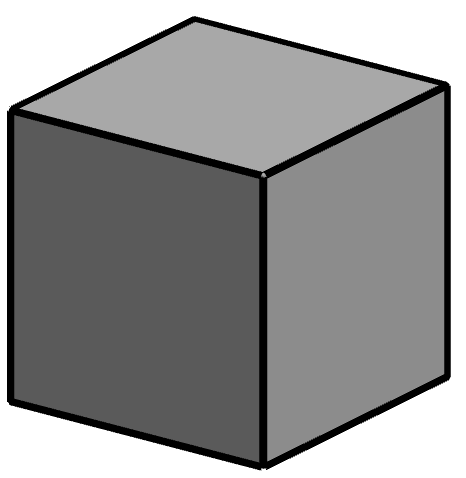})
=
\red{\phi}(\includegraphics[width=0.6cm]{0+0+1+2s.png})-
\red{\phi}(\includegraphics[width=0.6cm]{1+2s.png})-
\red{\phi}(\includegraphics[width=0.6cm]{0+1+1+2s.png})+
\red{\phi}(\includegraphics[width=0.6cm]{0+2s.png})+
\red{\phi}(\includegraphics[width=0.6cm]{0+1+2+2s.png}) -
\red{\phi}(\includegraphics[width=0.6cm]{0+1s.png})$
\end{tabular}
\end{tabular}
\caption{Boundary and coboundary (see Definition~\ref{def-boundary}). A \red{non-boundary 3-dimesional} face (to the left) is shown again by \red{dashed} lines (to the right). The face \red{at} which a particular function is evaluated is in bold. The signs in the expression for $\partial \red{\phi}$ are different from the ones in~\eqref{eq-Maxwell-1} because the latter depicts \red{the different equation $\partial \# F=0$ introduced in~\S\ref{ssec-Electrodynamics}.}}
\label{fig-3D}
\end{figure}
\varedit{clarity}{}{-2.5}

\begin{remark} \label{rem-form}
\edit{R2P2}{(For specialists)
The values of a $k$-cochain at the $k$-dimensional faces
with a common maximal vertex $v$ discretize the components of a $k$-form $\sum_{i_1<\dots<i_k} \phiup_{i_1,\dots,i_k}\,\mathrm{dx}_{i_1}\wedge\dots\wedge \mathrm{dx}_{i_k}$ at the point~$v$.
The special choice of orientations makes the coboundary consistent with the exterior derivative. The coboundary treats different spatial directions differently just like the exterior derivative.
}
\end{remark}

Informally, a Lagrangian is \emph{local} \clarity{or \emph{first-order}}, if its value at a vertex depends only on the values of
the field $\phi$ and the coboundary $\delta\phi$ at the faces for which the vertex is maximal. 
\emph{Partial derivatives} 
with respect to $\phi$ and $\delta\phi$ are fields of dimension $k$ and $k+1$ respectively, obtained by differentiating the Lagrangian as if $\phi$ and $\delta\phi$ were independent variables.
The precise definition is as follows. \remove{clarity}{}

\begin{definition} \label{def-local-particular}
A Lagrangian $\mathcal{L}\colon C^k({M};\mathbb{R})\to C_0({M};\mathbb{R})$ is \emph{local}, if for each vertex~$v\in M$ there is a smooth function $L_v(\phi_1,\dots,\phi_{p},\phi'_1,\dots,\phi'_{q})\in \mathrm{C}^1(\mathbb{R}^{p+q})$
such that for each $\phi \in C^k({M};\mathbb{R})$ we have
\begin{equation}\label{eq-local}
\mathcal{L}[\phi](v)=L_v(
\phi(e_1),\dots,\phi(e_p),[\delta\phi](f_1),\dots,[\delta\phi](f_q)),
\end{equation}
where $e_1,\dots,e_p$ and $f_1,\dots,f_q$ are
all the faces of dimension $k$ and $k+1$ respectively
with the maximal vertex $v$; see Figure~\ref{fig-grid9}.
Define
$$
\frac{\partial \mathcal{L}[\phi]}{\partial\phi}\in C_k({M};\mathbb{R})\qquad\text{and}\qquad \frac{\partial \mathcal{L}[\phi]}{\partial(\delta\phi)}\in C_{k+1}({M};\mathbb{R})
$$
by the following formulae for each $l=1,\dots,p$ and $m=1,\dots,q$:
\begin{align}\label{eq-Lagrangian-derivatives}
\frac{\partial \mathcal{L}[\phi]}{\partial\phi}({e}_{l})&:=
\frac{\partial L_v }{\partial\phi_l}(\phi(e_1),\dots,\phi(e_p),[\delta\phi](f_1),\dots,[\delta\phi](f_q)),\\
\label{eq-Lagrangian-derivatives2}
\frac{\partial \mathcal{L}[\phi]}{\partial(\delta\phi)}({f}_{m})&:=
\frac{\partial L_v }{\partial\phi'_{m}}(\phi(e_1),\dots,\phi(e_p),[\delta\phi](f_1),\dots,[\delta\phi](f_q)).
\end{align}
\end{definition}

The following theorem is a straightforward \gram{generalization} of known ones; cf.~\cite[Eq. (5.2)]{Mardsen-etal-98}.

\begin{theorem}[The discrete Euler--Lagrange equations]\label{th-Euler-Lagrange} 
  Let $\mathcal{L}\colon C^k({M};\mathbb{R})\to C_0({M};\mathbb{R})$ be a local Lagrangian. Then a field $\phi\in C^k({M};\mathbb{R})$ is \edit{R1P4}{an extremal}, if and only if the following equation holds:
  \begin{equation}\label{eq-Euler-Lagrange}
    \partial\frac{\partial \mathcal{L}[\phi]}{\mathrm{\partial}(\delta\phi)}+\frac{\partial \mathcal{L}[\phi]}{\mathrm{\partial}\phi}=0.
  \end{equation}
\end{theorem}

(Here a plus sign stands because the boundary operator $\partial$ for $k=0$ discretizes \emph{minus} divergence.)

\begin{figure}[htb]
  \centering
  \begin{tabular}[b]{cc}
  \definecolor{qqqqff}{rgb}{0.,0.,1.}
\begin{tikzpicture}[line cap=round,line join=round,>=triangle 45,x=2.0cm,y=2.0cm]
\clip(-0.36642341680689944,-0.29922777559595914) rectangle (1.940432744999426,0.29150580906597495);
\draw [->,color=qqqqff] (0.,0.) -- (0.5,0.);
\draw [->,color=qqqqff] (0.5,0.) -- (1.,0.);
\draw [->,color=qqqqff] (1.,0.) -- (1.5,0.);
\begin{scriptsize}
\draw [fill=qqqqff] (0.,0.) circle (2.5pt);
\draw[color=qqqqff] (-0.009023166386201077,-0.13416986223453645) node {0};
\draw [fill=qqqqff] (0.5,0.) circle (2.5pt);
\draw[color=qqqqff] (0.47292565615080123,-0.12258685077057696) node {1};
\draw[color=qqqqff] (0.2617345990840249,0.07432434411673444) node {$\mathbf{1}$};
\draw [fill=qqqqff] (1.,0.) circle (2.5pt);
\draw[color=qqqqff] (0.9819502552348259,-0.1167953450385972) node {2};
\draw[color=qqqqff] (0.7653440428586451,0.07432434411673444) node {$\mathbf{2}$};
\draw [fill=qqqqff] (1.5,0.) circle (2.5pt);
\draw[color=qqqqff] (1.490974854318851,-0.12258685077057696) node {3};
\draw[color=qqqqff] (1.263538331323861,0.07432434411673444) node {$\mathbf{3}$};
\end{scriptsize}
\end{tikzpicture} 
  &
  \begin{tabular}[b]{cc}
  \red{$[\partial \psi](1)=\psi(\mathbf{1})-\psi(\mathbf{2})$;} &
  \red{$\,[\psi\frown \phi](\mathbf{1})=\psi(\mathbf{1})\phi(0)$;}
  \\
  \red{$\,[\delta \phi](\mathbf{1})=\phi(1)-\phi(0)$;} &
  \red{$[\psi\smile \phi](\mathbf{1})=\psi(\mathbf{1})\phi(1)$.}
  \end{tabular}
  \end{tabular}
  \caption{\red{Cochain operations on the $1$-dimensional grid $I^1_3$. Here $\phi$  and $\psi$ are functions on vertices and edges respectively. Bold font is used for edge numbers.} Cf.~Definitions~\ref{def-coboundary}, \ref{def-cap-grid}, \ref{def-cap-general}.
  }
  \label{fig-pipeline2}
\end{figure}
\varedit{R12P3}{}{-2}


The Noether theorem gives a \emph{conserved current} for each continuous symmetry of the Lagrangian. \edit{R12P3}{It is most nicely stated in terms of the \emph{cap product}
relying on the ordering of the vertices; see Figure~\ref{fig-pipeline2}.}

\begin{definition} \label{def-cap-grid}
\move{clarity}{}A \emph{current} is \red{an arbitrary function on edges}. A current \red{$j$} is \emph{conserved}, if $\partial j=0$.

\edit{R12P3}{The (particular case of) \emph{cap product} $\psi\frown\phi$ of functions $\psi$ and $\phi$ on $(k+1)$- and $k$-dimensional faces respectively is the function on edges 
$uv$
given by
$$
[\psi\frown \phi](uv)=(-1)^k\sum_{i=1}^{q}\lambda_i\psi(f_i)\phi(e'_i),
$$
where $u<v$;
$f_1,\dots,f_q$ are all the ($k+1$)-dimensional faces containing $uv$ and having the maximal vertex $v$;
$e'_i$ is the $k$-dimensional face of $f_i$ containing $u$ and not containing $v$ (see Figure~\ref{fig-grid9});
$$
\lambda_i:=\begin{cases}
+1, &\text{if $f_i$ and $e'_i$ are cooriented;}\\
-1, &\text{if $f_i$ and $e'_i$ are opposite oriented.}
\end{cases}
$$
}
\end{definition}

\begin{theorem}[Discrete Noether theorem] 
\label{th-Noether}
    Let $\mathcal{L}\colon C^k({M};\mathbb{R})\to C_0({M};\mathbb{R})$ be a local Lagrangian and
    $\phi\in C^k({M};\mathbb{R})$ be \edit{R1P4}{an extremal}. The Lagrangian is \emph{invariant under an infinitesimal transformation} $\Delta\in C^k({M};\mathbb{R})$, i.e.,
  \begin{equation}\label{eq-invariance}
  \left.\frac{\partial}
  {\partial t}\mathcal{L}[\phi+ t\Delta]\right|_{ t=0}=0,
  \end{equation}
  if and only if the following current is conserved: 
  \begin{equation}\label{eq-Noether-current}
  j[\phi]=\frac{\partial\mathcal{L}[\phi]}{\partial(\delta\phi)}\frown \Delta.
  \end{equation}
\end{theorem}

This theorem is different from known discretizations of the Noether theorem in~\cite{Dorodnitsyn-04,Hydon-Mansfield-04,Kraus-15,Mardsen-etal-98}.

Discrete spacetime has no continuous symmetries, but there is still a corresponding conserved tensor. \emph{Conserved tensors} are functions on faces of the \emph{Cartesian square} ${M}\times {M}$ rather than of~$M$ itself; \clarity{see Figure~\ref{fig-notation-cross}}. We shall see that such functions appear naturally in examples in~\S\ref{sec-examples}. 


\varedit{R12P3\\moved text}{}{3.6}
\begin{figure}[htb]
    \begin{tabular}{cl}
    \begin{tabular}{c}
  \includegraphics[height=2.51cm]{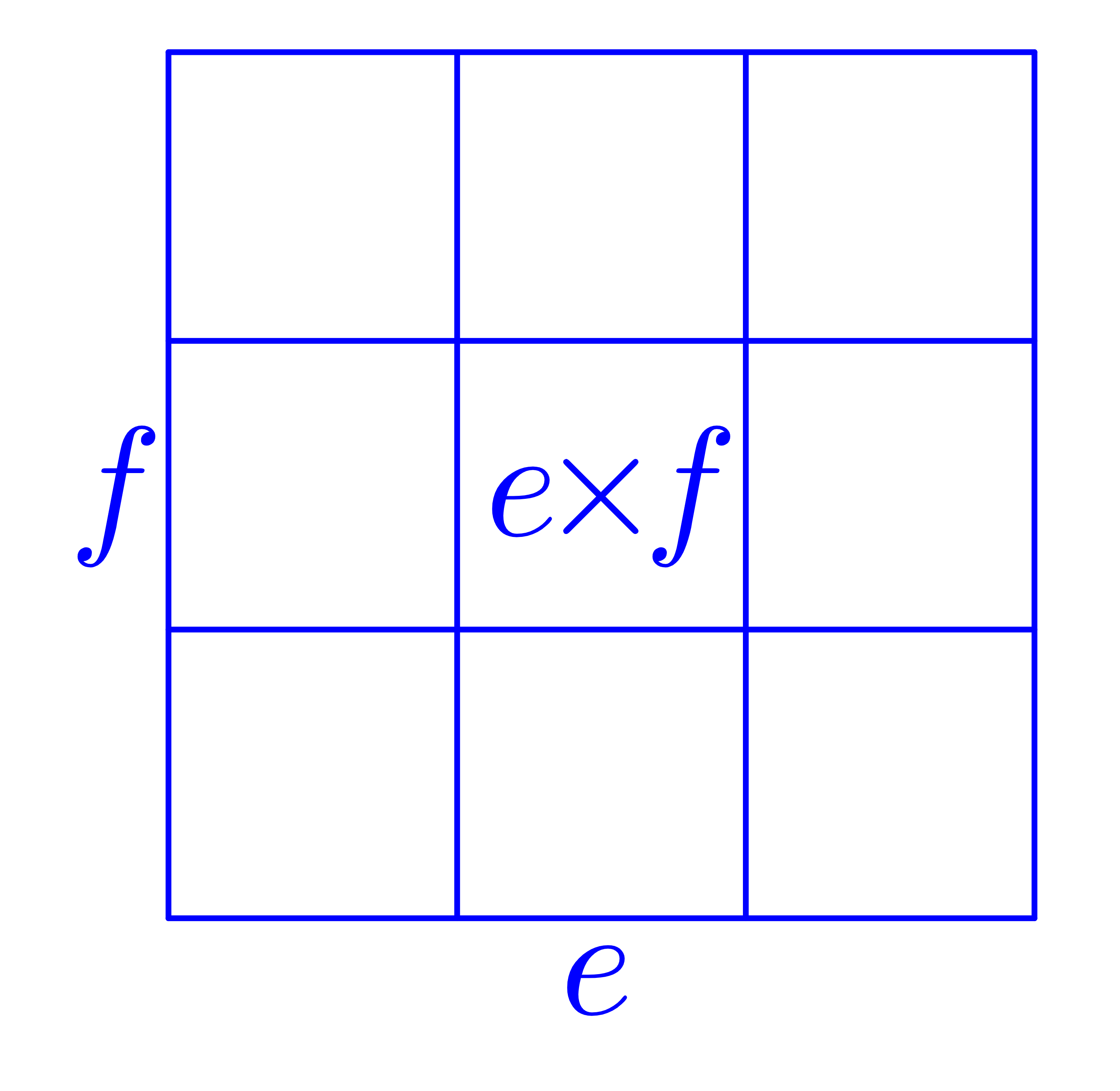}
  \end{tabular}
  &
  \begin{tabular}{l}
  \red{ $[\partial T](\includegraphics[width=1.0cm]{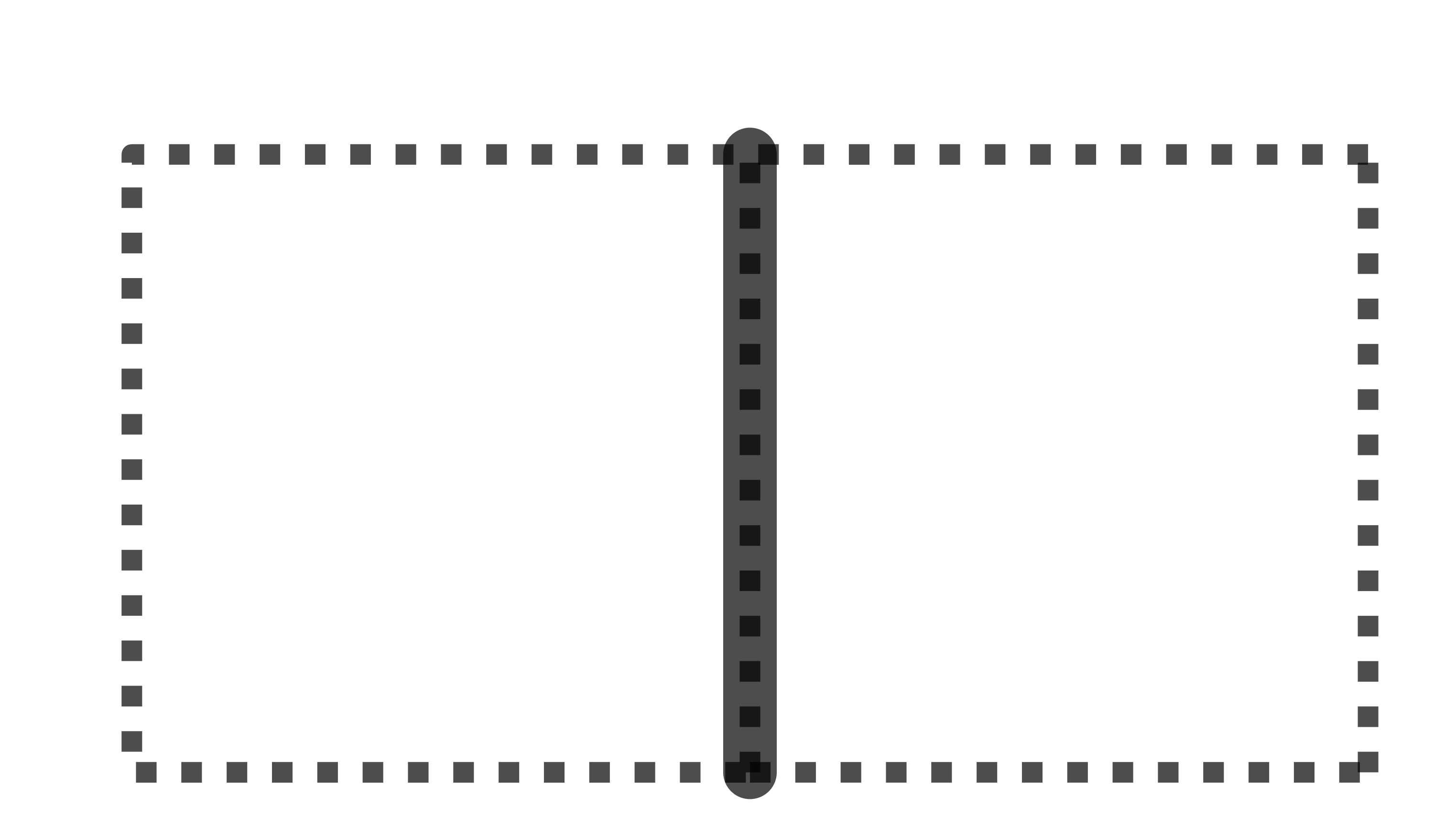}):=
  T(\includegraphics[width=1.0cm]{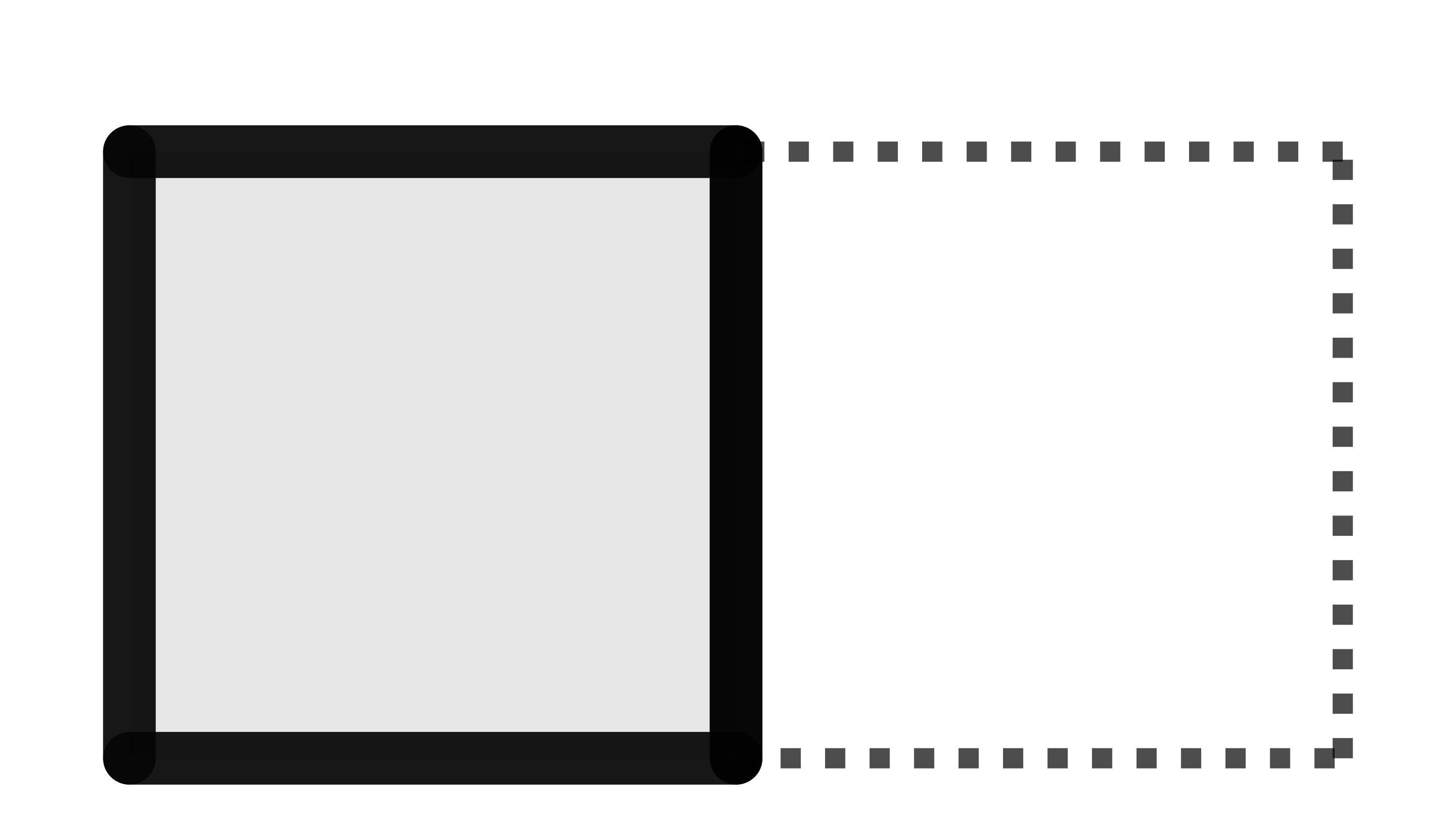})- T(\includegraphics[width=1.0cm]{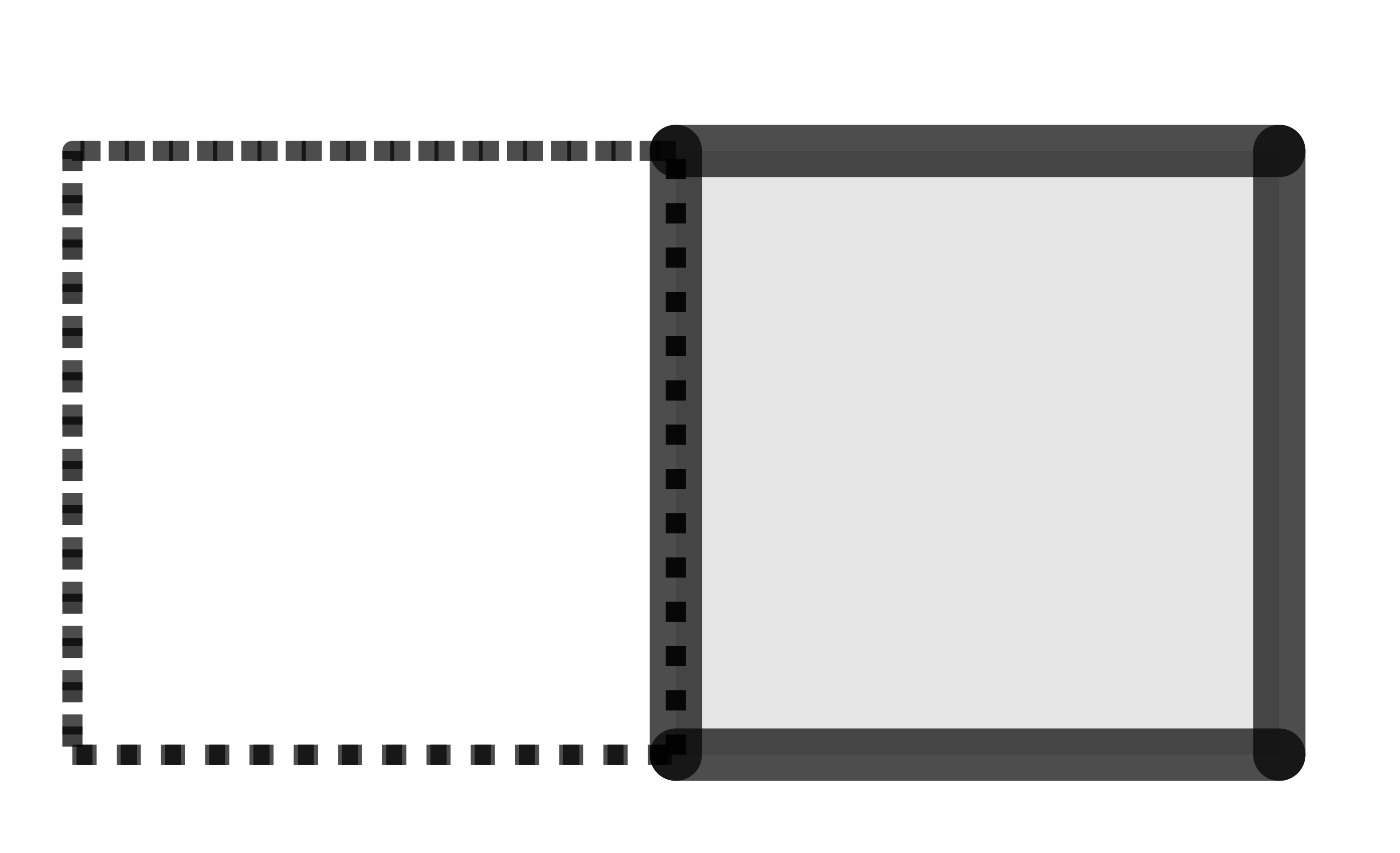})+
  T(\includegraphics[width=1.1cm]{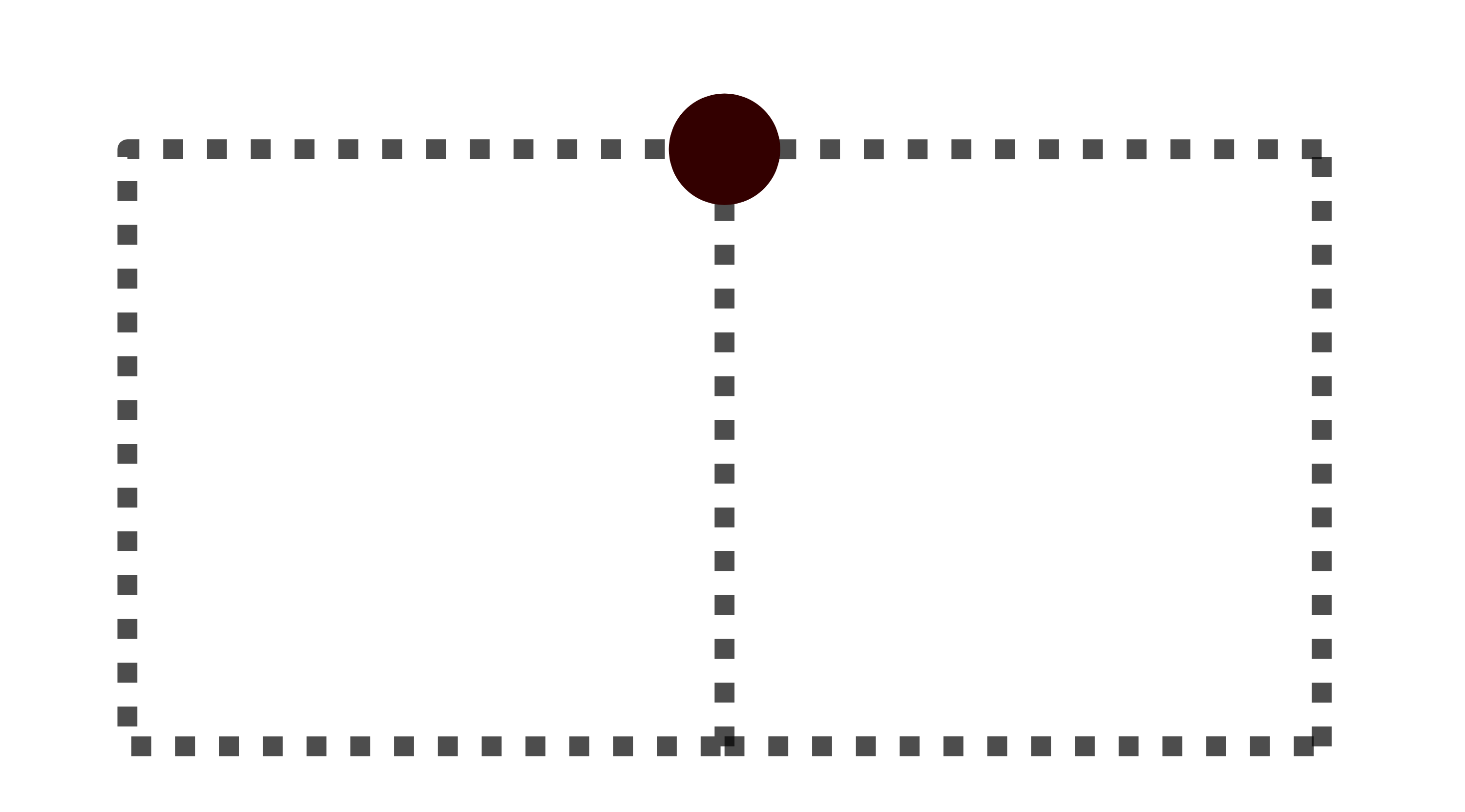})- T(\includegraphics[width=1.1cm]{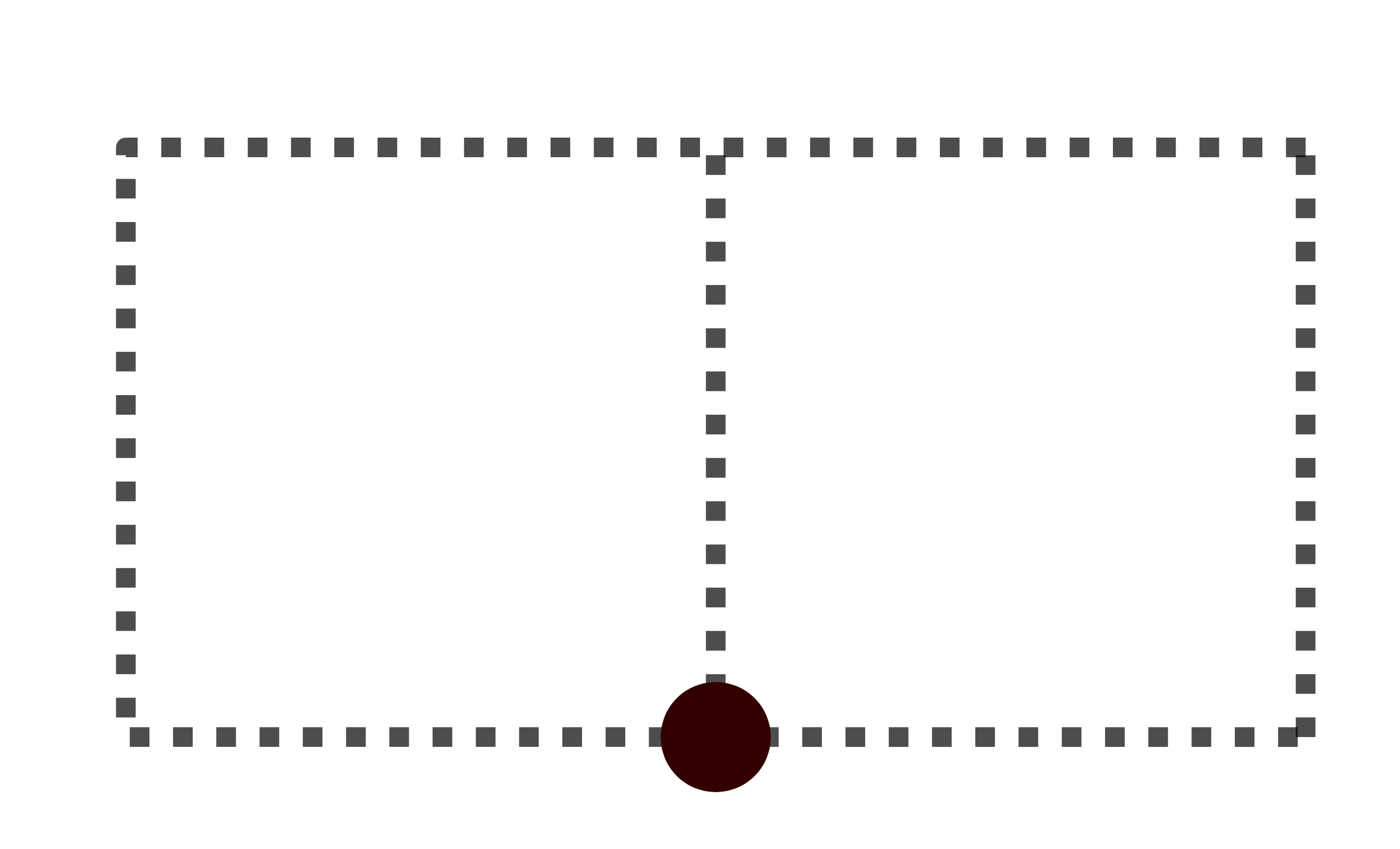})=0$ }
  \end{tabular}
  \end{tabular}
  \caption{\red{The Cartesian square of the $1$-dimensional grid $I_3^1$ and the equation of a conserved tensor; see Definition~\ref{def-tensor}.
  } This is a well-known discretization of the Cauchy--Riemann equations \cite[Eq.~(2.2)]{Bobenko-Skopenkov-13}, up to orientation. Thus tensor conservation means \red{one-half} of the Cauchy--Riemann equations (for \emph{vertical} edges only), like in \cite[Corollary~2.12(1)]{Chelkak-Smirnov-16}, although our setup is very different from theirs.
  }
  \label{fig-notation-cross}
\end{figure}

\begin{definition}\label{def-tensor}
  \move{R12P3}{}Let $I^d_N\times I^d_N$ be the Cartesian square of the $d$-dimensional grid. It is a $2d$-dimensional grid with the faces of the form $e\times f$, where $e$ and $f$ are faces of $I^d_N$ of arbitrary dimension.

  A \emph{tensor of type} $(q,1)$, where $q=1$ or $0$,
  is a function on all faces $e\times f$ \remove{clarity}{of $M\times M$}
  such that $\dim f-\dim e=1-q$. The \emph{chain-cochain \gram{cross product}} of fields $\phi$ and $\psi$ with $\dim\phi-\dim \psi=1-q$ is the tensor
  $$
  [\psi\times \phi](e\times f)=
  \begin{cases}
  \psi(e)\phi(f), &\mbox{if $\dim e=\dim\psi$ and $\dim f=\dim \phi$;}\\
  0, &\mbox{if $\dim e\ne \dim\psi$ or $\dim f\ne \dim\phi$.}
  \end{cases}
  $$
  The \emph{boundary} operator $\partial$ is the unique linear map between the spaces of type $(1,1)$ and $(0,1)$ tensors
  such that for each fields
  $\phi,\psi$ with $\dim\phi=\dim \psi$ we have
  $$
  \partial(\psi\times \phi)=\partial \psi\times \phi+\psi\times\delta \phi.
  $$
  \clarity{(Beware that this is \emph{not} the boundary operator on $I^{2d}_N$.)}
  A type $(1,1)$ tensor $T$ is \emph{conserved}, if $\partial T=0$.
\end{definition}

\begin{theorem}[Energy-momentum conservation]\label{th-energy-conservation}
For each local Lagrangian $\mathcal{L}\colon C^k({M};\mathbb{R})\to C_0({M};\mathbb{R})$ and each
\edit{R1P4}{extremal} $\phi\in C^k({M};\mathbb{R})$ the following \emph{energy-momentum} tensor \clarity{is conserved:} 
  \begin{equation}\label{eq-energy-conservation}
  T[\phi]=\frac{\partial\mathcal{L}[\phi]}{\partial(\delta\phi)}\times \delta\phi +\frac{\partial\mathcal{L}[\phi]}{\partial\phi}\times \phi.
  \end{equation}
\end{theorem}

\edit{R2P1}{The notion of discrete 
tensors extends~\cite{Forman-02}; see Remark~\ref{rem-Forman}.} This theorem is completely new.
\remove{R12P3}{}

\subsection{\red{Summit}}\label{ssec-summit}

\edit{R2P1}{We conclude the summary of main results with the most technical one: an integral form of energy conservation,} sketched already in~\S\ref{ssec-quick}. \move{R2P1}{} To tensor~\eqref{eq-energy-conservation} defined on $M\times M$ we \red{now} assign a conserved quantity defined on the grid $M$ itself. This allows \gram{us} to compare tensors with their continuum \red{analogs.}

\move{R2P1}{}
\begin{definition}\label{def-flow}
Let $\mathrm{e}_k$, where $k=0,\dots,d-1$, be the vector of length $\frac{1}{2}$ pointing in the direction of the \clarity{$x_k$-}axis. Each linear combination of such vectors with coefficients \red{in} the set $\{0,1,\dots, 2N\}$ is the center of a unique face of $I^d_N$. We use the same notation for a face $f$ and its center. In particular, $f+\mathrm{e}_k$ denotes the face with the center at the point obtained from the center of $f$ by \clarity{the shift} by the vector $\mathrm{e}_k$ (the dimensions of $f$ and $f+\mathrm{e}_k$ always \red{differ} by $1$). A \emph{hyperface} is a $(d-1)$-dimensional face. 
\remove{clarity}{}

A type $(1,1)$ tensor $T$ is \emph{partially symmetric}, if
$T(e\times f)=T(f\times e)$ for each pair of faces $e\parallel f$ (hereafter we set $e\parallel f$, if $e$ and $f$ are vertices).
Take a partially symmetric tensor $T$,
a \gram{non-boundary} hyperface $h$,
a number $k\in\{0,\dots, d-1\}$,
and the unique $l\in\{0,\dots, d-1\}$ such that 
$\mathrm{e}_l\perp h$.
Then the \emph{$k$-th component of the flux of 
$T$ across $h$ 
in the positive normal direction} is
\begin{align*}
\langle T,h\rangle_k&=
\frac{1}{2}\sum_{\substack{f:f\subset h,f\ni\max h;\\
f\parallel \mathrm{e}_k \text{ for }h\parallel \mathrm{e}_k}}
  (-1)^{\mathrm{dim}\,\mathrm{Pr}(f,k,l)+l+1}\cdot
\begin{cases}
T((f+\mathrm{e}_l-\mathrm{e}_k)\times f)
+T((f+\mathrm{e}_l+\mathrm{e}_k)\times f), & \mbox{if }h\parallel \mathrm{e}_k; \\
T(f\times f)-T((f+\mathrm{e}_k)\times (f-\mathrm{e}_k)), & \mbox{if } h\perp \mathrm{e}_k,
\end{cases}
\end{align*}
where the sum is over faces $f$ of arbitrary dimension (we set $f\not\parallel\mathrm{e}_k$, if $f$ is a vertex), and
$\mathrm{Pr}(f,k,l)$ is the orthogonal projection of $f$ to the linear span of all $\mathrm{e}_m$ with $\min\{k,l\}\le m\le\max\{k,l\}$.

\clarity{The \emph{$k$-th component of the flux of $T$ across the boundary of a $d$-dimensional face $g$} is
$$\langle T,\partial g\rangle_k:=\sum_{h\subset\partial g \text{ cooriented with }g}\langle T,h\rangle_k
-\sum_{h\subset\partial g \text{ oriented opposite to }g}\langle T,h\rangle_k.$$
}
%
\end{definition}

\begin{theorem}[Integral energy-momentum conservation] \label{prop-global-momentum-conservation}
\clarity{If a partially symmetric type $(1,1)$ tensor $T$ is conserved, then for each $d$-dimensional face $g$ disjoint with $\partial{I}^d_N$ and each $k$ 
we get $\langle T,\partial g\rangle_k=0$.}
\end{theorem}

\clarity{In particular, if tensor~\eqref{eq-energy-conservation} is partially symmetric, then its flux across any closed hypersurface composed of hyperfaces vanishes exactly.}
\move{R2P1}{}
In many examples, \eqref{eq-Noether-current}--\eqref{eq-energy-conservation} approximate their continuum analogues; see Theorem~\ref{th-networks-convergence}, Proposition\ifarxiv{s}{}~\ref{th-Maxwell-approximation}, \ifarxiv{\ref{th-Klein-Gordon-approximation}, \ref{th-Dirac-approximation}, and}{,} 
Remark~\ref{rem-curvature}\ifarxiv{}{, \edit{R2P1}{and \cite[Propositions~A.2.4 and A.3.6]{Skopenkov-preprint}}}. \red{Thus we have established the discretization principles from~\S\ref{s:intro}.}

\subsection{Limitations}\label{ssec-limitations}

So far the proposed general discrete field theory has no applications (as a mathematical theory) and is not refutable (as a candidate for a fundamental physical theory).

Most of the technical issues concern the discretization of energy conservation and tensor calculus:

On the one hand, the new notion of energy-momentum tensor~\eqref{eq-energy-conservation} seems to be too abstract and too general. It discretizes not the continuum energy-momentum tensor precisely but a related object mapped to the latter; see Remark~\ref{rem-tensor}.
Depending on a particular Lagrangian, \eqref{eq-energy-conservation} approximates either the nonsymmetric canonical energy-momentum tensor, or the symmetric Belinfante--Rosenfeld one, or even a nonconserved tensor; see Remark~\ref{rem-energy}.

\edit{R2P1}{}
On the other hand, discrete non-antisymmetric tensor calculus from \S\S\ref{ssec-statements}\red{--}\ref{ssec-summit} seems to be too restrictive: it includes only type $(1,1)$ tensors and only the trivial connection; 
\clarity{the flux} is defined only on a grid. The way of further generalization is unclear: e.g., for lattice gauge theory from \S\ref{ssec-gauge}, a naive way to define a real \remove{clarity}{}
energy-momentum tensor leads to a nonconserved tensor; cf.~Remark~\ref{rem-curvature}.

Concerning approximation of continuum theories by discrete ones, only the following warm-up results are proved: First, for electrical networks the known approximation result is recalled in \S\ref{ssec-Networks}. Second, for the completely new discrete energy-momentum tensor the continuum limit is found in \S\ref{sec-examples}.

Some other limitations are stated as open problems in \S\ref{sec-open}.

\subsection{Overview}

In \S\ref{sec-examples} we give basic examples of discrete field theories. It contains an exposition of known results and a few new ones for nonspecialists; \S\ref{sec-examples} is independent from~\S\ref{s:intro}
\edit{R12P3}{(except that Definitions~\ref{def-cochain}, \ref{def-boundary}, and~\ref{def-tensor} are cited and used
in \S\ref{ssec-Electrodynamics} after they become motivated).} In \S\ref{sec-general} we state the main results in full generality. The only prerequisites for \S\ref{sec-general} are
Definitions~\edit{R12P3}{}\ref{def-cochain}, \ref{def-boundary}, \ref{def-gauge}, \ref{def-projection}, \ref{def-curvature}.
In \S\ref{s:proofs} we prove the results of \S\S\ref{s:intro}--\ref{sec-general}. In \S\ref{sec-open} we state open problems. \edit{R12P3}{More examples are given in}~\ifarxiv{\S\ref{sec-more-examples}}{\cite[\S A]
{Skopenkov-preprint}}.

The paper is written in a mathematical level of rigor, i.e., all the definitions, conventions, and theorems (including corollaries, propositions, lemmas) should be understood literally. Theorems remain true, even if cut out from the text.
The proofs of theorems use the statements but not the proofs of the other ones. Most statements are much less technical than the proofs \edit{R12P3}{and the order of statements is different from the logical order of proofs; thus} the proofs are kept in a separate section. Remarks are informal and are not used elsewhere (hence skippable) unless the opposite is explicitly indicated. 

\remove{R12P3}{} 

\section{Examples}
\label{sec-examples}

\subsection{Electrical networks}
\label{ssec-Networks}

\subsubsection*{Basic model}

\edit{clarity}{We start with the simplest discrete field theory to illustrate and motivate the main concepts.}
Consider an $N\times N$ grid of unit resistors; see Figure~\ref{fig-grid}.
A standard problem is to find currents in the grid, given the current sources at the boundary.
It is solved using the following mathematical model.
\gram{}

\begin{definition} \label{def-epsilon}
Each of the $N^2$ unit squares of the \clarity{$N\times N$} grid is called
a \emph{face}. Orient the boundary $\partial f$ of each face $f$ counterclockwise. Assume that the coordinate axes are parallel to the edges, and orient edges in the directions of the axes. A \emph{function on vertices/edges/faces} is a real-valued function defined on the set of vertices/edges/faces of the grid.

A \emph{source} $s$ is a function on vertices vanishing at all the \gram{non-boundary} vertices.
The \emph{current generated}
by the source $s$, or the \edit{R1P4}{\emph{stationary}} \emph{current}, is the function on edges satisfying two equations:
\begin{itemize}
    \item\label{eq-Kirchhoff} \emph{the Kirchhoff current law or charge conservation law}: $\partial j=-s$;
    \item \emph{the Kirchhoff voltage law in the case of unit resistances}: $\delta j=0$.
\end{itemize}
Here the \emph{boundary} $\partial j$ and the \emph{coboundary} $\delta j$ of a function $j$ on edges are the functions on vertices and faces respectively given by the following formulae (see Figure~\ref{fig-grid} to the middle and the right): \remove{grammar}{}
\begin{align*}
[\partial j](v)&=\sum_{e\textrm{ ending at }v}j(e)-\sum_{e\textrm{ starting at }v}j(e),\\
[\delta j](f)&=\sum_{e\textrm{ oriented along }\partial f}j(e)-\sum_{e\textrm{ oriented opposite to }\partial f}j(e),
\end{align*}
for each vertex $v$ and face $f$, where the sums are over edges $e$ containing $v$ and contained in $\partial f$ respectively. 
Denote by $\epsilon s:=\sum_{v} s (v)$ the sum over all vertices~$v$ (the operator $\epsilon$ is defined only for functions on vertices).
\end{definition}

\begin{figure}[!t] 
\vspace{-0.3cm}
  \centering
  \begin{tabular}{cc}
   {\small\definecolor{qqqqff}{rgb}{0.,0.,1.}
\definecolor{ffqqqq}{rgb}{1.,0.,0.}
\begin{tikzpicture}[line cap=round,line join=round,>=triangle 45,x=0.2cm,y=0.2cm]
\draw[->,color=black] (-1.36,0.) -- (11.46,0.);
\foreach \x in {,3.,6.,9.}
\draw[shift={(\x,0)},color=black] (0pt,2pt) -- (0pt,-2pt);
\draw[->,color=black] (0.,-1.28) -- (0.,10.56);
\foreach \y in {,3.,6.,9.}
\draw[shift={(0,\y)},color=black] (2pt,0pt) -- (-2pt,0pt);
\clip(-1.36,-1.28) rectangle (11.46,10.56);
\draw [->,line width=0.1pt,color=qqqqff] (0.,0.) -- (0.,3.);
\draw [->,line width=0.1pt,color=qqqqff] (0.,6.) -- (0.,9.);
\draw [->,line width=0.1pt,color=qqqqff] (0.,3.) -- (0.,6.);
\draw [->,line width=0.1pt,color=qqqqff] (0.,0.) -- (3.,0.);
\draw [->,line width=0.1pt,color=qqqqff] (0.,3.) -- (3.,3.);
\draw [->,line width=0.1pt,color=qqqqff] (0.,6.) -- (3.,6.);
\draw [->,line width=0.1pt,color=qqqqff] (0.,9.) -- (3.,9.);
\draw [->,line width=0.1pt,color=qqqqff] (3.,0.) -- (3.,3.);
\draw [->,line width=0.1pt,color=qqqqff] (3.,3.) -- (3.,6.);
\draw [->,line width=0.1pt,color=qqqqff] (3.,6.) -- (3.,9.);
\draw [->,line width=0.1pt,color=qqqqff] (3.,0.) -- (6.,0.);
\draw [->,line width=0.1pt,color=qqqqff] (6.,0.) -- (6.,3.);
\draw [->,line width=0.1pt,color=qqqqff] (3.,3.) -- (6.,3.);
\draw [->,line width=0.1pt,color=qqqqff] (3.,6.) -- (6.,6.);
\draw [->,line width=0.1pt,color=qqqqff] (6.,3.) -- (6.,6.);
\draw [->,line width=0.1pt,color=qqqqff] (3.,9.) -- (6.,9.);
\draw [->,line width=0.1pt,color=qqqqff] (6.,6.) -- (6.,9.);
\draw [->,line width=0.1pt,color=qqqqff] (6.,0.) -- (9.,0.);
\draw [->,line width=0.1pt,color=qqqqff] (9.,0.) -- (9.,3.);
\draw [->,line width=0.1pt,color=qqqqff] (6.,3.) -- (9.,3.);
\draw [->,line width=0.1pt,color=qqqqff] (6.,6.) -- (9.,6.);
\draw [->,line width=0.1pt,color=qqqqff] (9.,3.) -- (9.,6.);
\draw [->,line width=0.1pt,color=qqqqff] (6.,9.) -- (9.,9.);
\draw [->,line width=0.1pt,color=qqqqff] (9.,6.) -- (9.,9.);
\draw [shift={(1.52,1.48)},line width=0.8pt,dotted,color=ffqqqq]  plot[domain=0.7216548508647611:5.81953769817878,variable=\t]({1.*0.6660330322138684*cos(\t r)+0.*0.6660330322138684*sin(\t r)},{0.*0.6660330322138684*cos(\t r)+1.*0.6660330322138684*sin(\t r)});
\draw [->,line width=0.1pt,color=ffqqqq] (1.7888377107094406,0.8706345223919335) -- (2.12,1.18);
\draw [shift={(4.52,1.48)},line width=0.8pt,dotted,color=ffqqqq]  plot[domain=0.7216548508647611:5.81953769817878,variable=\t]({1.*0.6660330322138684*cos(\t r)+0.*0.6660330322138684*sin(\t r)},{0.*0.6660330322138684*cos(\t r)+1.*0.6660330322138684*sin(\t r)});
\draw [shift={(7.52,1.48)},line width=0.8pt,dotted,color=ffqqqq]  plot[domain=0.7216548508647611:5.81953769817878,variable=\t]({1.*0.6660330322138684*cos(\t r)+0.*0.6660330322138684*sin(\t r)},{0.*0.6660330322138684*cos(\t r)+1.*0.6660330322138684*sin(\t r)});
\draw [shift={(1.52,4.48)},line width=0.8pt,dotted,color=ffqqqq]  plot[domain=0.7216548508647611:5.81953769817878,variable=\t]({1.*0.6660330322138684*cos(\t r)+0.*0.6660330322138684*sin(\t r)},{0.*0.6660330322138684*cos(\t r)+1.*0.6660330322138684*sin(\t r)});
\draw [shift={(1.52,7.48)},line width=0.8pt,dotted,color=ffqqqq]  plot[domain=0.7216548508647611:5.81953769817878,variable=\t]({1.*0.6660330322138684*cos(\t r)+0.*0.6660330322138684*sin(\t r)},{0.*0.6660330322138684*cos(\t r)+1.*0.6660330322138684*sin(\t r)});
\draw [shift={(4.52,4.48)},line width=0.8pt,dotted,color=ffqqqq]  plot[domain=0.7216548508647611:5.81953769817878,variable=\t]({1.*0.6660330322138684*cos(\t r)+0.*0.6660330322138684*sin(\t r)},{0.*0.6660330322138684*cos(\t r)+1.*0.6660330322138684*sin(\t r)});
\draw [shift={(7.52,4.48)},line width=0.8pt,dotted,color=ffqqqq]  plot[domain=0.7216548508647611:5.81953769817878,variable=\t]({1.*0.6660330322138684*cos(\t r)+0.*0.6660330322138684*sin(\t r)},{0.*0.6660330322138684*cos(\t r)+1.*0.6660330322138684*sin(\t r)});
\draw [shift={(4.52,7.48)},line width=0.8pt,dotted,color=ffqqqq]  plot[domain=0.7216548508647611:5.81953769817878,variable=\t]({1.*0.6660330322138684*cos(\t r)+0.*0.6660330322138684*sin(\t r)},{0.*0.6660330322138684*cos(\t r)+1.*0.6660330322138684*sin(\t r)});
\draw [shift={(7.52,7.48)},line width=0.8pt,dotted,color=ffqqqq]  plot[domain=0.7216548508647611:5.81953769817878,variable=\t]({1.*0.6660330322138684*cos(\t r)+0.*0.6660330322138684*sin(\t r)},{0.*0.6660330322138684*cos(\t r)+1.*0.6660330322138684*sin(\t r)});
\draw [->,line width=0.1pt,color=ffqqqq] (1.7888377107094406,3.8706345223919336) -- (2.12,4.18);
\draw [->,line width=0.1pt,color=ffqqqq] (1.7888377107094406,6.870634522391933) -- (2.12,7.18);
\draw [->,line width=0.1pt,color=ffqqqq] (4.788837710709441,0.8706345223919335) -- (5.12,1.18);
\draw [->,line width=0.1pt,color=ffqqqq] (7.788837710709441,0.8706345223919335) -- (8.12,1.18);
\draw [->,line width=0.1pt,color=ffqqqq] (7.788837710709441,3.8706345223919336) -- (8.12,4.18);
\draw [->,line width=0.1pt,color=ffqqqq] (7.788837710709441,6.870634522391933) -- (8.12,7.18);
\draw [->,line width=0.1pt,color=ffqqqq] (4.788837710709441,6.870634522391933) -- (5.12,7.18);
\draw [->,line width=0.1pt,color=ffqqqq] (4.788837710709441,3.8706345223919336) -- (5.12,4.18);
\begin{scriptsize}
\draw [color=ffqqqq] (0.,0.) circle (1.5pt);
\draw [color=ffqqqq] (0.,3.) circle (1.5pt);
\draw [color=ffqqqq] (0.,6.) circle (1.5pt);
\draw [color=ffqqqq] (0.,9.) circle (1.5pt);
\draw [color=ffqqqq] (3.,0.) circle (1.5pt);
\draw [fill=qqqqff] (3.,3.) circle (1.5pt);
\draw [fill=qqqqff] (3.,6.) circle (1.5pt);
\draw[color=qqqqff] (3.52,6.77) node {d};
\draw [color=ffqqqq] (3.,9.) circle (1.5pt);
\draw[color=qqqqff] (3.52,3.79) node {a};
\draw [color=ffqqqq] (6.,0.) circle (1.5pt);
\draw [fill=qqqqff] (6.,3.) circle (1.5pt);
\draw [fill=qqqqff] (6.,6.) circle (1.5pt);
\draw[color=qqqqff] (6.58,6.79) node {c};
\draw[color=qqqqff] (6.54,3.79) node {b};
\draw [color=ffqqqq] (6.,9.) circle (1.5pt);
\draw [color=ffqqqq] (9.,0.) circle (1.5pt);
\draw [color=ffqqqq] (9.,3.) circle (1.5pt);
\draw [color=ffqqqq] (9.,6.) circle (1.5pt);
\draw [color=ffqqqq] (9.,9.) circle (1.5pt);
\draw[color=ffqqqq] (4.46,4.45) node {$f$};
\draw[color=ffqqqq] (7.38,4.31) node {$g$};
\end{scriptsize}
\end{tikzpicture} }&
  \begin{tabular}{cc}\\[-3cm] 
   {\small\definecolor{qqqqff}{rgb}{0.,0.,1.}
\definecolor{ffqqqq}{rgb}{1.,0.,0.}
\begin{tikzpicture}[line cap=round,line join=round,>=triangle 45,x=0.2cm,y=0.2cm]
\clip(-3.46,-0.44) rectangle (3.44,6.34);
\draw [->,line width=0.1pt,color=qqqqff] (-3.,3.) -- (0.,3.);
\draw [->,line width=0.1pt,color=qqqqff] (0.,0.) -- (0.,3.);
\draw [->,line width=0.1pt,color=qqqqff] (0.,3.) -- (0.,6.);
\draw [->,line width=0.1pt,color=qqqqff] (0.,3.) -- (3.,3.);
\begin{scriptsize}
\draw [color=ffqqqq] (-3.,3.) circle (1.5pt);
\draw [fill=qqqqff] (0.,0.) circle (1.5pt);
\draw [fill=qqqqff] (0.,3.) circle (1.5pt);
\draw[color=qqqqff] (0.58,3.77) node {$v$};
\draw[color=qqqqff] (-1.84,2.59) node {$\mathbf{3}$};
\draw [color=ffqqqq] (0.,6.) circle (1.5pt);
\draw[color=qqqqff] (-0.46,1.43) node {$\mathbf{4}$};
\draw[color=qqqqff] (-0.46,4.65) node {$\mathbf{2}$};
\draw [fill=qqqqff] (3.,3.) circle (1.5pt);
\draw[color=qqqqff] (1.48,2.61) node {$\mathbf{1}$};
\end{scriptsize}
\end{tikzpicture} } &
   {\small\definecolor{ffqqqq}{rgb}{1.,0.,0.}
\definecolor{qqqqff}{rgb}{0.,0.,1.}
\begin{tikzpicture}[line cap=round,line join=round,>=triangle 45,x=0.2cm,y=0.2cm]
\clip(-3.76,-0.92) rectangle (0.96,3.7);
\draw [->,line width=0.1pt,color=qqqqff] (-3.,0.) -- (-3.,3.);
\draw [->,line width=0.1pt,color=qqqqff] (-3.,0.) -- (0.,0.);
\draw [->,line width=0.1pt,color=qqqqff] (-3.,3.) -- (0.,3.);
\draw [->,line width=0.1pt,color=qqqqff] (0.,0.) -- (0.,3.);
\draw [shift={(-1.48,1.48)},dotted,color=ffqqqq]  plot[domain=0.7216548508647611:5.81953769817878,variable=\t]({1.*0.6660330322138684*cos(\t r)+0.*0.6660330322138684*sin(\t r)},{0.*0.6660330322138684*cos(\t r)+1.*0.6660330322138684*sin(\t r)});
\draw [->,line width=0.1pt,color=ffqqqq] (-1.0275222870382181,0.9912629344188479) -- (-0.88,1.18);
\begin{scriptsize}
\draw [fill=qqqqff] (-3.,0.) circle (1.5pt);
\draw [fill=qqqqff] (-3.,3.) circle (1.5pt);
\draw[color=qqqqff] (-3.5,1.53) node {$\mathbf{4}$};
\draw [fill=qqqqff] (0.,0.) circle (1.5pt);
\draw[color=qqqqff] (-1.48,-0.47) node {$\mathbf{1}$};
\draw [fill=qqqqff] (0.,3.) circle (1.5pt);
\draw[color=qqqqff] (-1.54,3.33) node {$\mathbf{3}$};
\draw[color=qqqqff] (0.26,1.49) node {$\mathbf{2}$};
\draw[color=ffqqqq] (-1.62,1.43) node {$f$};
\end{scriptsize}
\end{tikzpicture} }\\
  $[\partial j](v)=-j(\mathbf{1})-j(\mathbf{2})+j(\mathbf{3})+j(\mathbf{4})$; &
  $[\delta j](f)=j(\mathbf{1})+j(\mathbf{2})-j(\mathbf{3})-j(\mathbf{4})$;\\
  $[\phi\frown\psi](v)=\phi(\mathbf{3})\psi(\mathbf{3})+\phi(\mathbf{4})\psi(\mathbf{4})$; &
  $[\phi\smile\psi](f)=\phi(\mathbf{1})\psi(\mathbf{2})-\phi(\mathbf{4})\psi(\mathbf{3})$;
  \end{tabular}
  \end{tabular}
  \vspace{-0.5cm}
  \caption{A $3\times 3$ grid, boundary, coboundary; cap and \red{cup product} (see Definitions~\ref{def-epsilon}, \ref{def-cap}, \ref{def-cup})} \label{fig-grid}
\end{figure}

The following existence and uniqueness result is well-known.

\begin{proposition}\label{prop-current-existence}
A current generated by a source $s$ exists, if and only if $\epsilon s=0$. If a current generated by the source $s$ exists, then it is unique.
\end{proposition}

\begin{remark} It could be more conceptual to write the Kirchhoff voltage law in the form $\delta Rj=0$, where $R$ is a map between $1$-chains and $1$-cochains depending on the resistances. In our setup, chains and cochains are identified and the resistances equal $1$, hence $R$ is the identity map and is omitted.
\end{remark}

\subsubsection*{Electrical potential}

Let us state a least-action principle for electrical networks. Throughout \S\ref{ssec-Networks} $j$ is a \edit{R1P4}{stationary} current.

\begin{definition}
An \emph{electrical potential} $\phi$ is a function on vertices satisfying
\begin{itemize}
\item\label{eq-Ohm} \emph{the Ohm law in the case of unit resistances}: $j=-\delta \phi$.
\end{itemize}
Here the \emph{coboundary} $\delta\phi$ is the function on edges given by the formula
$$
[\delta\phi](uv)=\phi(v)-\phi(u),
$$
where $uv$ denotes an oriented edge starting at $u$ and ending at $v$ hereafter.
\end{definition}

The following well-known existence and uniqueness result is straightforward.

\begin{proposition}
  For each
  \edit{R1P4}{stationary} current there is a unique up to additive constant electrical potential.
\end{proposition}

The following properties of an electrical potential $\phi$ may serve as equivalent definitions:
\begin{itemize}
  \item\emph{the Laplace equation with the Neumann boundary condition}: $\partial\delta\phi=s$;
  \item\emph{the least action principle}: among all the functions on vertices, $\phi$ minimizes the functional
      \begin{align*}
      \mathcal{S}[\phi]&=\frac{1}{2}\sum_{\text{edges } uv}(\phi(u)-\phi(v))^2-\sum_{\text{vertices } v}s(v)\phi(v)=
      \epsilon\mathcal{L}[\phi], \qquad \text{where}\\
      \mathcal{L}[\phi]&=\tfrac{1}{2}\delta\phi\frown \delta\phi-s\frown \phi.
      \end{align*}
\end{itemize}

Here the (particular case of) cap-product $\frown$ is defined as follows; see Figure~\ref{fig-grid} to the middle.

\begin{definition} \label{def-cap}
Denote by $\max f$ the vertex of a face $f$ or an edge $f$ having the maximal sum of the coordinates. Set $\max f:=f$, if $f$ is a vertex. The \emph{cap-product} $\phi\frown\psi$ of two functions $\phi$ and $\psi$ on faces (respectively, edges or vertices) is the function on vertices given by
$$
[\phi\frown \psi](v)=\sum_{f\,:\,\max f=v}\phi(f)\psi(f),
$$
where the sum is over faces (respectively, edges or vertices) $f$ such that $\max f=v$.
\end{definition}

\begin{remark}
\edit{R2P1}{The Euler--Lagrange equation (given by Theorem~\ref{th-Euler-Lagrange}) for the Lagrangian $\mathcal{L}[\phi]$ is the Laplace equation. Apart the grid boundary,}\move{R2P1}{} the Lagrangian is invariant under the transformation $\phi\mapsto\phi-t$, where $t\in\mathbb{R}$. The resulting conserved current (given by Theorem~\ref{th-Noether}) is \red{$j=-\delta\phi$.}
\end{remark}

\subsubsection*{Magnetic field}


There is one more discrete field in an electrical network: the current $j$ generates a magnetic field.

\begin{definition} 
A \emph{magnetic field} $F$ (or \emph{magnetic flux through faces in the $(0,0,-1)$-direction}) \emph{generated} by a 
current $j$ is a function on faces satisfying the following equation apart the 
boundary: \remove{clarity}{}
\begin{itemize}
  \item \emph{the Amp\`ere law in the case of unit-area faces}: $-\partial F=j$.
\end{itemize}
Here the \emph{boundary} $\partial F$ is the function on edges given by the formula
$$
[\partial F](e)=F(f)-F(g)
$$
for each pair of adjacent faces $f$ and $g$
such that $\partial f$ (respectively, $\partial g$) is oriented along (respectively, opposite to) the common edge $e$; see Figure~\ref{fig-grid} to the left. (The definition of $[\partial F](e)$ for boundary edges $e$ is not required for this subsection.) \remove{R12P3}{}
\end{definition}

The following well-known existence and uniqueness result is straightforward.

\begin{proposition}\label{prop-magnetic-existence}
For 
\edit{R1P4}{a stationary} current there is a unique up to 
additive constant magnetic field.
\end{proposition}

Throughout \S\ref{ssec-Networks} the functions $\phi$ and $F$ are an electrical potential and a magnetic field respectively.

\begin{remark}
  The pair $(\phi,F)$ and $-j$ \clarity{discretize} an analytic function
  and its derivative
  \cite{Chelkak-Smirnov-08,Bobenko-Skopenkov-13}.
\end{remark}

\begin{definition}
A \emph{magnetic 
vector potential} $A$ of the field $F$ is a function on edges such that  $\delta A=F$.
\end{definition}

A magnetic \gram{vector potential} $A$ has the following properties
(proved similarly to the ones from~\S\ref{ssec-Electrodynamics}):
\begin{itemize}
  \item\emph{the source equation}: $-\partial\delta A=j$ apart the grid boundary;
  \item\emph{gauge invariance}: $A+\delta g$ is a \gram{vector potential} of the same field for any function $g$ on vertices;
  \item\emph{the least action principle}: among all functions on edges, $A$ minimizes  $\mathcal{S}[A]=\epsilon\mathcal{L}[A]$, where
      $$
      \mathcal{L}[A]=\tfrac{1}{2}\delta A\frown \delta A+j\frown A.
      $$
\end{itemize}

\subsubsection*{Energy and momentum}

Let us state energy and momentum conservation in an electrical network in a simple heuristic form. This is a visual motivation for more abstract \edit{R12P3}{Definition~\ref{def-tensor} (not used in this subsection).}  

For functions $\phi$, $\psi$ on faces (respectively, edges or vertices), denote by $\langle\phi,\psi\rangle=\sum_f \phi(f)\psi(f)$ the sum over all faces~\clarity{$f$} (respectively, edges or vertices). The obvious identity $\langle\delta\phi,j\rangle=\langle\phi,\partial j\rangle$ implies 
\begin{itemize}
\item \emph{the Tellegen theorem or \emph{global} energy conservation}: $\langle\delta\phi,j\rangle+\langle\phi,s\rangle=0$.
\end{itemize}

Now we study \emph{local} conservation and the flow of energy. 
Energy flows in the direction of the Poynting vector, hence \emph{transversely to} (not along) the resistors.
Thus we assign energy flow to \emph{bisectors} of edges. The cross-product formula for the Poynting vector is then discretized directly.

\begin{definition} \label{def-doubling}
The \emph{doubling} is the $2N\times 2N$ grid with the vertices at vertices, edge midpoints, and face centers of the initial $N\times N$ grid. Orient all the edges still in the direction of the coordinate axes.

The \emph{heat power} $W$ is the function on the vertices $v$ of the doubling given by the formula (Figure~\ref{fig-SW})
$$
W(v)=
\begin{cases}
  -[\delta\phi](e)j(e), & \mbox{if $v$ is the midpoint of an edge }e;  \\
  0, & \mbox{if $v$ is the center of a face or a vertex of the initial grid}.
\end{cases}
$$
The \emph{Poynting vector} or \emph{energy flux} $S$ is the function on edges $uv$ of the doubling, $\max uv=v$, given by
$$
S(uv)=
\begin{cases}
[\delta\phi](e)F(f), &\mbox{$u$ and $v$ are the centers of a vertical edge $e$ and a face $f$ or vice versa;}\\
-[\delta\phi](e)F(f), &\mbox{$u$ and $v$ are the centers of a horizontal edge $e$ and a face $f$ or vice versa;}\\
0, & \mbox{$u$ or $v$ is a vertex of the initial grid}.
\end{cases}
$$
The \emph{Lorentz force} $L$ is defined analogously to $S$, only $\delta\phi$ is replaced by $-j/2$ (thus $L=S/2$ in our basic model). The \emph{magnetic pressure} $P$ (or \emph{momentum flux of the magnetic field towards the edges in the normal direction}) is the function on \gram{non-boundary} vertices $v$ of the doubling given by the formula
$$
P(v)=
\begin{cases}
  F(f)F(f)/2, & \mbox{if $v$ is the center of a face }f;  \\
  F(f)F(g)/2, & \mbox{if $v$ is the midpoint of the common edge of faces $f$ and $g$};\\
  0,        & \mbox{if $v$ is a vertex of the initial grid}.  \\
\end{cases}
$$
\end{definition}

\begin{figure}[htb]
  \centering
  \vspace{-0.5cm}
  \begin{tabular}{cl}
  \begin{tabular}{l}
  \includegraphics[width=2.5cm]{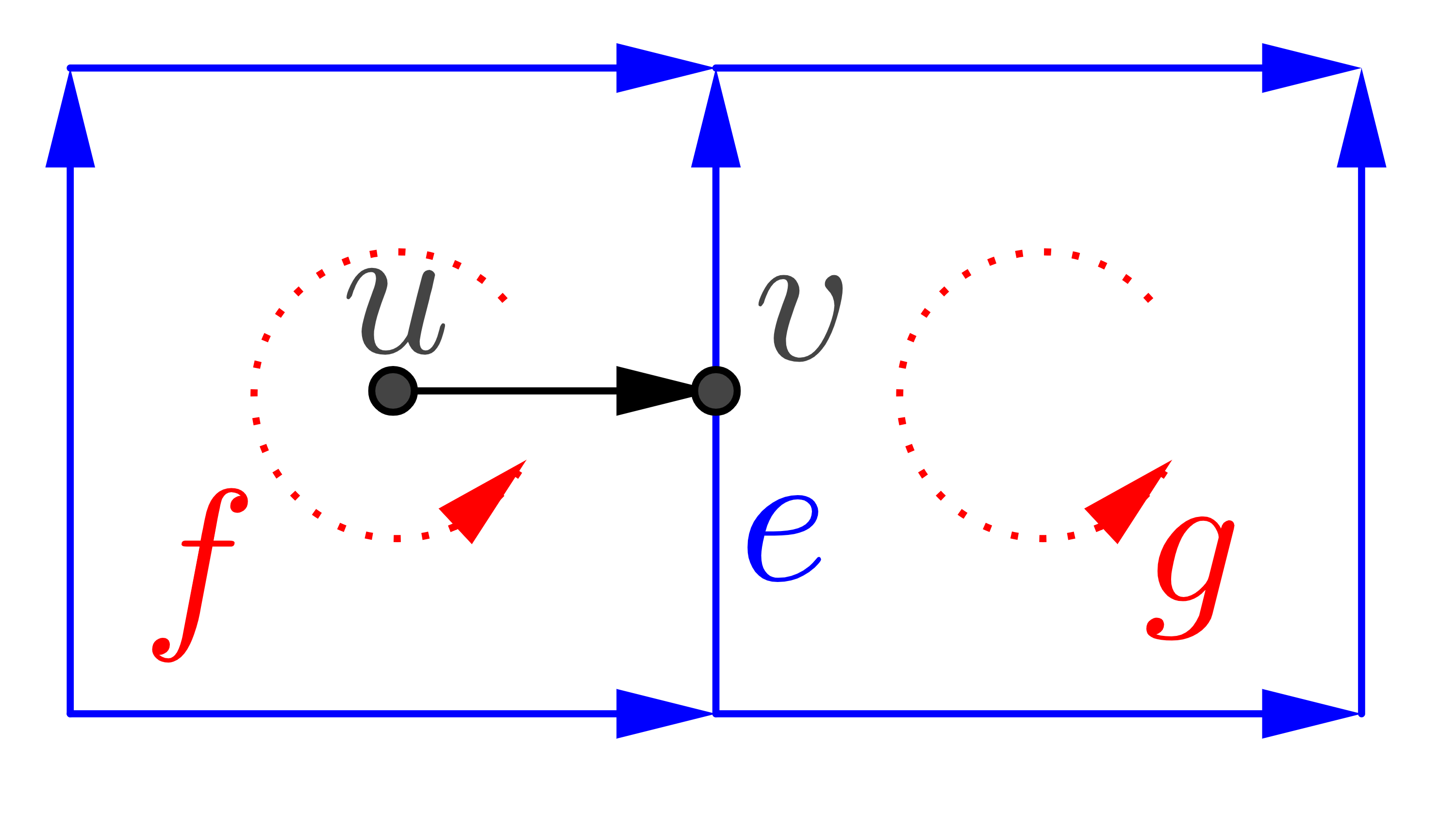}
  \end{tabular}
&
\begin{tabular}{l}
$W(\includegraphics[width=0.4cm]{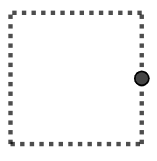})=
[\delta\phi](\includegraphics[width=0.4cm]{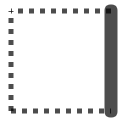})\cdot
j(\includegraphics[width=0.4cm]{square-r.png})$\\
$S(\includegraphics[width=0.4cm]{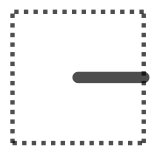})=
[\delta\phi](\includegraphics[width=0.4cm]{square-r.png})\cdot
F(\includegraphics[width=0.4cm]{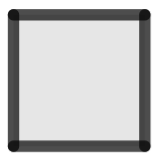})$
\end{tabular}
\end{tabular}
  \caption{Notation in Definition~\ref{def-doubling} of discrete energy flux. The face $f$ is shown by dotted lines to the right. The vertex, edge, or face \red{at} which a particular function is evaluated is shown in bold.}
  \label{fig-SW}
\end{figure}

\gram{The} straightforward consequences of these definitions and the Kirchhoff laws are:
\begin{itemize}
  \item \emph{Energy conservation}: $\partial S-W=0$. 
  \item \emph{Momentum conservation for the magnetic field}: $\delta P+L=0$ on those edges of the doubling that contain the face-centers of the initial grid.
\end{itemize}

In \S\ref{ssec-Electrodynamics} we introduce a more conceptual form of the two laws, explaining the latter restriction.

Now we state a less visual momentum conservation law for the \emph{electric} field. This is essentially \cite[Example in \S8 \red{of Ch.~III}]{Dorodnitsyn-04}. One expects the following properties of the momentum flux $\sigma(e)$ across edges $e$ of the initial grid (the latter property \clarity{is} required by the discretization principles from~\S\ref{s:intro}):
\begin{itemize}
\item
$\sigma(e)$ equals the momentum flux of a continuum electric field across $e$, if the potential is linear;
\item
$\sigma(e)$ depends only on the values of $\delta\phi$ at the edges intersecting $e$ and is bilinear in these values;
\item
$\delta\sigma=0$ apart the grid boundary: the momentum flux across the boundary of each face  vanishes.
\end{itemize}
The simplest function $\sigma$ satisfying these properties is defined as follows; cf.~Figure~\ref{fig-notation}. \remove{R12P3}

\begin{definition}\label{def-Maxwell-stress}
The \emph{momentum flux of the electric field across edges in the negative normal direction}, or \emph{the electric part of the Maxwell stress tensor}, is the pair $\sigma=(\sigma_1,\sigma_2)$ of functions on edges disjoint with the grid boundary given by the following formula for each $k=1,2$:
$$
\sigma_k(uv)=
\frac{(-1)^{k+1}}{2}
\begin{cases}
  \delta\phi(uu_+)\delta\phi(uv)+\delta\phi(vv_+)\delta\phi(uv), & \mbox{if }uv\parallel Ox_k;\\
  \delta\phi(uv)\delta\phi(uv)-\delta\phi(vv_+)\delta\phi(v_-v), & \mbox{if }uv\perp Ox_k,
\end{cases}
$$
where $uu_+$, $v_-v$, $vv_+$ are the edges orthogonal to $uv$ with the maximal vertices $u_+$, $v$, $v_+$; see Figure~\ref{fig-notation}.
\end{definition}

\begin{figure}[htb]
  \centering
  \begin{tabular}{cl}
\begin{tabular}{l}
  \definecolor{qqqqff}{rgb}{0.,0.,1.}
\definecolor{cqcqcq}{rgb}{0.7529411764705882,0.7529411764705882,0.7529411764705882}
\begin{tikzpicture}[line cap=round,line join=round,>=triangle 45,x=1.2cm,y=1.2cm]
\draw [color=cqcqcq,, xstep=0.6cm,ystep=0.6cm] (-0.06317423003158845,-0.07915055778072938) grid (1.870035745713,1.8494208509685262);
\draw[->,color=black] (-0.06317423003158845,0.) -- (1.870035745713,0.);
\foreach \x in {,0.5,1.,1.5}
\draw[shift={(\x,0)},color=black] (0pt,2pt) -- (0pt,-2pt);
\draw[color=black] (1.6913356639214834,0.023166022927918986) node [anchor=south west] { $x_1$};
\draw[->,color=black] (0.,-0.07915055778072938) -- (0.,1.8494208509685262);
\foreach \y in {,0.5,1.,1.5}
\draw[shift={(0,\y)},color=black] (2pt,0pt) -- (-2pt,0pt);
\draw[color=black] (0.02707576996841159,1.693050196205073) node [anchor=west] { $x_2$};
\clip(-0.06317423003158845,-0.07915055778072938) rectangle (1.870035745713,1.8494208509685262);
\draw [->,line width=0.8pt,color=qqqqff] (1.,0.5) -- (1.,1.);
\draw [->,line width=0.8pt,color=qqqqff] (0.5,1.) -- (1.,1.);
\draw [->,line width=0.8pt,color=qqqqff] (1.,1.) -- (1.5,1.);
\draw [->,line width=0.8pt,color=qqqqff] (1.,0.5) -- (1.5,0.5);
\begin{scriptsize}
\draw [fill=qqqqff] (1.,0.5) circle (2.5pt);
\draw[color=qqqqff] (1.160650572540616,0.6071428714588706) node {$u$};
\draw [fill=qqqqff] (1.,1.) circle (2.5pt);
\draw[color=qqqqff] (1.1660657265342982,1.1515444102649668) node {$v$};
\draw [fill=qqqqff] (1.5,0.5) circle (2.5pt);
\draw[color=qqqqff] (1.6832129329309597,0.6563706701806985) node {$u_+$};
\draw [fill=qqqqff] (0.5,1.) circle (2.5pt);
\draw[color=qqqqff] (0.6705791361123659,1.1544401631309567) node {$v_-$};
\draw [fill=qqqqff] (1.5,1.) circle (2.5pt);
\draw[color=qqqqff] (1.7102887028993712,1.1370656459350175) node {$v_+$};
\end{scriptsize}
\end{tikzpicture}\quad
  \definecolor{qqqqff}{rgb}{0.,0.,1.}
\definecolor{cqcqcq}{rgb}{0.7529411764705882,0.7529411764705882,0.7529411764705882}
\begin{tikzpicture}[line cap=round,line join=round,>=triangle 45,x=1.2cm,y=1.2cm]
\draw [color=cqcqcq,, xstep=0.6cm,ystep=0.6cm] (-0.06858938402527076,-0.09652507497666793) grid (1.6588447399593886,1.76833977072081);
\draw[->,color=black] (-0.06858938402527076,0.) -- (1.6588447399593886,0.);
\foreach \x in {,0.5,1.,1.5}
\draw[shift={(\x,0)},color=black] (0pt,2pt) -- (0pt,-2pt);
\draw[color=black] (1.4801446581678723,0.02316602292791898) node [anchor=south west] { $x_1$};
\draw[->,color=black] (0.,-0.09652507497666793) -- (0.,1.76833977072081);
\foreach \y in {,0.5,1.,1.5}
\draw[shift={(0,\y)},color=black] (2pt,0pt) -- (-2pt,0pt);
\draw[color=black] (0.027075769968411604,1.6119691159573568) node [anchor=west] { $x_2$};
\clip(-0.06858938402527076,-0.09652507497666793) rectangle (1.6588447399593886,1.76833977072081);
\draw [->,line width=0.8pt,color=qqqqff] (0.5,1.) -- (1.,1.);
\draw [->,line width=0.8pt,color=qqqqff] (1.,0.5) -- (1.,1.);
\draw [->,line width=0.8pt,color=qqqqff] (0.5,1.) -- (0.5,1.5);
\draw [->,line width=0.8pt,color=qqqqff] (1.,1.) -- (1.,1.5);
\begin{scriptsize}
\draw [fill=qqqqff] (0.5,1.) circle (2.5pt);
\draw[color=qqqqff] (0.6299654811597484,1.1167953758730886) node {$u$};
\draw [fill=qqqqff] (1.,1.) circle (2.5pt);
\draw[color=qqqqff] (1.1444051105595685,1.111003870141109) node {$v$};
\draw [fill=qqqqff] (1.,0.5) circle (2.5pt);
\draw[color=qqqqff] (1.1362823795690453,0.6447876587167394) node {$v_-$};
\draw [fill=qqqqff] (0.5,1.5) circle (2.5pt);
\draw[color=qqqqff] (0.6435033661439541,1.6351351388852757) node {$u_+$};
\draw [fill=qqqqff] (1.,1.5) circle (2.5pt);
\draw[color=qqqqff] (1.1633581495374568,1.6235521274213163) node {$v_+$};
\end{scriptsize}
\end{tikzpicture}\quad
\end{tabular}
&
\begin{tabular}{l}
$\sigma_2(\includegraphics[width=0.4cm]{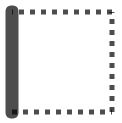})=
-\frac{1}{2}\left[
\delta\phi(\includegraphics[width=0.4cm]{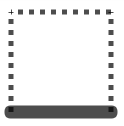})\cdot \delta\phi(\includegraphics[width=0.4cm]{square-l.png})+
\delta\phi(\includegraphics[width=0.4cm]{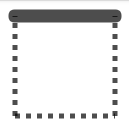})\cdot \delta\phi(\includegraphics[width=0.4cm]{square-l.png})
\right]$
\\[0.4cm]
$\sigma_2(\includegraphics[width=0.4cm]{square-d.png})=
-\frac{1}{2}\left[
\delta\phi(\includegraphics[width=0.4cm]{square-d.png})\cdot \delta\phi(\includegraphics[width=0.4cm]{square-d.png})-
\delta\phi(\includegraphics[width=0.4cm]{square-r.png})\cdot \delta\phi(
\hspace{-0.3cm}\begin{tabular}{c}\vspace{-0.3cm}
\includegraphics[width=0.4cm]{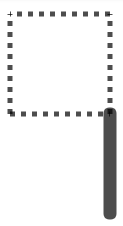}
\end{tabular}\hspace{-0.2cm})
\right]$
\end{tabular}
\end{tabular}
  \caption{Notation in Definition~\ref{def-Maxwell-stress} of discrete momentum flux. The square $uvv_+u_+$ is shown by dotted lines to the right. The edge \red{at} which a particular function is evaluated is 
  in bold. Cf.~\S\ref{ssec-quick}.}
  \label{fig-notation}
\end{figure}

\begin{corollary}[Momentum conservation for the electric field] \label{cor-networks-momentum-conservation} (Cf.~\cite[Example in \S8 \red{of Ch.~III}]{Dorodnitsyn-04}.) For each 
\remove{clarity}{} electric potential $\phi$  we have $\delta\sigma_1=\delta\sigma_2=0$ on each face not intersecting the grid boundary.
\end{corollary}

\edit{R12P3}{This corollary and all the other corollaries in~\S\ref{sec-examples} are particular cases of the results of~\S\ref{ssec-statements} and~\S\ref{ssec-general-connections}, and are directly deduced from those. We keep detailed proofs in~\S\ref{ssec-example-proofs} for the reader's reference.}

\begin{remark} \label{rem-networks-momentum}
  The function $\sigma_k$ is 
  the flux (given by Definition~\ref{def-flow}) of the energy-momentum tensor $T[\phi]=\delta\phi\times\delta\phi$ (given by Theorem~\ref{th-energy-conservation} for the Lagrangian $\mathcal{L}[\phi]=\frac{1}{2}\delta\phi\frown \delta\phi$).
  \remove{clarity}{}
\end{remark}

%

\subsubsection*{Approximation}

The basic network model indeed converges to a continuum one, as the grid becomes finer and finer.


The continuum model is a \emph{homogeneous conducting plate} defined as follows. Let $\mathrm{I}^2 $ be the unit square, $\vec{\mathrm{ n}}$ be the unit inner normal vector field on~$\partial \mathrm{I}^2 $ besides the corners, $*$ be the counterclockwise rotation through $\pi/2$ about the origin (the \emph{Hodge star}), $\deltaup_{kl}=\deltaup^l_k:=\begin{cases}
1, &\text{if }k=l;\\
0, &\text{if }k\ne l.
\end{cases}$

A \emph{source} $\mathrm{s}$ is a continuous function on $\partial \mathrm{I}^2 $.
The \emph{fields} $\vec{\mathrm{ j}}$,${\phiup}$,$\mathrm{F}$,$\mathrm{W} $,$\vec{\mathrm{ S}}$,$\vec{\mathrm{ L}}$,$\mathrm{P}$,$\mathcal{L}$,${\sigmaup}$
\emph{generated} by ${\mathrm{ s}}$ are continuous
scalar/vector/matrix fields on $\mathrm{I}^2 $, being  $\mathrm{C}^1$ and satisfying the following conditions apart~$\partial \mathrm{I}^2$:
\begin{align*}
-\nabla  {\phiup}&=\vec{\mathrm{ j}},&
\mathrm{W} &=-\nabla  {\phiup}\cdot \vec{\mathrm{ j}}, &
\vec{\mathrm{ S}}&=-\!*\!\nabla  {\phiup}\cdot \mathrm{F} ,
&
\mathcal{L}&=\tfrac{1}{2}(\nabla  {\phiup})^2 \quad\!\left(=\tfrac{1}{2}\mathrm{d}\phiup\lrcorner \mathrm{d}\phiup\right),
\\
*\nabla  \mathrm{F} &=\vec{\mathrm{  j}},&
\vec{\mathrm{ L}}&=\!*\vec{\mathrm{ j}}\cdot\mathrm{F} ,
&
\mathrm{P} &=\tfrac{1}{2}\mathrm{F} \cdot\mathrm{F} ,&
{\sigmaup}_{kl}&=
\frac{\partial{\phiup}}{\partial \mathrm{x}^k} \frac{\partial{\phiup}}{\partial \mathrm{x}^l}
-\frac{1}{2}\deltaup_{kl}(\nabla  {\phiup})^2,
\end{align*}
and the following boundary condition on $\partial \mathrm{I}^2 $ besides the corners:
$$
\vec{\mathrm{ j}}\cdot \vec{\mathrm{ n}}
=\mathrm{s}.
$$
In other words, ${\phiup}+i\mathrm{F} $ is an analytic function such that $\tfrac{\partial}{\partial\vec{\mathrm{ n}}}{\phiup}=-\mathrm{s}$; the other fields are expressions in it.

Let the unit square $\mathrm{I}^2 $ be dissected into $N^2$ equal squares. Given a source $s_N$, define the fields $j_N,\phi_N,F_N,W_N,S_N,L_N,P_N,\mathcal{L}_N,\sigma_N$ on the resulting grid literally as above on the grid of size $N\times N$.

\begin{remark}
  It would be somewhat more conceptual to modify the above Amp\`ere law for the resulting grid because the faces are not unit squares anymore. This leads just to \gram{the} normalization of the fields by powers of $N$. We avoid such modification for simplicity.

  \clarity{It would be more conceptual to write the Lagrangian as $\mathcal{L}=\tfrac{1}{2}\mathrm{d}\phiup\lrcorner \mathrm{d}\phiup-\mathrm{s}\lrcorner \phiup$ but the second term vanishes apart the boundary anyway.}

  \red{The} continuum model has more symmetries than the discrete one: e.g., $\mathcal{L}$ is rotational-invariant whereas $\mathcal{L}_N$ is not, at least in a naive sense; cf.~\cite[Definition~5.2.36]{Hydon-Mansfield-04}.
\end{remark}

Dissect each side of $\partial\mathrm{I}^2$ into $N+1$ (not $N$) equal segments called \emph{auxiliary segments}.
Write $a_N(x)\rightrightarrows b_N(x)$ for functions $a_N,b_N$ on a set $M_N$, if $\max_{x\in M_N}|a_N(x)- b_N(x)|\to 0$ as~$N\to\infty$.


\begin{theorem}[Approximation theorem] \label{th-networks-convergence}
  Let $\mathrm{s}\colon \partial \mathrm{I}^2 \to\mathbb{R}$
  be a continuous source with $\int_{\partial \mathrm{I}^2 }\mathrm{s}\,\mathrm{dl}=0$. Dissect $\mathrm{I}^2 $ into $N^2$ equal squares and define a discrete source $s_N$ on the resulting grid by the formula
  $$
  s_N(v):=\int_{v_-v_+}\mathrm{s}\,\mathrm{dl},
  $$
  where $v_-v_+\subset \partial I^2$ is the arc formed by $1$ or $2$ auxiliary segments containing a vertex $v\in\partial I^2$.
  Take continuous fields $\vec{\mathrm{ j}}$, ${\phiup}$,
  $\mathrm{F} $, $\mathrm{W} $, $\vec{\mathrm{ S}}$,
  $\vec{\mathrm{ L}}$, $\mathrm{P} $, $\mathcal{L}$, ${\sigmaup}$ and discrete ones $j_N$, $\phi_N$, $F_N$, $W_N$, $S_N$, $L_N$, $P_N$, $\mathcal{L}_N$, $\sigma_N=(\sigma_{N,1},\sigma_{N,2})$ generated by the sources. Assume that ${\phiup},\mathrm{F} $ and $\phi_N,F_N$ vanish at the center of $\mathrm{I}^2 $ and at one of the vertices or faces closest to the center respectively. Take $r>0$.
  Then on the set of all vertices $v$, edges $e$, faces $f$, edge-midpoints $e'$, and face-centers $f'$ at distance $\ge r$ from $\partial \mathrm{I}^2 $ we have:
  \begin{align*}
    \phi_N(v)   &\rightrightarrows
     {\phiup}(v), &
    Nj_N(e)     &\rightrightarrows
     N\int_e \vec{\mathrm{ j}}\cdot \mathrm{d}\vec{\mathrm{l}}, &
    N^2W_N(e')  &\rightrightarrows \mathrm{W} (e'), &
    N S_N(e'f') &\rightrightarrows
     2N\int_{e'f'} \vec{\mathrm{ S}}\cdot\mathrm{d}\vec{\mathrm{l}},
    \\
    N^2\mathcal{L}_N(v)   &\rightrightarrows \mathcal{L}(v), &
    F_N(f)      &\rightrightarrows
     N^2\int_f\mathrm{F} \,\mathrm{d\red{A}}, &
    P_N(e')     &\rightrightarrows \mathrm{P} (e'),  &
    NL_N(e'f')   &\rightrightarrows
     N\int_{e'f'} \vec{\mathrm{ L}}\cdot\mathrm{d}\vec{\mathrm{l}},\\[-0.9cm]
    \end{align*}
    $$
    N^2\sigma_{N,k}(e) \rightrightarrows N\int_e\left(\sigmaup_{k2}\,\mathrm{dx}^1
    -\sigmaup_{k1}\,\mathrm{dx}^2\right)
    \qquad\mbox{as }N\to\infty.
    $$
\end{theorem}


The theorem is essentially known; it is easily deduced from highly nontrivial known results in \S\ref{s:proofs}.




\subsection{Lattice electrodynamics}
\label{ssec-Electrodynamics}

A standard problem in electrodynamics is to find forces between given charges and currents. This is done in two steps: first the field generated by the charges and currents is computed, then --- the action of the field upon them. For a discretization, continuum spacetime is replaced by a $4$-dimensional grid.

\subsubsection*{Generation of the field by the current}

\edit{R12P3}{The discrete theory is nicely stated in terms of Definitions~\ref{def-cochain} and~\ref{def-boundary}, which we use hereafter.}

\begin{definition}\label{def-sharp}
The \emph{Minkowski sharp operator} $\#$ applied to a function $F$ on $k$-dimensional faces $f$ \edit{R12P2}{of the grid $I^d_N$, where $k\ne 0$,} is
$$
[\mathop{\#}\!F](f):=
\begin{cases}
(-1)^{k-1}F(f), &\text{if }f\parallel (1,\underbrace{0,\dots,0}_{d-1\text{ zeroes}}),\\
(-1)^{k}F(f), &\text{if }f\perp (1,0,\dots,0).
\end{cases}
$$
An \emph{electromagnetic \gram{vector potential} $A$ generated} by a current $j$ is a function on edges satisfying
\begin{itemize}
  \item\emph{The source equation}:
$-\partial\#\delta A=j$.
\end{itemize}
\end{definition}

\begin{remark}\label{rem-boundary}
We do not discuss conditions under which the \gram{vector potential} exists and is unique.

The operator $\#$ is new. It is a discrete \gram{analog} of raising all indices in the metric of signature $(+,-,...,-)$.
We use it instead of a discrete Hodge star \cite{Teixeira-13} to avoid working with the dual lattice, which would complicate the theory and its generalization to \edit{R12P2}{arbitrary complexes}.
\end{remark}

The following $3$ properties of an electromagnetic \gram{vector potential} $A$ generated by a current $j$ immediately follow from the well-known identities $\delta\delta=0$ and $\partial\partial=0$; cf.~\eqref{eq-Maxwell-1}:
\begin{itemize}
  \item \emph{Maxwell\red{'s} equations}: $\delta F=0$ and $-\partial\!\mathop{\#}\!F=j$, where $F:=\delta A$ is the \emph{electromagnetic field};
  \item \emph{Gauge invariance}: $A+\delta g$ is generated by the same current $j$ for any function $g$ on~vertices;
  \item \emph{Charge conservation}: $\partial j=0$, if there exists a \gram{vector potential} generated by the current $j$.
\end{itemize}


\begin{corollary} 
  \label{cor-Maxwell}
  An electromagnetic \gram{vector potential} $A$ is generated by a current $j$, if and only if $A$ is \edit{R1P4}{an extremal} of the functional $\mathcal{S}[A]=\epsilon\mathcal{L}[A]$, where
  $$
  \mathcal{L}[A]=-\tfrac{1}{2}\!\mathop{\#}\!\delta A \frown\delta A-j\frown A.
  $$
\end{corollary}

\begin{remark}
  Electrodynamics in linear nondispersive media is discretized analogously, only the Minkowski sharp operator is replaced by a linear operator depending on the media. 

  To convince the reader that lattice electrodynamics is a realistic model, let us informally sketch a network model for it \cite{Kron-44}. Set $d=4$. For each edge of the grid $I^{d-1}_N$, take an oscillatory circuit consisting of one (nonconstant) current source, one unit capacitor, and as many unit-transformer coils as there are faces containing the edge; see Figure~\ref{fig-network-model} to the bottom-left. Join the obtained circuits in the shape of the grid, join the transformer cores in the shape of the $1$-dimensional skeleton of the dual grid, \gram{and} join the capacitor dielectric cores in the shape of the $2$-dimensional skeleton of the dual grid. We get an electric, a magnetic, and a dielectric network coupled together; a part is shown in Figure~\ref{fig-network-model}. We conjecture that the integrals of appropriate currents and voltages over time intervals $[n,n+1]$, where $n\in\mathbb{Z}$, satisfy the discrete Maxwell equations above.
\end{remark}



\begin{figure}[htbp]
  \centering
\ctikzset{bipoles/length=0.32cm}
\begin{circuitikz}[scale=0.3]
\draw[color=brown]
(10,15) to[short, -] (10,16)
        to[C, -]     (0,11)
        to[short, -] (0,10)
(10,15) to[short, -] (5.5,12.75)
        to[short, -] (9.7,11)
        to[american inductor, -]     (9.7,10)
        to[short, -] (4.5,12.25)
(0,10)  to[american current source, -] (4.5,12.25)
(10,16) to[C, -]     (20,11)
        to[short, -] (20,10)
(10,15) to[short, -] (14,13)
        to[short, -] (10.3,11)
(10.3,10) to[american inductor, -]     (10.3,11)
(10.3,10) to[short, -]                 (15,12.5)
        to[american current source, -] (20,10)
(20,11) to[C, -]     (30,6)
        to[short, -] (30,5)
(20,10) to[american current source, -] (24,8)
        to[short, -] (20.3,6)
(20.3,5)to[american inductor, -]     (20.3,6)
(20.3,5)to[short, -]                 (25,7.5)
        to[short, -] (30,5)
(20,11) to[C, -]     (10,6)
        to[short, -] (10,5)
(15.5,7.75)to[american current source, -] (20,10)
(15.5,7.75)to[short, -] (19.7,6)
        to[american inductor, -]     (19.7,5)
        to[short, -] (14.5,7.25)
(15.5,7.75)to[short, -] (10.3,9.9)
(10.3,8.9) to[american inductor, -]  (10.3,9.9)
(10.3,8.9) to[short, -] (14.5,7.25)
(10,5)  to[short, -] (14.5,7.25)
(10,6)  to[C, -]     (20,1)
        to[short, -] (20,0)
(10,5)  to[short, -] (15,2.5)
        to[short, -] (19.7,4.9)
        to[american inductor, -]      (19.7,3.9)
        to[short, -]                 (16,2)
        to[american current source, -] (20,0)
(0,11)  to[C, -]     (10,6)
        to[short, -] (10,5)
(0,10)  to[american current source, -] (5,7.5)
        to[short, -] (9.7,9.9)
        to[american inductor, -]      (9.7,8.9)
        to[short, -]                  (6,7)
        to[short, -] (10,5)
(30,5) to[short, -] (30,6)
        to[C, -]     (20,1)
        to[short, -] (20,0)
(30,5) to[short, -] (25.5,2.75)
        to[short, -] (20.3,4.9)
(20.3,3.9) to[american inductor, -] (20.3,4.9)
(20.3,3.9) to[short, -] (24.5,2.25)
(20,0)  to[american current source, -] (24.5,2.25)
(10,13) to[short, -]     (10,6.6)
(20,8)  to[short, -]     (20,1.6)
(10,4.7)to[short, -]     (10,2.5)
        to[short, -]     (20,-2.5)
        to[short, -]     (20,-0.3)
(0,0) to[short, -] (1,0)
      to[short, -] (1,1)
      to[C, -]     (5,1)
      to[short, -] (5,0)
      to[short, -] (6,0)
(1,0) to[short, -] (1,-1)
      to[short, -] (2,-1)
      to[short, -] (2,-0.5)
      to[american inductor, -]     (3,-0.5)
      to[short, -] (3,-1.5)
      to[american inductor, -] (2,-1.5)
      to[short, -] (2,-1)
(3,-1)to[american current source, -] (5,-1)
      to[short, -] (5,0)
(2,-0.2)
      to[short, -]     (3,-0.2)
(2,-1.8)
      to[short, -]     (3,-1.8)
;\end{circuitikz}
  \caption{A network model for lattice electrodynamics; cf.~\cite{Kron-44}}\label{fig-network-model}
\end{figure}

\subsubsection*{Action of the field on the current}


The field acts on the current by the Lorenz force, which we are going to discretize now. The rest of \S\ref{ssec-Electrodynamics} contains completely new notions and results;\remove{R12P2}{(except the \gram{cross product})} cf.~\cite{Bossavit-03} \edit{P12P1}{and~\cite{Forman-02}.}

\edit{R12P2}{We start with an informal motivation.}
\move{R12P2}{}The formula \red{for the Lorenz force $L$ in \S\ref{ssec-Networks}} involves the product of the values of \red{fields} at \red{edges and faces. Thus} it is reasonable to view it as a ``projection'' of a more fundamental quantity defined on the Cartesian square of the \red{grid. More precisely,}
\move{R12P2}{} 
the set of faces of \clarity{the Cartesian square} is naturally mapped to the set of faces of the doubling: to a face $e\times f$ assign the face of the doubling with the center at the midpoint of the segment joining the centers of $e$ and $f$.
\remove{clarity}{}
Up to sign and factor $1/2$, the fields $W,S,L,P$ from~\S\ref{ssec-Networks} are ``induced'' by the latter map from \red{the \emph{cross products}} $j\times\delta\phi$, $F\times\delta\phi$, $j\times F$, $F\times F$ respectively. \red{This naturally leads to Definition~\ref{def-tensor}, which we use hereafter.} 
These heuristic fields \red{$W,S,L,P$} are now replaced by tensors.

\begin{definition}\label{def-electrodynamics-tensor}
Let $A$ be a \gram{vector potential} generated by a current $j$, and $F=\delta A$. The \emph{Lorentz force} is the type $(0,1)$ tensor $L=j\times F$. It has support on faces $e\times f\subset I^d_N\times I^d_N$ such that $\dim e\!\!=\!\!1$, $\dim f\!\!=\!\!2$.

The \emph{energy-momentum tensor}, or \emph{stress-energy tensor}, of the \gram{electromagnetic} field (respectively, of both the field and the current) is the type $(1,1)$ tensor $T'=-\#F\times F$ (respectively, $T=-\#F\times F-j\times A$).
The tensor $T'$ has support on $4$-dimensional faces $e\times f\subset I^d_N\times I^d_N$ such that $\dim e=\dim f=2$.
\end{definition}

An immediate consequence of these definitions, Maxwell's equations, and charge conservation is
\begin{itemize}
\item \emph{Energy and momentum conservation}: $\partial T'=L$ and $\partial T=0$.
\end{itemize}

\begin{remark} The latter is a particular case of Theorem~\ref{th-energy-conservation} for the Lagrangian from Corollary~\ref{cor-Maxwell}.
\end{remark}

\begin{remark} \label{rem-Forman} \clarity{Let us give general comments on the discretization of tensors in Definition~\ref{def-tensor}.}

\move{clarity}{} In contrast to continuum theory,
  type $(0,1)$ tensors are \emph{not} $1$-dimensional fields.


  \edit{R1P5}{Our tensors can be alternatively viewed as functions on faces of $I^d_N$, with two arguments. We prefer working with functions on $I^d_N\times I^d_N$ to get elegant expressions with chain cross product such as~\eqref{eq-energy-conservation}.}

  Although $I^d_N\times I^d_N$ is naturally identified with $I_N^{2d}$, the boundary operator on tensors is \emph{not} 
  the boundary operator on $I_N^{2d}$. To avoid confusion, we distinguish between $I^d_N\times I^d_N$ and $I^{2d}_N$ below.

  A type $(q,1)$ tensor can be equivalently defined  as an element of $C^{d+q-1}(I^d_N\times I^{d*}_N;\mathbb{R})$, where $I^{d*}_N$ is the dual grid. Then the boundary operator on tensors is exactly the boundary operator on $I^d_N\times I^{d*}_N$. We avoid working with dual grids for simplicity and easier generalization to arbitrary \edit{R12P2}{complexes}.

  \edit{R12P1}{A type $(q,1)$ tensor can also be viewed a collection of linear maps $C^{k}(I^d_N;\mathbb{R})\to C^{k+q-1}(I^d_N;\mathbb{R})$ for all $k=0,\dots,d$: the values of the tensor at faces $e\times f$ comprise the coefficients of those maps. This shows that our tensors generalize Forman's forms \cite{Forman-02}, with the same role of doubling.}

  It would be more conceptual to restrict the domain of a tensor to a ``neighborhood of the diagonal'' in $I^d_N\times I^d_N$: \move{R12P1}{the values at the other faces do not contribute to integration.} Type $(0,1)$ tensors can be restricted to the set of faces $e\times f$ such that $e\subset f$; \edit{R12P1}{this makes the definition completely equivalent to Forman's discretization of $1$-forms. Concerning type $(1,1)$ tensors, their natural domain is faces $e\times f$ such that either $e=f$ or $e\cap f$ is a codimension $1$ face in both $e$ and $f$. This is not equivalent to Forman's discretization (type $(1,1)$ tensors are not $0$-forms, as expected), but is a natural generalization.} We avoid such restriction for simplicity, \red{especially in computations involving cross products.}
\end{remark}



\begin{remark} \label{rem-tensor}
\move{clarity}{} \red{(For specialists.)} Let us clarify the relation \red{of our tensor calculus} to continuum theory.
The set of type $(q,1)$ tensors is \clarity{naturally isomorphic to} $\bigoplus_{p=0}^d C_{p}({I}^d_N;\mathbb{R}) \otimes_\mathbb{R} C^{p-q+1}({I}^d_N;\mathbb{R})$. Thus it discretizes the space $\bigoplus_{p=0}^d \Omega^{p}(\mathrm{I}^d)^* \otimes\Omega^{p-q+1}(\mathrm{I}^d)$ rather than the space $\mathrm{T}^{q}_{1}(\mathrm{I}^d)$ of continuum type $(q,1)$ tensors. (Here $\Omega^{p}(\mathrm{I}^d)$ denotes the set of $\mathrm{C}^\infty$ $p$-forms on the unit hypercube $\mathrm{I}^d$ and $\otimes$ denotes the tensor product over $\Omega^{0}(\mathrm{I}^d)$). But the former space is mapped to the latter by the `contraction' map
  $$
  \mathrm{T}^{m_1\dots m_p}_{n_1\dots n_{p-q+1}}\mapsto
    \begin{cases}
    \mathrm{T}^{m_1\dots m_p}_{km_1\dots m_{p}},
    &\text{if }q=0;\\
    \mathrm{T}^{lm_2\dots m_p}_{km_2\dots m_{p}}-
    \tfrac{1}{2p}\deltaup^l_k \mathrm{T}^{m_1\dots m_p}_{m_1\dots m_{p}},
    &\text{if }q=1, p>0.
    \end{cases}
  $$
  (Summation over repeating indices is understood.) 
  Since no discretization of the image is available (at least for $q=1$), the discretization of the domain is proclaimed to be space of type $(q,1)$ tensors.
  Here the role of the $\deltaup^l_k$-term is the same as in the Einstein tensor: it makes the `contraction' map commute with certain codifferentials when $\mathrm{T}$ has certain symmetry properties (\clarity{namely}, $\sharp\mathrm{T}$ is symmetric wrt interchanging $m_j$ and $n_j$ but antisymmetric
  wrt interchanging $m_i$ and $m_j$):
  $$
\begin{CD}
\frac{\partial\mathcal{L}}{\partial(\mathrm{d}\phiup)}\otimes \mathrm{d}\phiup +\frac{\partial\mathcal{L}}{\partial\phiup}\otimes \phiup
@. \quad\in\quad @. \bigoplus_{p=0}^d\Omega^{p}(\mathrm{I}^d)^*
\otimes\Omega^{p}(\mathrm{I}^d)
@>\mathrm{d}^*\otimes \mathrm{id}+\mathrm{id}\otimes \mathrm{d}>>
\bigoplus_{p=0}^{d-1}\Omega^{p}(\mathrm{I}^d)^*
\otimes\Omega^{p+1}(\mathrm{I}^d)\\
@VVV @. @V\text{`contraction'}VV @V\text{`contraction'}VV\\
\mathrm{T}^l_k @.\quad\in\quad @. \mathrm{T}^{1}_{1}(\mathrm{I}^d)
@>\text{divergence}>>
\mathrm{T}^{0}_{1}(\mathrm{I}^d).
\end{CD}
$$
  Similarly, \eqref{eq-energy-conservation} discretizes
  $\frac{\partial\mathcal{L}}{\partial(\mathrm{d}\phiup)}\otimes \mathrm{d}\phiup +\frac{\partial\mathcal{L}}{\partial\phiup}\otimes \phiup$
  rather than the continuum energy-momentum tensor $\mathrm{T}^l_k$,
  but the former is usually taken to the latter by the `contraction' map. Here $\left(\frac{\partial\mathcal{L}}{\partial(\mathrm{d}\phiup)}\right)
  ^{m_1\dots m_p}:=\frac{\partial\mathcal{L}}{\partial(\mathrm{d}\phiup)_{m_1\dots m_p}}$.
  The former is conserved (i.e. taken to $0$ by $\mathrm{d}^*\otimes \mathrm{id}+\mathrm{id}\otimes \mathrm{d}$) regardless of  symmetries of $\mathcal{L}[\phiup]$.

\clarity{In particular}, $L$ and $T'$ discretize the tensors $\mathrm{j}^{l}\mathrm{F}_{kn}$ and $-\mathrm{F}^{lm}\mathrm{F}_{kn}$, but the latter two are taken to the continuum Lorenz force and energy-momentum tensor by the `contraction' map\red{s}. \remove{clarity}{} \move{clarity}{} In contrast to $T'$, the tensor $T$ has \emph{no} conserved continuum analogue.

The formula for the discrete energy-momentum tensor $T'$ is even simpler than the continuum analogue. This is achieved at the cost of rather subtle \edit{R2P1}{Definition~\ref{def-flow}} of discrete tensor integration.
\end{remark}

\subsubsection*{Integral conservation laws}

\remove{R2P1}{}
To compare \red{discrete} tensors with their continuum analogues, \edit{R2P1}{we need their integration. This naturally leads to Definition~\ref{def-flow}, which we use in the rest of~\S\ref{ssec-Electrodynamics}.
Actually, we have already applied it in the particular cases $d=3$, $k=0$, $T=-\#F\times F$ and $d=2$, $k=1,2$, $T=\delta\phi\times\delta\phi$ in~\S\ref{ssec-quick} and Definition~\ref{def-Maxwell-stress} respectively. A more general setup when this construction is well-applicable to energy-momentum tensor~\eqref{eq-energy-conservation} 
is a \emph{free field}.}

\begin{theorem}[Integral energy-momentum conservation for a free field] \label{cor-free} 
If the Lagrangian is $\mathcal{L}[\phi]=-\red{\frac{1}{2}}\mathop{\#}\!\delta \phi \frown\delta \phi-\red{\frac{1}{2}}m^2\phi\frown \phi$, \clarity{where $m\ge 0$, and $\phi\in C^k({I}^d_N;\mathbb{R})$ satisfies Euler-Lagrange equation~\eqref{eq-Euler-Lagrange} on all non-boundary faces,
then
for each $d$-dimensional face $g$ disjoint with $\partial{I}^d_N$ and each $0\le l<d$ 
the flux of energy-momentum
tensor~\eqref{eq-energy-conservation} across $\partial g$ vanishes, i.e., $\langle T[\phi],\partial g\rangle_l=0$.}
\end{theorem}

\clarity{Here we have dropped the Euler--Lagrange equations on the boundary, otherwise the system degenerates; the boundary faces do not contribute to the tensor flux anyway.}

\red{Electrodynamics (without currents) is the particular case of a free field theory with $m=0$, $k=1$. Further specification to $d=3$ gives~\eqref{eq-discrete-Poynting}, where the function $T(h)$ on faces $h$ is actually $\langle T',h \rangle_0$.
The case  $m=0$, $k=1$, $d=2$ 
was} established in~\cite[\red{Example in \S8 of Ch.~III}]{Dorodnitsyn-04} by a different method.

\begin{remark}\label{rem-energy}
  \clarity{Let us give general comments on the discretization of tensor flux in Definition~\ref{def-flow}.}

  There are many other ways to define a tensor flux; we have chosen the simplest one.

  Our definition has the following informal motivation. Values of a tensor are ``sitting'' on the faces of the doubling; see the paragraph
  \edit{RP12P3}{before Definition~\ref{def-electrodynamics-tensor}}. The flux across a hyperface is then the sum of these values over the faces adjacent to the hyperface from \gram{the} appropriate ``side''.

  For nonconserved tensors, an \gram{analog} of the Stokes formula holds; see Proposition~\ref{prop-Stokes}.

  \remove{clarity}{}

  Unlike continuum theory, the $0$-th component of the flux of the energy-momentum tensor $T'$ (see Definition~\ref{def-electrodynamics-tensor}) across a hyperface $h\perp (1,0,\dots,0)$ is not necessarily positive, thus cannot be interpreted as energy density. 
  This is a \gram{higher-order} effect with respect to the discretization step~$1/N$.

  We use the notation $\langle T,\red{h}\rangle_k$, with literally the same definition, even if $T$ is not partially symmetric. This makes no sense in \gram{a} discrete setup but is useful for the continuum limit.
  \remove{R12P3}{}

  The energy-momentum tensor $T$ of both the field and the current (see Definition~\ref{def-electrodynamics-tensor}) is not partially symmetric. In a sense, it still approximates some continuum tensor, but the latter is not conserved. We know neither an integral conservation law nor a conserved continuum \gram{analog} for $T$.

  The energy-momentum tensor $T'$ is symmetric in a sense (after ``raising an index''). In particular, we shall see that it approximates the symmetric Belinfante--Rosenfeld energy-momentum tensor rather than the nonsymmetric canonical energy-momentum tensor. In other field theories, e.g., for the Dirac field, the discrete energy-momentum tensor approximates the nonsymmetric canonical energy-momentum tensor rather than the Belinfante--Rosenfeld one; see \ifarxiv{Proposition~\ref{th-Dirac-approximation}.}{ \edit{R12P3}{\cite[Proposition~A.3.6]{Skopenkov-preprint}.}} 
\end{remark}

Let us illustrate analogy between tensor~\eqref{eq-energy-conservation} and the continuum \emph{canonical energy-momentum tensor}
$$
\mathrm{T}_k{}^l=\frac{\partial\mathcal{L}}
{\partial(\partial\phiup/\partial \mathrm{x}_l)}
\frac{\partial\phiup}{\partial \mathrm{x}_k}-
\deltaup^l_k\mathcal{L}.
$$

\begin{proposition}\label{prop-quadratic-lagrangian}
\remove{clarity}{}
Let a local Lagrangian $\mathcal{L}\colon C^0({I}^d_N;\mathbb{R})\to C_0({I}^d_N;\mathbb{R})$ be homogeneous quadratic in $\phi$ and $\delta\phi$. Let $\phi$ be a $0$-dimensional field (not necessarily \edit{R1P4}{an extremal}) and $T\red{[\phi]}$ be the \remove{clarity}{} 
tensor
(not necessarily partially symmetric) given by~\eqref{eq-energy-conservation}.
Then for each $0\le k,l<d$ and each hyperface $h\perp \mathrm{e}_l$ having maximal vertex $v$ and disjoint with the grid boundary we have
\begin{align*}
(-1)^l\langle T\red{[\phi]},h\rangle_k&=
\frac{1}{2}\left(
\frac{\partial\mathcal{L}[\phi]}{\partial(\delta\phi)}
(v+\mathrm{e}_l)+
\frac{\partial\mathcal{L}[\phi]}{\partial(\delta\phi)}
(v+\mathrm{e}_l-2\mathrm{e}_k)
\right)
\delta\phi(v-\mathrm{e}_k)
-\deltaup^l_k\mathcal{L}[\phi](v).
\end{align*}
\end{proposition}

\subsubsection*{Approximation}

The discrete energy-momentum tensor $T'$ indeed approximates the continuum one, as we show now.
In continuum theory, an \emph{electromagnetic field} is a continuous antisymmetric matrix field $\mathrm{F} _{mn}$ on the unit hypercube $\mathrm{I}^d$. The \emph{(Belinfante--Rosenfeld) energy-momentum tensor} of the field (for the metric of signature $(+,-,...,-)$) is the matrix field
$$
\mathrm{T}^l_k=-\mathrm{F}^{lm}\mathrm{F}_{km}
+\tfrac{1}{4}\deltaup^l_k\mathrm{F}^{mn}\mathrm{F}_{mn},
$$
where summation over repeating indices is understood and
$
\mathrm{F}^{mn}:=
\begin{cases}
  -\mathrm{F}_{mn}, & \mbox{if $m=0$ or $n=0$}; \\
  \mathrm{F}_{mn}, & \mbox{if $m\ne 0$ and $n\ne 0$}.
\end{cases}
$

Let $\mathrm{I}^d$ be dissected into $N^d$ equal hypercubes. Given an arbitrary discrete $2$-dimensional field $F$, define the energy-momentum tensor $T'=-\#F\times F$ on the resulting grid literally as on the grid $I^d_N$.

\begin{remark}\label{rem-renormalization}
  It is somewhat more natural to modify the definition of the operator $\#$ by the factor $N^{2k-d}$ because the faces are not unit hypercubes anymore. This leads just to \gram{the} normalization of the energy-momentum tensor $T'$ by a power of $N$. We avoid such modification for simplicity.
\end{remark}

\begin{proposition}[Approximation property] \label{th-Maxwell-approximation}
Let $\mathrm{F}_{mn}$ be a continuous electromagnetic field on $\mathrm{I}^d$. Dissect $\mathrm{I}^d$ into $N^d$ equal hypercubes and define a discrete $2$-dimensional field $F_N$ on faces $f$ of the resulting grid
by the formula 
$$
F_N(f):=\mathrm{F}_{mn}(\max f),
$$
where the integers $m<n$ are determined by the conditions $\mathrm{e}_m,\mathrm{e}_n\parallel f$.
Let $\mathrm{T}^l_k$ and $T'_N=-\# F_N\times F_N$ be the continuous and discrete energy-momentum tensor respectively. Take $0\le k,l< d$.
Then on the set of all hyperfaces $h\perp \mathrm{e}_l$ not intersecting $\partial \mathrm{I}^d$ we have (under the notation before Theorem~\ref{th-networks-convergence})
$$
(-1)^l \langle T'_N,h\rangle_k\rightrightarrows \mathrm{T}_k^l(\max h) \qquad\mbox{as }N\to\infty.
$$
\end{proposition}

\begin{remark} \remove{clarity}{}
Here 
$\mathrm{F}_{mn}$ and $F_N$ do \emph{not} necessarily satisfy Maxwell\red{'s} equations (and typically $F_N$ cannot, even if $\mathrm{F}_{mn}$ does). \gram{The} approximation of a smooth solution of Maxwell\red{'s} equations by discrete ones, a standard question of computational electrodynamics, is not discussed in the paper.
\end{remark}

\subsection{Lattice gauge theory}
\label{ssec-gauge}

Classical gauge theory generalizes electrodynamics. It is a basis for quantum gauge theory describing all known interactions except gravity. The idea is simple, as shown by the following toy model; cf.\cite{Maldacena-16}.

\subsubsection*{Toy model}

\begin{wrapfigure}{r}{4cm}
  \vspace{-1.3cm}
  \centering
  \includegraphics[width=4cm]{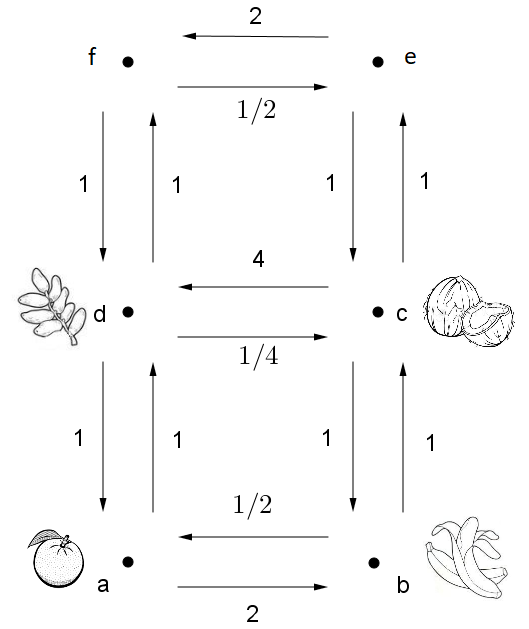}
  \caption{
  Gauge theory on a $1\times 2$ grid}\label{fig-toy}
  \vspace{-0.6cm}
\end{wrapfigure}
Several cities are 
\clarity{joined}
by roads in the shape of an $M\times N$ grid; see Fig.~\ref{fig-toy}. Each city has its own type of goods in an unlimited quantity. E.g., city $a$ has apples and city $b$ has bananas. For two neighboring cities $a$ and $b$ an exchange rate $U(ab)>0$ is fixed, e.g., $2$ banana for an apple. The rate is symmetric, i.e., $U(ba)=U(ab)^{-1}$: one gets back an apple for $2$ banana.

A cunning citizen can travel and exchange along a square $abcd$ to multiply his initial amount of goods by a factor of $U(ab)U(bc)U(cd)U(da)$.  The total speculation profit is measured by the quantity
$$
\mathcal{S}[U]:=\sum_{\text{all faces } abcd}\log^2(U(ab)U(bc)U(cd)U(da)).
$$
Here $\log^2(x)$ is chosen as a function vanishing at $x=1$ and positive for $x\ne 1$.

The king can set exchange rates except those on the boundary of the grid. He sets them to minimize
the quantity $\mathcal{S}[U]$. The resulting collection of rates is \edit{R1P4}{\emph{a stationary}} \emph{Abelian gauge group field}.



A \edit{R1P4}{stationary} gauge group field is far from being unique. For an interior city, one can change the units
,\remove{clarity}{} e.g., exchange dozens of apples instead of single ones. Such \emph{gauge transformation} multiplies the rates for all the roads starting from the city by the same value but preserves $\mathcal{S}[U]$.

A similar model on a $d$-dimensional grid (with an additional minus sign for each summand in $\mathcal{S}[U]$ such that $abcd$ is parallel to $(1,0,\dots,0)$) is equivalent to lattice electrodynamics discussed in \S\ref{ssec-Electrodynamics}. This follows from Corollary~\ref{cor-Maxwell}, if one sets $A(ab)=\log U(ab)$ and $j=0$; see also Remark~\ref{ex-electrodynamics-vs-gauge}.


\subsubsection*{Currents}

Now modify the model by introducing production of goods. For each pair of neighboring cities $a$ and $b$ fix a production rate $j(ab)\ge 0$: e.g., if $a$ has apples and $b$ has jam, then one produces $j(ab)$ units of jam from one apple. The rate is not at all symmetric: one cannot produce apples from jam. Assume that production always goes in the direction of the coordinate axes.

There is a new way to profit: producing jam and exchanging back to apples, one multiplies the initial amount of apples by $j(ab)U(ba)$. The total profit is now measured by the quantity $S[U,j]=S[U]+\sum_{ab}(j(ab)U(ba)-1)$.
A collection of rates $U$ minimizing $S[U,j]$ for fixed $j$ is called \emph{generated} by $j$. These rates may not exist, and the total profit can be negative.

These rates satisfy the \emph{conservation law}\/
$-j(\mathbf{1})U(\mathbf{1})^{-1}-j(\mathbf{2})U(\mathbf{2})^{-1}
+U(\mathbf{3})^{-1}j(\mathbf{3})+U(\mathbf{4})^{-1}j(\mathbf{4})=0$ for each interior city $v$, where we use the notation from Figure~\ref{fig-grid} to the middle (this law is a version of Corollary~\ref{cor-gauge}). This is a ``gauge-invariant'' equation, which coincides with the usual charge conservation $\partial j=0$ in the case when $U=1$ identically.

\subsubsection*{Non-Abelian gauge theory}

In non-Abelian gauge theory the goods become vectors and the rates become matrices. To catch the idea, one can start with the case when $d=2$, $n=1$, $G=\{g\in\mathbb{C}:|g|=1\}$, and drop all $\#$-operators.

\begin{definition} \label{def-gauge} Denote by $\mathbb{C}^{m\times n}$ the set of matrices with complex entries having $m$ rows and $n$ columns. For $u\in \mathbb{C}^{m\times n}$ denote by $u^*\in\mathbb{C}^{n\times m}$ the conjugate transpose matrix.

A \emph{gauge group} $G$ is a Lie group represented by unitary transformations of $\mathbb{C}^n$.
A \emph{gauge group field}
$U$ and a \emph{covariant current} $j$  are functions on edges of \edit{R12P2}{$M$} assuming values in $G$ and $\mathbb{C}^{n\times n}$ respectively.

\remove{clarity}{}
The \emph{operator of parallel transport} along \clarity{a simple} oriented path $\pi$ going along the edges is 
$$
U(\pi):=\prod_{e}U(e)^{\langle e,\pi\rangle},
$$
where the product is over all the edges $e$ of the path $\pi$,
and $\langle e,\pi\rangle=+1$ \clarity{if $e$ is cooriented with $\pi$, and $-1$ otherwise.}
In particular, the trace $\mathrm{Tr}\,U(\partial f)$ is a well-defined complex-valued function on  $2$-dimensional faces $f$.
A gauge group field $U$ \emph{generated} by a covariant current $j$ is
\edit{R1P4}{an extremal}
for the functional (for fixed $j$)
\begin{equation}\label{eq-def-gauge}
\mathcal{S}[U]=\sum_{\text{faces } f} \#\left( \mathrm{Re}\,\mathrm{Tr}\,U(\partial f)-n\right)-\sum_{\text{edges } e}\mathrm{Re}\,\mathrm{Tr}\,[j^*(e)U(e)].
\end{equation}
\end{definition}


Since $\mathcal{S}[U]$ is a continuous function on a compact set, we get the following existence theorem.

\begin{proposition}
For each covariant current there exists a gauge group field generated by it.
\end{proposition}

Now we state the Yang--Mills equation (necessary and sufficient for $U$ to be generated by $j$) and a conservation law.
This is a new Corollary~\ref{cor-gauge} extending \cite[Eq.~(4.15)]{Dimakis-etal-94}. It involves projection to certain tangent space of the Lie group~$G$. In gauge theory the role of the (co)boundary is played by the covariant (co)boundary, which is a ``gauge covariant'' operator equal the (co)boundary for $U=1$.

\begin{definition}
Fix a gauge group field $U$. Let $j$ be a $\mathbb{C}^{n\times n}$-valued function on edges. Its \emph{covariant boundary} $D_{A}^* j$ is a $\mathbb{C}^{n\times n}$-valued function on vertices $v$ given by
\begin{equation}\label{eq-def-boundary-particular}
[D_{A}^* j](v)=\sum_{e\textrm{ ending at }v}U(e)^{-1}j(e)-\sum_{e\textrm{ starting at }v}j(e)U(e)^{-1}.
\end{equation}

Denote by $D_{A}^* \#F$ the $\mathbb{C}^{n\times n}$-valued function on edges $e$ given by
\begin{equation}\label{eq-def-unmotivated}
[D_{A}^* \# F](e)=
\sum_{\textrm{2-faces }f \supset e}\# (U(e)-U(\partial f-e)),
\end{equation}
where $\partial f-e$ is the path starting at the vertex $\min e$, consisting of the $3$ edges of $\partial f-e$, and ending at $\max e$. E.g., in Figure~\ref{fig-toy} we have $[D_{A}^* \# F](dc)=U(dabc)+U(dfec)-2U(dc)$.

So far the notations $D_{A}^* j$ and $D_{A}^* \# F$ should be viewed as indivisible. Separate conceptual definitions of $A$, $F$, $D_A^*$ \edit{R12P3}{are given in Definitions~\ref{def-connection} and~\ref{def-covar-boundary-gauge},}
where~\eqref{eq-def-boundary-particular}--\eqref{eq-def-unmotivated} become easy propositions.
\end{definition}

\begin{definition} \label{def-projection}
The \emph{scalar product} of $u,v\in\mathbb{C}^{n\times n}$ is
$
\langle u,v\rangle:=\mathrm{Re}\,\mathrm{Tr}\,[u^*v].
$
Let $T_u G\subset \mathbb{C}^{n\times n}$ 
be the linear subspace parallel to the tangent subspace to $G$ at a point $u\in G$.
Let
$
\mathrm{Pr}_{T_u G}\colon\mathbb{C}^{n\times n}\to T_u G
$
be the orthogonal projection
and $\mathrm{Pr}_{T_U G}j$ be the function on edges $e$
given by
$[\mathrm{Pr}_{T_U G} j](e)=\mathrm{Pr}_{T_{U(e)}G} j(e).$ A covariant current $j$ is \emph{conserved}, if~$D^*_{A}\mathrm{Pr}_{T_U G}\,j=0$.
\end{definition}

\begin{corollary} \label{cor-gauge}
  A gauge field $U$ generated by a covariant current $j$ satisfies the following equations:
\begin{itemize}
  \item \emph{the Yang--Mills equation:} $-\mathrm{Pr}_{T_U G}\,D^*_A\# F=\mathrm{Pr}_{T_U G}\,j$;
  \item \emph{Charge conservation law:}
  $D^*_A\mathrm{Pr}_{T_U G}\,j=0$.
\end{itemize}
\end{corollary}

\begin{remark} \label{ex-electrodynamics-vs-gauge}
  The latter form of change conservation, different from the usual $\partial j=0$, reflects the fact that non-Abelian gauge fields are themselves charged. In contrast to continuum theory, this remains true even if $G$ is Abelian (the \clarity{deep} reason is that the \gram{cup product} is non-Abelian).\remove{R12P3}{}
  Also, $D^*_A j\ne 0$ in general: e.g., if $j$ vanishes on all edges except one, then $D^*_A j\ne 0$ whatever $U$ is.

  \clarity{But} for the Abelian group $G=\{e^{i\phi}:\phi\in\mathbb{R}\}$ and $d\!=\!2$ the action can be modified so that charge conservation returns to the form $\partial j=0$ (here $j\in C_1(I^2_N;\mathbb{R})$ is not a covariant current anymore):
  $$
  \mathcal{S}^\mathrm{Ab}[U]=
    -\tfrac{1}{2}\sum_{\textrm{\red{$2$-}faces }f}\arccos^2\mathrm{Re}\,\#U(\partial f)+i\sum_{\textrm{edges }e} j(e)\log U(e).
  $$
  The range of $U$ must be restricted to $\{e^{i\phi}:-\pi/4<\phi<\pi/4\}$ to keep the action single-valued and differentiable.
  The resulting theory is equivalent to lattice electrodynamics of \S\ref{ssec-Electrodynamics}, also with restricted range, because
  $
  \mathcal{S}^\mathrm{Ab}[e^{i\phi}]=
  \epsilon\,\left[
  -\frac{1}{2}\#\delta\phi\frown \delta\phi-j\frown \phi\right]$ for
  $\phi\in C^1({I}^2_N;\mathbb{R})$ with $|\phi|<\pi/4$.
\end{remark}


\subsubsection*{Connection and curvature}


\begin{definition} \label{def-cup0}
  Let $g$ and $\phi$ be $G$- and $\mathbb{C}^{n\times n}$-valued functions on vertices and $k$-faces respectively. The \emph{gauge transformation} of $\phi$ by $g$ is the function $g^*\smile\phi\smile g$ on $k$-faces $f$ given by (cf.~Table~\ref{tab-products})
\begin{equation*}
  [g^*\smile\phi\smile g](f) :=g^*(\min f)\,\phi(f)\,g(\max f).
\end{equation*}
\end{definition}

%



\begin{corollary}[Gauge invariance] \label{prop-gauge-invariance-Wilson}
  Each simultaneous gauge transformation of $U$ and $j$ by the same element $g$ preserves $\mathcal{S}[U]$. If $U$ is generated by $j$, then 
  $g^*\smile U\smile g$ is generated by  $g^*\smile j\smile g$.
\end{corollary}

\begin{definition} \label{def-cup} \label{def-connection} \label{def-curvature}
The \emph{unit gauge group field} $1$ equals the unit $n\times n$ matrix at each edge. For a gauge group field~$U$, the \emph{connection} (or \emph{gauge potential}) is the $\mathbb{C}^{n\times n}$-valued function $A[U]=U-1$. The \emph{curvature} (or \emph{field strength}) is the $\mathbb{C}^{n\times n}$-valued function on 
\clarity{$2$-dimensional} faces given by
\begin{equation}\label{eq-def-curvature}
  F[U](abcd):=U(ab)U(bc)-U(ad)U(dc)
\end{equation}
for each face $abcd$ with the vertices listed counterclockwise starting from the minimal one; see~Figure~\ref{fig-grid}.
\end{definition}

\begin{remark}\label{rem-curvature}\label{prop-flatness-criterion}
On a grid, a gauge group field $U$ is a gauge transformation of the unit gauge group field, if and only if the curvature $F[U]$ vanishes (this is proved by a standard ``homological'' argument.)

In contrast to continuum theory, the connection and curvature assume values \emph{not} in the Lie algebra of the Lie group $G$ but in certain other subsets of $\mathbb{C}^{n\times n}$ approximating the Lie algebra in a sense. The fields $A$ and $F$ from \S\ref{ssec-Networks}--\ref{ssec-Electrodynamics} are neither connection nor curvature for \emph{no} gauge group field.

\remove{R12P2}{}
\clarity{Similarly} to Proposition~\ref{th-Maxwell-approximation}, the tensor $-\mathrm{Re}\,\mathrm{Tr}\,\left[\#F^*\times F\right]$ approximates the continuum Belinfante--Rosenfeld energy-momentum tensor. But the former is not conserved and \red{not even} gauge invariant.
\end{remark}

\begin{proposition}
\label{l-Wilson} 
There is the following expression for action~\eqref{eq-def-gauge}:
$$
\mathcal{S}[U]=\epsilon\,\mathrm{Re}\,\mathrm{Tr}\,\left[
-\tfrac{1}{2}\#F^*\frown F-j^*\frown U\right].
$$
\end{proposition}

\clarity{The latter} is the one given by the algorithm from~\S\ref{ssec-main-tools} up to an additive constant; \red{see Table~\ref{tab-discretization}.}

\edit{R12P3}{}
\begin{table}[t]
\caption{Products of \red{$0$- and $1$-}(co)chains ($ab$ denotes an edge with $a<b$; \red{the sums are over edges})}
\label{tab-products}
\begin{tabular}{ccc}
\toprule
$\dim\phi=1$, $\dim\psi=0$ & $\dim\phi=0$, $\dim\psi=1$ & $\dim\phi=\dim\psi=1$\\[0.1cm]
\midrule
$
\begin{matrix}
[\phi\smile \psi](ab) = \phi(ab)\,\psi(b)\\
[\phi\frown \psi](ab) = \phi(ab)\,\psi(a)\\
\phi\spleen \psi = 0 
\end{matrix}
$
&
$
\begin{matrix}
[\phi\smile \psi](ab) = \phi(a)\,\psi(ab)\\
\phi\frown \psi =0 
\\
[\phi\spleen \psi](ab) = \phi(b)\,\psi(ab)
\end{matrix}
$
&
$
\begin{matrix}
\phi\smile \psi  
\text{ is defined in Figure~\ref{fig-grid}}\qquad\\
[\phi\frown \psi](b) =\sum\nolimits_{
ab: a<b} \phi(ab)\,\psi(ab)\\
[\phi\spleen \psi](b) = \sum\nolimits_{
bc: c>b} \phi(bc)\,\psi(bc)
\end{matrix}
$
\\
\bottomrule
\end{tabular}
\end{table}

\section{Generalizations}\label{sec-general}

In this section we state the main results in their full generality, i.e., for \edit{R12P2}{general} connections and arbitrary \red{simplicial and cubical complexes} . The \red{results of} \S\ref{ssec-statements} are obtained in the particular case when \red{the complex is a grid}, the gauge group is trivial, i.e., $G=\{1\}$, and the fields \clarity{are real-valued}. 
Most of the results of~\S\ref{sec-examples} are obtained from these general results by substituting specific Lagrangians.

\subsection{\red{General connections}} \label{ssec-general-connections}
\edit{R12P2}{}


\move{R12P2}{}
Interaction with a gauge field is introduced by replacement of (co)boundary by covariant (co)boundary.
\move{R12P2}{} \red{The latter} is defined in terms of cochain products as follows; \red{see Table~\ref{tab-products} and} cf.~\cite[\S IV--V]{Dimakis-etal-94}.
Let $U\in C^1(\red{M};G)$, $A=U-1$, $F$ be a gauge group field, the connection, and the curvature respectively.


\begin{definition} \label{def-cap-general} \label{def-covar-boundary-gauge} \label{def-covar-coboundary-gauge}
Denote by $C^k(\red{M};V)$ the set of functions defined on the set of $k$-dimensional faces and assuming values in a set $V$. Here $V$, and hence $C^k(\red{M};V)$, is a set, not necessarily a group.

Denote by $a\dots b$ the face $f$ such that $\min f=a$, $\max f=b$ (if such face $f$ exists, then it is unique). \remove{R12P2}{}
An ordered triple of faces $a\dots b, b\dots c\subset a\dots c$ of dimensions $k$, $l$, $k+l$ respectively is \emph{cooriented} (repectively, \emph{opposite oriented}), if the ordered set consisting of a positive basis in $a\dots b$ and a positive basis in $b\dots c$ is a positive (respectively, negative) basis in $a\dots c$. Write
$$
\langle a,b,c\rangle=
\begin{cases}
  +1, & \mbox{if }a\dots b,b\dots c, a\dots c\mbox{ are cooriented},  \\
  -1, & \mbox{if }a\dots b,b\dots c, a\dots c\mbox{ are oppositely oriented}.
\end{cases}
$$
The \emph{cup}, \emph{cap}, and \gram{\emph{cop product}} of functions $\Phi\in C^k(\red{M};\mathbb{C}^{p\times q})$ and $\Psi\in C^l(\red{M};\mathbb{C}^{q\times r})$ are the $\mathbb{C}^{p\times r}$-valued functions on $(k+l)$-, $(k-l)$-, and $(l-k)$-dimensional faces respectively given by
\begin{align*}
[\Phi\smile \Psi](a\dots c)&=\sum_{b:\,\dim(a\dots b)=k,\dim(b\dots c)=l}\langle a,b,c\rangle\Phi(a\dots b)\Psi(b\dots c);\\
[\Phi\frown \Psi](b\dots c)&=\sum_{a:\,\dim(a\dots c)=k,\dim(a\dots b)=l}\langle a,b,c\rangle\Phi(a\dots c)\Psi(a\dots b);\\
[\Phi\spleen  \Psi](a\dots b)&=\sum_{c:\,\dim(b\dots c)=k,\dim(a\dots c)=l}\langle a,b,c\rangle\Phi(b\dots c)\Psi(a\dots c),
\end{align*}
where the sums are over all the vertices such that there exist $3$ faces $a\dots b,b\dots c\subset a\dots c$.

For $\Phi\in C^k(\red{M};\mathbb{C}^{n\times n})$, the \emph{covariant coboundary} and
the \emph{covariant boundary} are respectively
\begin{align}\label{eq-covar-coboundary-gauge}
  D_A\Phi & := \delta \Phi+A\smile \Phi -(-1)^k \Phi\smile A;\\
  \label{eq-covar-boundary-gauge}
  D^*_A\Phi & :=\partial \Phi+ (\Phi^*\frown A)^* +(-1)^k (A\spleen  \Phi^*)^*.
\end{align}

\move{R12P3}{}
For $\phi\in C^k(\red{M};\mathbb{C}^{1\times n})$, 
the \emph{gauge transformation} by $g\in C^0(\red{M};G)$ is \red{the field} $\phi\smile g$, \red{and} the \emph{covariant coboundary} and
the \emph{covariant boundary} are respectively
\begin{align}
  \label{eq-covar-coboundary}
  D_A\phi & := \delta\phi-(-1)^k\phi\smile A;\\
  \label{eq-covar-boundary}
  D^*_A\phi & :=\partial\phi+ (-1)^k (A\spleen  \phi^*)^*.
\end{align}
\end{definition}

\begin{remark}\label{rem-covariant}\label{rem-covariant-boundary-informally}
\remove{R12P2}{}
\move{R12P3}{}Definitions of a gauge transformation and covariant (co)boundary crucially depend on the set of field values (more precisely, on the representation of $G$): compare \eqref{eq-covar-coboundary-gauge}--\eqref{eq-covar-boundary-gauge} and \eqref{eq-covar-coboundary}--\eqref{eq-covar-boundary}.
For $n=1$ there is a minor conflict of notation between these pairs of equations, cleared up by context.

Informally, \eqref{eq-covar-coboundary}--\eqref{eq-covar-boundary} mean the following. Think of the field value at a face $e$ as sitting at the maximal vertex $\max e$. Then the covariant (co)boundary value at a face $v$ is defined just as the ordinary (co)boundary, but all the involved field values are parallelly transported to the maximal
vertex $\max v$.

The definition of the {cup product} is equivalent to~\cite[(22.3)]{Whitney-38} but not~\cite[Chapter IX, \S14, Eq.~(7)]{Whitney-57}.

Up to sign and factors interchange, the \gram{cop product} is the \red{cap product} in the same \edit{R12P2}{grid} but with reversed vertices ordering. The \red{cap and cop} products vanish for $k<l$ and $k>l$ respectively, and do \emph{not} coincide for $k=l\ne 0$. Usually both are denoted in the same way, which does not lead to a conflict until one identifies chains and cochains (hence the domains of the products). Since we have performed such identification, we need to introduce new notation $\spleen$ and new term ``\gram{cop product}''.
\end{remark}

\begin{proposition}
\label{l-Bianchi}
For each gauge group field $U$ we have
$F=\delta A+A\smile A$, $D_{A}F=0$, and~\eqref{eq-def-boundary-particular}--\eqref{eq-def-unmotivated}.
\end{proposition}


\remove{R12P2}{}


\begin{definition} \label{def-local}
(Cf.~Definition~\ref{def-local-particular})
\remove{R12P3}{}
A map
$$
\mathcal{L}\colon C^k({M};\mathbb{C}^{1\times n})\times C^1({M};\mathbb{C}^{n\times n})\to C_0({M};\mathbb{R})
$$
is a \emph{local Lagrangian}, if for each vertex~$v\in M$ there is a smooth function $L_v(\phi_1,\dots,\phi_{p},\phi'_1,\dots,\phi'_{q})
\in \mathrm{C}^1\left(\left(\mathbb{C}^{1\times n}\right)^{p+q}\right)$
such that for each $\phi \in C^k({M};\mathbb{C}^{1\times n})$ and $U \in C^1({M};\mathbb{C}^{n\times n})$ we have
\begin{equation*}
\mathcal{L}[\phi,U](v)=
L_v(
\phi(e_1),\dots,\phi(e_p),[D_{A[U]}\phi](f_1),\dots,[D_{A[U]}\phi](f_q)),
\end{equation*}
where $e_1,\dots,e_p$ and $f_1,\dots,f_q$ are
all the faces of dimension $k$ and $k+1$ respectively
with the maximal vertex $v$; see Figure~\ref{fig-grid9}.
Define $\tfrac{\partial L_v}{\partial\phi_l}\colon \left(\mathbb{C}^{1\times n}\right)^{p+q}\to\mathbb{C}^{n\times 1}$ by
$\left(\tfrac{\partial L_v}{\partial\phi_l}\right)_{m}=
\tfrac{\partial  L_v}{\partial (\mathrm{Re}\,\phi_{l}^m)}-i\tfrac{\partial L_v}{\partial (\mathrm{Im}\,\phi_{l}^m)}$, where $\phi_{l}=(\phi_{l}^1,\dots,\phi_{l}^n)\in \mathbb{C}^{1\times n}$. Define $\tfrac{\partial L_v}{\partial\phi'_m}$ analogously.
Define
$$
\frac{\partial \mathcal{L}[\phi,U]}{\partial\phi}\in C_k({M};\mathbb{C}^{n\times 1})\qquad\text{and}\qquad \frac{\partial \mathcal{L}[\phi,U]}{\partial(D_A\phi)}\in C_{k+1}({M};\mathbb{C}^{n\times 1})
$$
by the following formulae for each $l=1,\dots,p$ and $m=1,\dots,q$:
\begin{align*}
\frac{\partial \mathcal{L}[\phi,U]}{\partial\phi}(e_{l})&:=
\frac{\partial L_v }{\partial\phi_l}(\phi(e_1),\dots,\phi(e_p),[D_{A[U]}\phi](f_1),\dots,[D_{A[U]}\phi](f_q)),\\
\frac{\partial \mathcal{L}[\phi,U]}{\partial(D_A\phi)}(e_{m})&:=
\frac{\partial L_v }{\partial\phi'_{m}}(\phi(e_1),\dots,\phi(e_p),[D_{A[U]}\phi](f_1),\dots,[D_{A[U]}\phi](f_q)).
\end{align*}
A field $\phi\in C^k({M};\mathbb{C}^{1\times n})$ is \edit{R1P4}{\emph{an extremal} or \emph{stationary}} for the functional $\mathcal{S}[\phi,U]=\epsilon\mathcal{L}[\phi,U]$, \red{
if $\left.\frac{\partial}
  {\partial t}\mathcal{S}[\phi+ t\Delta,U]\right|_{ t=0}=0$ for each $\Delta\in C^k({M};\mathbb{C}^{1\times n})$
and} given fixed~$U\in C^1({M};G)$.
\remove{R12P3}{}
\end{definition}


\begin{table}
  \caption{Partial derivatives of basic Lagrangians} 
  \label{tab-derivatives}
  \centering
  \begin{tabular}{|l|l|l|l|l|l|l|}
    \hline
    & Lagrangian \red{$\mathcal{L}[\phi]$} & 
    $L_v(\phi_1,\dots,\phi_{p},\phi'_1,\dots,\phi'_{q})$&
    $\tfrac{\partial\mathcal{L}}{\partial\phi}$
    &
    $\tfrac{\partial\mathcal{L}}{\partial(\delta\phi)}$
    \\
    \hline
    1 & $j\frown\red{\phi} \quad\red{(\dim\phi=1)}$
    &    $\sum_{l=1}^{p} j({e}_{l})\red{\phi_l}$ & $j$ & $0$  \\
    2 & \red{$\#\delta\phi\frown \delta\phi$}
    & $\sum_{l=1}^{q} g(k,l)\phi'_l\red{\phi'_l}$ & $0$ & \red{$2\#\delta\phi$}
    \\
    \hline
    & Lagrangian \red{$\mathcal{L}[\phi,U]$} & 
    $L_v(\phi_1,\dots,\phi_{p},\phi'_1,\dots,\phi'_{q})$&
    $\left(\tfrac{\partial\mathcal{L}}{\partial\phi}\right)^*$
    &
    $\left(\tfrac{\partial\mathcal{L}}{\partial(D_A\phi)}\right)^*$
    \\
    \hline
    3 & $\phi\frown \phi^*$
    & $\sum_{l=1}^{p} \phi_l\phi_l^*$ & $2\phi$ & $0$ \\
    4 &$\#D_A\phi\frown (D_A\phi)^*$
    & $\sum_{l=1}^{q} g(k,l)\phi'_l(\phi'_l)^*$ & $0$ & $2\#D_A\phi$  \\
    \hline
    & Lagrangian $\mathcal{L}[U]$ 
    &
    $L_v(U,\dots,U_{p},U'_1,\dots,U'_{q})$ & $\left(\tfrac{\partial\mathcal{L}}{\partial U}\right)^*$ &
    $\left(\tfrac{\partial\mathcal{L}}{\partial(F[U])}\right)^*$  \\
    \hline
    5 & $\mathrm{Re}\,\mathrm{Tr}[j^*\frown U]$
    &
    $\mathrm{Re}\,\mathrm{Tr}\sum_{l=1}^{p}j^*(e_{l})U_l$
    & $j$ & $0$\\
    6 & $\mathrm{Re}\,\mathrm{Tr}[\#F^*\frown F]$
    &
    $\mathrm{Re}\,\mathrm{Tr}\sum_{l=1}^{q}g(1,l)(U'_l)^*U'_l$
    & $0$ & $2\# F$\\
    \hline
  \end{tabular}
\end{table}

\begin{proposition}\label{l-differentiate-Lagrangian}
  For fixed \edit{R12P2}{current or covariant current} $j$, each of the Lagrangians in Table~\ref{tab-derivatives} to the left is local and the partial derivatives are given by the two columns to the right.
  \remove{R12P2}{}
\end{proposition}

\begin{theorem}[the Euler--Lagrange equation]\label{th-Euler-Lagrange-covar} 
  Let $\mathcal{L}\colon C^k({M};\mathbb{C}^{1\times n})\times C^1({M};\mathbb{C}^{n\times n})\to C_0({M};\mathbb{R})$ be a local Lagrangian, $A\in  C^1({M};\mathbb{C}^{n\times n})$ be
  a connection. Then $\phi\in C^k({M};\mathbb{C}^{1\times n})$ is \edit{R1P4}{an extremal} if and only if 
  \begin{equation}\label{eq-Euler-Lagrange-covar}
    D^*_A\left(\frac{\partial \mathcal{L}[\phi,1+A]}{\partial(D_A\phi)}\right)^*+\left(\frac{\partial \mathcal{L}[\phi,1+A]}{\partial\phi}\right)^*=0.
  \end{equation}
\end{theorem}

A \emph{local Lagrangian} $\mathcal{L}\colon C^1({M};\mathbb{C}^{n\times n})\to C_0({M};\mathbb{R})$ and the \emph{partial derivatives} $\frac{\partial \mathcal{L}}{\partial U}\in C_1({M};\mathbb{C}^{n\times n})$, $\frac{\partial \mathcal{L}}{\partial (F[U])}\in C_{2}({M};\mathbb{C}^{n\times n})$
are defined analogously to Definition~\ref{def-local}, only the fields $\phi$ and $D_A\phi$ are replaced by a gauge group field $U$ and the curvature $F[U]$ respectively (notice that $F[U]\ne D_AU$).
A gauge group field $U$ is \edit{R1P4}{an \emph{extremal}}, if it is stationary for the functional $\mathcal{S}[U]=\epsilon\mathcal{L}[U]$ under the constraint $U\in C^1({M};G)$. For fixed $\phi\in C^k({M};\mathbb{C}^{1\times n})$, a local Lagrangian $\mathcal{L}[\phi,U]$ in the sense of Definition~\ref{def-local}
is a local Lagrangian in the sense of this paragraph (by the second paragraph of Remark~\ref{rem-covariant-boundary-informally}). \clarity{The latter is the reason for using row-vectors $\phi$ rather than column-vectors.}

\begin{theorem}[the Euler--Lagrange equation]\label{th-Euler-Lagrange-gauge} 
  Let $\mathcal{L}\colon C^1({M};\mathbb{C}^{n\times n})\to C_0({M};\mathbb{R})$ be a local Lagrangian.
  Then a gauge group field $U\in C^1({M};G)$ is \edit{R1P4}{an extremal}, if and only if
  \begin{equation}\label{eq-Euler-Lagrange-gauge}
    \mathrm{Pr}_{T_U G} \left[D^*_A\left(\frac{\partial \mathcal{L}[U]}{\partial(F[U])}\right)^*+\left(\frac{\partial \mathcal{L}[U]}{\partial U}\right)^*\right]=0.
  \end{equation}
\end{theorem}



\begin{theorem}[Noether's theorem]\label{th-Noether-covar}
  If a local Lagrangian $\mathcal{L}[\phi,U]$ satisfies~\eqref{eq-invariance} for some $\Delta\in C^k({M};\mathbb{C}^{1\times n})$ and fixed $U\in C^1({M};G)$, then for each \edit{R1P4}{extremal} $\phi$ the edgewise scalar product of the covariant current $j[\phi,U]=\left(\frac{\partial\mathcal{L}[\phi,U]}{\partial(D_A\phi)}\frown \Delta\right)^*$ with 
  $U$ is conserved, i.e. $\partial\langle j[\phi,U],U\rangle=0$.\remove{clarity}{}
\end{theorem}

A Lagrangian $\mathcal{L}[\phi,U]$ is \emph{gauge invariant}, if $\mathcal{L}[\phi\smile g,g^*\smile U \smile g]=\mathcal{L}[\phi,U]$ for each $\phi\in C^k(M;\mathbb{C}^{1\times n})$, $U\in C^1(M;G)$, $g\in C^0(M;G)$.
For gauge invariant Lagrangians the numerous Noether currents are combined together as follows.

\begin{theorem}[Charge conservation] \label{th-charge-conservation}
  If a local Lagrangian $\mathcal{L}[\phi,U]$ is gauge invariant, then for \move{clarity}{} each gauge group field $U$ and each \edit{R1P4}{extremal} $\phi$ the following covariant current is conserved:
  $$
  j[\phi,U]=\left(\frac{\partial{\mathcal{L}[\phi,U]}}{\partial (D_A\phi)}\frown \phi\right)^*=
  \left(\frac{\partial{\mathcal{L}[\phi,U]}}{\partial U}\right)^*,
  \qquad  \text{i.e., }\quad D^*_A\mathrm{Pr}_{T_U G} j[\phi,U]=0.
  $$
\end{theorem}

\begin{theorem}[Charge conservation] \label{th-charge-conservation-gauge}
  Let
  $\mathcal{L}[U]=
  \mathcal{L}'[U]-\mathrm{Re}\,\mathrm{Tr}\,[j^*\frown U]$ be a local Lagrangian, where $j\in C_1({M};\mathbb{C}^{n\times n})$ is fixed \clarity{and} $\mathcal{L}'[U]$ is gauge invariant and does not depend on $j$. Then for each \edit{R1P4}{extremal $U\in C^1(M;G)$} the covariant current $j$ is conserved, i.e., $D^*_A\mathrm{Pr}_{T_U G} j=0$.
\end{theorem}

The last three theorems are not completely obvious even 
\clarity{for} a $1\times 1$ grid.
The crucial gauge invariance 
is usually guaranteed by the following result.

\begin{proposition}[Gauge covariance, see \cite{Dimakis-etal-94}] \label{l-gauge-covariance}
  For each $U\in C^1({M};G)$, $\Phi\in C^k({M};\mathbb{C}^{n\times n})$, $\phi\in C^k({M};\mathbb{C}^{1\times n})$, $g\in C^0({M};G)$ we have:
  \begin{align*}
    A[g^*\smile U\smile g]&=g^*\smile A[U]\smile g+g^*\smile \delta g&(=g^*\smile A[U]\smile g-\delta g^*\smile g);  \\
    F[g^*\smile U\smile g]&=g^*\smile F[U]\smile g; \\
    D_{A[g^*\smile U\smile g]}(g^*\smile\Phi\smile g)&=g^*\smile (D_{ A[U]}\Phi) \smile g;
    &D_{A[g^*\smile U\smile g]}(\phi\smile g)=(D_{ A[U]}\phi)\smile g;     \\
    D^*_{A[g^*\smile U\smile g]}(g^*\smile\Phi\smile g)&=g^*\smile (D^*_{ A[U]}\Phi) \smile g;
    &D^*_{A[g^*\smile U\smile g]}(\phi\smile g)=(D^*_{ A[U]}\phi)\smile g.
  \end{align*}
  \red{The} Lagrangians in the left column \edit{R12P2}{and rows 3, 4, 6} of Table~\ref{tab-derivatives} 
  are gauge invariant.
\end{proposition}

\subsection{\red{Simplicial and cubical complexes}}
\edit{R12P2}{}

\move{R12P2}{}
\begin{definition}\label{def-spacetime}
A \emph{finite simplicial} (respectively, \emph{cubical}) \emph{complex} is a finite set of simplices (respectively, hypercubes) in a Euclidean space of some dimension satisfying the following properties:
\begin{itemize}
  \item
the intersection of any two simplices (respectively, hypercubes) from the set is either empty or their common face
(a simplex/hypercube itself is also viewed as its own face);
  \item
all the faces of a simplex (respectively, a hypercube)
  from the set belong to the set as well.
\end{itemize}
\emph{Spacetime} ${M}$ is an arbitrary finite simplicial or cubical complex with a fixed vertices ordering. For a cubical complex, we require that the minimal and the maximal vertex of each $2$-dimensional face are opposite \edit{R12P2}{(this is essential for the definition of products and curvature).} 
The simplices/hypercubes of $M$ are called \emph{faces} of $M$.
\end{definition}

\begin{remark} \label{rem-ordering} \edit{R12P2}{While vertices ordering is required, a particular choice is not that important. 
For an \emph{arbitrary} ordering, the discretization algorithm from~\S\ref{ssec-main-tools} automatically produces a local Lagrangian for all field theories we considered (cf.~Proposition~\ref{l-differentiate-Lagrangian}).
Changing the ordering is like changing the lattice: combinatorial relations are changed but the underlying physical theory remains the same.
}
\end{remark}

\red{Until this subsection, spacetime was a grid with the dictionary vertices ordering. Passing to general spacetime is like passing from a coordinate chart to a coordinate-free formulation. The paper is intentionally designed to make this 
almost automatic. For} an arbitrary spacetime ${M}$,
all notions \red{in} the middle column of Table~\ref{tab-translation} except $\#$ \red{and} 
$\langle T,h\rangle_k$ are defined literally as
\red{above} (see the right column for definition numbers) \red{up to the following modifications required for simplicial complexes ${M}$ only:}
\begin{description}
  \item[Definition~\ref{def-tensor}:] \red{The \emph{Cartesian square} $M\times M$ is now
      a \emph{cell complex} (rather than simplicial or cubical complex) with faces of the form $e\times f$, where $e$ and $f$ are faces of $M$.}
  \item[Definition~\ref{def-curvature}:] \red{The \emph{curvature} is no longer defined by~\eqref{eq-def-curvature} but now} \move{R12P2}{} by the formula
$$
F[U](abc)=U(ab)U(bc)-U(ac)
$$
for each face $abc$ with the vertices listed in increasing order $a<b<c$.
   \item[Definition~\ref{def-cap-general}:]
   \move{R12P2}{}A face is \red{no longer} determined by just the minimal and the maximal vertices. \red{Thus we} denote by $a_1a_2\dots a_{s+1}$ the $s$-dimensional face with the vertices $a_1< a_2<\dots<a_{s+1}$. 
   \remove{R12P2}{} \red{Then} $a\dots b$, $b\dots c$, $a\dots c$ are replaced by $a_1\dots a_sb$, $bc_1\dots c_t$, $a_1\dots a_sbc_1\dots c_t$ respectively, summation over $b$ is omitted, and summation over $a$ and $c$ is replaced by summation over all collections $(a_1,\dots ,a_s)$ and $(c_1,\dots, c_t)$ respectively.
\end{description}
\edit{R12P2}{With these modifications,
all the theorems and corollaries in \S\ref{ssec-statements} 
and~\S\ref{ssec-general-connections}, as well as their proofs, remain literally true for an arbitrary spacetime $M$. (Propositions~\ref{l-Bianchi} and~\ref{l-differentiate-Lagrangian} remain true, once one drops all $\#$-operators; see the proofs.) } We do not use and do not define $\#$ \red{and} 
$\langle T,h\rangle_k$ for ${M}\ne I^d_N$.

\begin{remark} \label{rem-compare-definition} \label{rem-identification}
\move{R12P2}{} To make the definition of \remove{clarity}{}
fields more accessible to nonspecialists, we took the liberty to use equivalent definitions of some commonly used notions and to identify spaces connected by the unique fixed isomorphism. \remove{clarity}{}
Now we compare Definition~\ref{def-cochain} with the other ones in literature.

Often \emph{simplicial} (or \emph{cubical}) $k$-\emph{chains} are defined in a more abstract way, as the elements of the linear space $C_k({M};\mathbb{R})$ generated by the $k$-dimensional faces of $M$ (with somehow fixed orientation); and $k$-\emph{cochains} are defined as elements of the dual space $C^k({M};\mathbb{R})$.
But space $C_k({M};\mathbb{R})$ comes with the obvious unique distinguished basis: the basis consists of all the $k$-dimensional faces; the orientation of the faces is determined by the order of their vertices in spacetime~$M$ as specified in Definition~\ref{def-boundary}; the faces are listed in the dictionary order with respect to the ordered lists of their vertices. The distinguished basis identifies both $C_k({M};\mathbb{R})$ and $C^k({M};\mathbb{R})$ with the set of real-valued functions defined on the set of $k$-dimensional faces, that is, $k$-\emph{cochains} in the sense of Definition~\ref{def-cochain}. Notice that this identification
is not related to spacetime metric.

Thus we do \emph{not} distinguish between chains and cochains. 
Inserting the obvious isomorphism between their spaces in our formulae would give no advantage but would only complicate notation. 
However, to make notation compatible with the commonly used one, we sometimes switch between  different notation $C^k({M};\mathbb{R})$ and $C_k({M};\mathbb{R})$ for the same object (in our setup).

We \emph{do} distinguish between row- and column-vectors. This makes clear, if the product of two vectors is a number or a matrix. Some of our results depend on the type of vectors used as field values.



We do \emph{not} assume that ${M}$ is a manifold. In fact, faces of ${M}$ of dimension $>2$ have never appeared at all in the examples from \S\ref{sec-examples}. \remove{R12P3}{} 
The whole ambient spacetime is not that important: think of an electric network lying on a table; is spacetime of the model  1-, 2-, 3- or 4-dimensional? This is why we avoid dual grids and the Hodge star. However dimension-like properties of $M$ like the average vertex degree \emph{are} of course important.
\end{remark}

\section{Proofs} \label{s:proofs}

\subsection{Basic results}\label{ssec-proofs-basic}

First we prove the results of~\S\ref{ssec-statements}.
\edit{R12P3}{The statements are recalled right before the proofs for the convenience.
Throughout \S\ref{ssec-proofs-basic} $\mathcal{L}\colon C^k({M};\mathbb{R})\to C_0({M};\mathbb{R})$ is a local Lagrangian
and $\phi,\Delta\in C^k({M};\mathbb{R})$ (they are not necessarily extremals).
Besides the notation from~\S\ref{ssec-statements}, we only use the following one:
\begin{itemize}
\item $\epsilon$ is the sum of the values of a $0$-dimensional field over all the vertices;
\item $\phi\frown \Delta$ is the $0$-dimensional field given by $[\phi\frown \Delta](v)=\sum_{f:\dim f=k,\max f=v}\phi(f)\Delta(f)$, where the sum is over all the $k$-dimensional faces $f$ 
    with the maximal vertex $v$; cf.~Definition~\ref{def-cap-grid}.
\end{itemize}
}


\begin{lemma}[Lagrangian functional derivative]\label{l-Lagrangian-functional-derivative}
\remove{R12P3}{}
For arbitrary fields $\phi,\Delta\in C^k({M};\mathbb{R})$ we have
$$
 \left.\frac{\partial \mathcal{L}[\phi+ t\Delta]}{\partial t}\right|_{ t=0}=
 \left(\frac{\partial\mathcal{L}[\phi]}{\partial\phi}+
 \partial\frac{\partial\mathcal{L}[\phi]}{\partial(\delta\phi)}\right)
 \frown \Delta
 \,-\,(-1)^k\partial\left(
 \frac{\partial\mathcal{L}[\phi]}{\partial(\delta\phi)}
 \frown\Delta\right).
$$
\end{lemma}

\begin{proof} Take a vertex $v\in {M}$. Starting with \eqref{eq-local}--\eqref{eq-Lagrangian-derivatives2}, \red{then} using the chain rule, \remove{R12P3}{} 
and finally the well-known 'integration by parts' identity \cite{Whitney-38} \edit{R12P3}{(which holds for any $\psi\in C^{k+1}({M};\mathbb{R})$)}
\begin{equation}\label{eq-by-parts}
\partial(\red{\psi}\frown \red{\phi}) = (-1)^{\dim\red{\phi}}(\partial \red{\psi}\frown \red{\phi}-\red{\psi}\frown \delta \red{\phi})
\vspace{-0.2cm}
\end{equation}
we get
\begin{align*}
\left.\frac{\partial \mathcal{L}[\phi+ t\Delta]}{\partial t}(v)\right|_{t=0}
&=
\left.\frac{\partial}{\partial t} L_v\left([\phi+ t\Delta]({e}_1),\dots,[\phi+ t\Delta]({e}_p),[\delta\phi
+\delta\,t\Delta]({f}_{1}),\dots,[\delta\phi
+\delta\,t\Delta]({f}_{q})\right)\right|_{ t=0}
\\&=
\sum_{l=1}^{p}
\frac{\partial L_v}{\partial\phi_l}
\left(\phi({e}_1),\dots,\phi({e}_p),
[\delta\phi]({f}_{1}),\dots,[\delta\phi]({f}_{q})\right)
\left.\frac{\partial}{\partial t}
[\phi+ t\Delta]({e}_{l})\right|_{ t=0}
\\&+
\sum_{m=1}^{q}
\frac{\partial L_v}{\partial\phi'_{m}}
\left(\phi({e}_1),\dots,\phi({e}_p),
[\delta\phi]({f}_{1}),\dots,[\delta\phi]({f}_{q})\right)
\left.\frac{\partial}{\partial t}
[\delta\phi+\delta\,t\Delta]({f}_{m})
\right|_{ t=0}\\
&=
\sum_{l=1}^{p}
\frac{\partial \mathcal{L}[\phi]}{\partial\phi}({e}_{l})
\Delta({e}_{l})+
\sum_{m=1}^{q}
\frac{\partial \mathcal{L}[\phi]}{\partial(\delta\phi)}({f}_{m})
[\delta\Delta]({f}_{m})
=
\left[\frac{\partial \mathcal{L}[\phi]}{\partial\phi}\frown
\Delta+
\frac{\partial \mathcal{L}[\phi]}{\partial(\delta\phi)}\frown
\delta\Delta\right](v)
\\&=
\left[
  \left(\frac{\partial\mathcal{L}[\phi]}{\partial\phi}+
  \partial \frac{\partial\mathcal{L}[\phi]}{\partial(\delta\phi)}\right)
  \frown \Delta
  -(-1)^k\partial\left(
  \frac{\partial\mathcal{L}[\phi]}{\partial(\delta\phi)}
  \frown\Delta\right)\right](v).\\[-1.4cm]
\end{align*}
\end{proof}

\begin{lemma}\label{l-nondegeneracy-triv}
  \clarity{Fix} $\phi\in C_k({M};\mathbb{R})$.
  If $\epsilon [\phi\frown \Delta]=0$ for each $\Delta\in C^k({M};\mathbb{R})$, then $\phi=0$.
\end{lemma}

\begin{proof} Take $\Delta=\phi$. Then \remove{R12P3}{}
$0=\epsilon [\phi\frown \phi]=\sum_{f:\dim f=k}\phi(f)^2$.
Thus $\phi=0$.
\end{proof}

\begin{theorem}[\red{Restatement of Theorem~\ref{th-Euler-Lagrange}}]
\label{th-Euler-Lagrange-repeated}
\edit{R12P3}{A field $\phi$ is an extremal $\Leftrightarrow$
$
    \partial\frac{\partial \mathcal{L}[\phi]}{\mathrm{\partial}(\delta\phi)}+\frac{\partial \mathcal{L}[\phi]}{\mathrm{\partial}\phi}=0.
$}
\end{theorem}

\begin{proof}[Proof of the Euler--Lagrange Theorem~\ref{th-Euler-Lagrange-repeated}]
A field $\phi$ is \edit{R1P4}{an extremal} if and only if for \red{any} 
 $\Delta$ 
we have
\begin{multline*}
0=
\epsilon \left.\frac{\partial \mathcal{L}[\phi+ t\Delta]}{\partial t}\right|_{ t=0}=
\epsilon \left[\left(\frac{\partial \mathcal{L}[\phi]}{\partial\phi} +\partial \frac{\partial\mathcal{L}[\phi]}{\partial(\delta\phi)}\right)\frown \Delta\right]
-\\
 -(-1)^k\epsilon \partial\left[
  \frac{\partial\mathcal{L}[\phi]}{\partial(\delta\phi)}
  \frown\Delta\right]
=
\epsilon \left[\left(\frac{\partial \mathcal{L}[\phi]}{\partial\phi} +\partial\frac{\partial\mathcal{L}[\phi]}{\partial(\delta\phi)}\right)
\frown \Delta\right].
\end{multline*}
The latter two equalities follow from Lemma~\ref{l-Lagrangian-functional-derivative} and the obvious identity $\epsilon\partial=0$ respectively.
Since $\Delta$ is arbitrary, by  Lemma~\ref{l-nondegeneracy-triv} the resulting equation is equivalent to \clarity{desired equation}~\eqref{eq-Euler-Lagrange}.
\end{proof}

\begin{theorem}[\red{Restatement of Theorem~\ref{th-Noether}}] 
\label{th-Noether-repeated}
    \edit{R12P3}{An extremal $\phi$ satisfies
  $
  \left.\frac{\partial}
  {\partial t}\mathcal{L}[\phi+ t\Delta]\right|_{ t=0}=0
  $
  $\Leftrightarrow$ the current
  $
  j[\phi]=\frac{\partial\mathcal{L}[\phi]}{\partial(\delta\phi)}\frown \Delta
  $
  is conserved.}
\end{theorem}

\begin{proof}[Proof of the Noether Theorem~\ref{th-Noether-repeated}]
By Lemma~\ref{l-Lagrangian-functional-derivative} and Theorem~\ref{th-Euler-Lagrange-repeated} for \edit{R1P4}{an extremal} $\phi$ we get
\begin{multline*}
  \left.\frac{\partial \mathcal{L}[\phi+ t\Delta]}{\partial t}\right|_{ t=0}
  =
  \left(\frac{\partial\mathcal{L}[\phi]}{\partial\phi}
  +\partial\frac{\partial\mathcal{L}[\phi]}{\partial(\delta\phi)}
  \right)
  \frown \Delta
  -(-1)^{k}\partial\left(
  \frac{\partial\mathcal{L}[\phi]}{\partial(\delta\phi)}
  \frown\Delta\right)
  =
  -(-1)^{k}\partial j[\phi].
\end{multline*}
Thus $j[\phi]$ is a conserved current, if and only if the left-hand side vanishes.
\end{proof}

\begin{remark} \edit{R1P6}{This is immediately generalized to symmetries of the action $\mathcal{S}$ rather than the Lagrangian~$\mathcal{L}$: if
$\left.\frac{\partial}
  {\partial t}\mathcal{S}[\phi+ t\Delta]\right|_{ t=0}=0$ then
$\left.\frac{\partial}
  {\partial t}\mathcal{L}[\phi+ t\Delta]\right|_{ t=0}=\partial\Lambda[\phi,\Delta]$ for some $\Lambda[\phi,\Delta]\in C_1({M};\mathbb{R})$ because ${M}$ is connected. Then $j[\phi]+(-1)^k\Lambda[\phi,\Delta]$ is a conserved current.}
\end{remark}

\begin{theorem}[\red{Restatement of Theorem~\ref{th-energy-conservation}}]
\label{th-energy-conservation-repeated}
\edit{R12P3}{For each extremal $\phi$ the \emph{energy-momentum} tensor
  $
  T[\phi]=\frac{\partial\mathcal{L}[\phi]}{\partial(\delta\phi)}\times \delta\phi +\frac{\partial\mathcal{L}[\phi]}{\partial\phi}\times \phi
  $
  is conserved.} 
\end{theorem}

\begin{proof}[Proof of \red{Energy-momentum conservation} Theorem~\ref{th-energy-conservation-repeated}]
  By Theorem~\ref{th-Euler-Lagrange-repeated}, Definition~\ref{def-tensor}, and the known identity $\partial\partial=\delta\delta=0$, \clarity{for each extremal $\phi$} we have
\begin{multline*}
\partial T[\phi]
\red{=\partial\left(\frac{\partial\mathcal{L}[\phi]}{\partial(\delta\phi)}\times \delta\phi +\frac{\partial\mathcal{L}[\phi]}{\partial\phi}\times \phi\right)}
=\partial\left(\frac{\partial\mathcal{L}[\phi]}{\partial(\delta\phi)}\times \delta\phi -\partial\frac{\partial\mathcal{L}[\phi]}{\partial(\delta\phi)}\times \phi\right)
=\\=\frac{\partial\mathcal{L}[\phi]}{\partial(\delta\phi)}\times \delta\delta\phi
+\partial\frac{\partial\mathcal{L}[\phi]}{\partial(\delta\phi)}\times \delta\phi -\partial\frac{\partial\mathcal{L}[\phi]}{\partial(\delta\phi)}\times \delta\phi
-\partial\partial\frac{\partial\mathcal{L}[\phi]}{\partial(\delta\phi)}\times \phi
=0.\\[-1.6cm]
\end{multline*}
\end{proof}

\subsection{\red{Summit}} \label{ssec-proofs-global}

\edit{R2P1}{Now we prove the result of~\edit{R2P1}{\S\ref{ssec-summit}}.}
We start with a \red{visual} heuristic 
proof of \red{a}
particular case \red{from~\S\ref{ssec-quick}}.

\begin{proof}[Proof of identity~\eqref{eq-discrete-Poynting}]
By definition, twice the left-hand side of~\eqref{eq-discrete-Poynting} equals

$
+
\underbrace{F(\includegraphics[width=0.4cm]{0+0+1+2s.png})\, F(\includegraphics[width=0.4cm]{0+0+1+2s.png})}_{1}+
\underbrace{F(\hspace{-0.3cm}\begin{tabular}{c}\vspace{-0.1cm}
\includegraphics[width=0.4cm]{0+1+1+2s.png}
\end{tabular}\hspace{-0.2cm})\, F(\includegraphics[width=0.4cm]{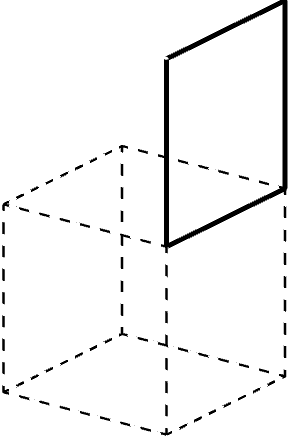})}_{2}+
\underbrace{F(\hspace{-0.2cm}\begin{tabular}{c}
\includegraphics[width=0.4cm]{0+1+2+2s.png}\end{tabular}\hspace{-0.2cm})\, F(\includegraphics[width=0.4cm]{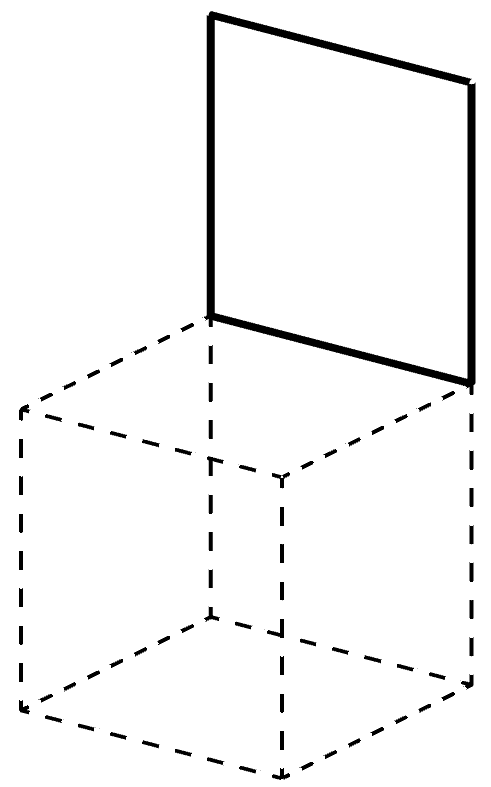})}_{3}$
$
-\underbrace{F(\includegraphics[width=0.4cm]{1+2s.png})\, F(\includegraphics[width=0.4cm]{1+2s.png})}_{4}-
\underbrace{F(\hspace{-0.3cm}\begin{tabular}{c}\vspace{-0.1cm}
\includegraphics[width=0.4cm]{_0+1+1+2s.png}
\end{tabular}\hspace{-0.2cm})\, F(\includegraphics[width=0.4cm]{0+1+1+2s.png})}_{5}-
\underbrace{F(\hspace{-0.2cm}\begin{tabular}{c}
\includegraphics[width=0.4cm]{_0+1+2+2s.png}\end{tabular}\hspace{-0.2cm})\, F(\includegraphics[width=0.4cm]{0+1+2+2s.png})}_{6}$

$
-
\underbrace{F(\includegraphics[width=0.6cm]{1+1+1+2s.png})\, F(\includegraphics[width=0.4cm]{0+1+1+2s.png})}_{7}-
\underbrace{F(\includegraphics[width=0.6cm]{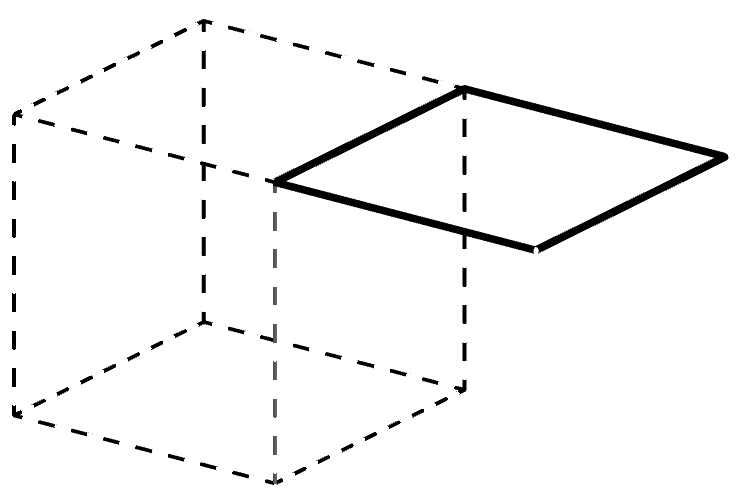})\, F(\includegraphics[width=0.4cm]{0+1+1+2s.png})}_{8}+
\underbrace{F(\includegraphics[width=0.4cm]{1+2s.png})\, F(\includegraphics[width=0.4cm]{0+2s.png})+
F(\includegraphics[width=0.4cm]{0+0+1+2s.png})\, F(\includegraphics[width=0.4cm]{0+2s.png})}_{9}
$

$
+
\underbrace{F(\includegraphics[width=0.6cm]{1+2+2+2s.png})\, F(\includegraphics[width=0.4cm]{0+1+2+2s.png})}_{10}+
\underbrace{F(\includegraphics[width=0.6cm]{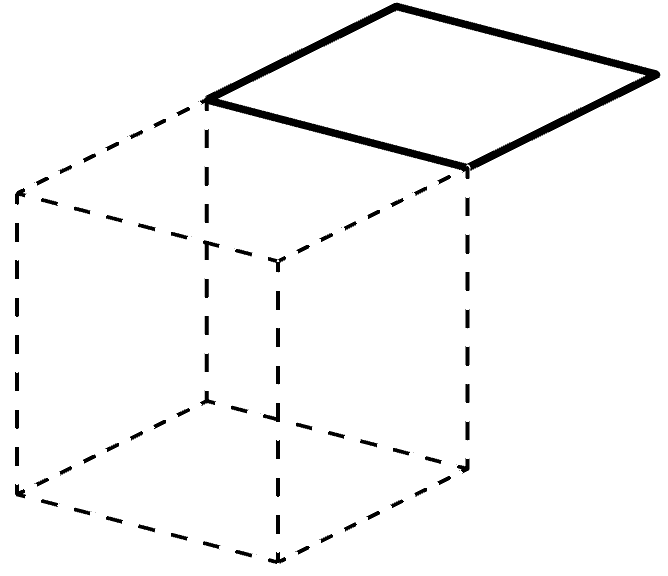})\, F(\includegraphics[width=0.4cm]{0+1+2+2s.png})}_{11}-
\underbrace{F(\includegraphics[width=0.4cm]{1+2s.png})\, F(\includegraphics[width=0.4cm]{0+1s.png})-
F(\includegraphics[width=0.4cm]{0+0+1+2s.png})\, F(\includegraphics[width=0.4cm]{0+1s.png})}_{12}$

\mbox{$
=\!\!$ $\left(\!
  F(\includegraphics[width=0.4cm]{1+2s.png})\!+\!               F(\includegraphics[width=0.4cm]{0+0+1+2s.png})
  \!\right)
  \!\!\left[
  \underbrace{F(\includegraphics[width=0.4cm]{0+0+1+2s.png})
  -F(\includegraphics[width=0.4cm]{1+2s.png})}_{1-4}-
  {\color{blue}\underbrace{F(\includegraphics[width=0.4cm]{0+1+1+2s.png})}_{a+b}}+
  \underbrace{F(\includegraphics[width=0.4cm]{0+2s.png})}_{9}+
  {\color{green}\underbrace{F(\includegraphics[width=0.4cm]{0+1+2+2s.png})}_{c+d}}-
  \underbrace{F(\includegraphics[width=0.4cm]{0+1s.png})}_{12}
  \right]$}

\smallskip
  $
  +
  \left[
  {\color{blue}\underbrace{F(\includegraphics[width=0.4cm]{1+2s.png})}_{a}}-
  \underbrace{F(\hspace{-0.3cm}\begin{tabular}{c}\vspace{-0.0cm}
       \includegraphics[width=0.6cm]{1+1+1+2s.png}
       \end{tabular}\hspace{-0.2cm})}_{7}-
  \underbrace{F(\hspace{-0.3cm}\begin{tabular}{c}\vspace{-0.2cm}
       \includegraphics[width=0.4cm]{_0+1+1+2s.png}
       \end{tabular}\hspace{-0.2cm})}_{5}+
  {\color{brown} \underbrace{F(\includegraphics[width=0.4cm]{0+1+1+2s.png})}_{e}}
  \right]
  F(\includegraphics[width=0.4cm]{0+1+1+2s.png})$
  $
  +
  \left[
  {\color{blue}\underbrace{F(\includegraphics[width=0.4cm]{0+0+1+2s.png})}_{b}}-
  \underbrace{F(\hspace{-0.3cm}\begin{tabular}{c}\vspace{-0.0cm}
       \includegraphics[width=0.6cm]{0+0+1+1+1+2s.png}
       \end{tabular}\hspace{-0.2cm})}_{8}-
  {\color{brown} \underbrace{F(\hspace{-0.3cm}\begin{tabular}{c}\vspace{-0.2cm}
       \includegraphics[width=0.4cm]{0+1+1+2s.png}
       \end{tabular}\hspace{-0.2cm})}_{e}}+
  \underbrace{F(\includegraphics[width=0.4cm]{0+0+0+1+1+2s.png})}_{2}
  \right]
  F(\includegraphics[width=0.4cm]{0+1+1+2s.png})$

  $
  -
  \left[
  {\color{green}\underbrace{F(\includegraphics[width=0.4cm]{1+2s.png})}_{c}}-
  \underbrace{F(\hspace{-0.3cm}\begin{tabular}{c}\vspace{-0.0cm}
       \includegraphics[width=0.6cm]{1+2+2+2s.png}
       \end{tabular}\hspace{-0.2cm})}_{10}+
  \underbrace{F(\hspace{-0.3cm}\begin{tabular}{c}\vspace{-0.2cm}
       \includegraphics[width=0.4cm]{_0+1+2+2s.png}
       \end{tabular}\hspace{-0.2cm})}_{6}-
  {\color{red}\underbrace{F(\includegraphics[width=0.4cm]{0+1+2+2s.png})}_{f}}
  \right]
  F(\includegraphics[width=0.4cm]{0+1+2+2s.png})$
  $
  -
  \left[
  {\color{green}\underbrace{F(\includegraphics[width=0.4cm]{0+0+1+2s.png})}_{d}}-
  \underbrace{F(\hspace{-0.3cm}\begin{tabular}{c}\vspace{-0.0cm}
       \includegraphics[width=0.6cm]{0+0+1+2+2+2s.png}
       \end{tabular}\hspace{-0.2cm})}_{11}+
  {\color{red} \underbrace{F(\hspace{-0.3cm}\begin{tabular}{c}\vspace{-0.2cm}
       \includegraphics[width=0.4cm]{0+1+2+2s.png}
       \end{tabular}\hspace{-0.2cm})}_{f}}-
  \underbrace{F(\includegraphics[width=0.4cm]{0+0+0+1+2+2s.png})}_{3}
  \right]
  F(\includegraphics[width=0.4cm]{0+1+2+2s.png})
 =0.
 $\\
Here the terms labeled by letters cancel each other; the terms in square brackets vanish by~\eqref{eq-Maxwell-1}.
\end{proof}

\edit{R2P1}{To proceed,} we are going to integrate tensors defined on $I^d_N\times I^d_N$ over the faces of the \emph{doubling}.


\begin{definition} \label{def-flow-doubling}
Dissect the hypercube $0\le x_0,x_1\dots, x_{d-1}\le N$ in $\mathbb{R}^d$ into $(2N)^d$ equal hypercubes. The cubical complex consisting of all the faces of the resulting hypercubes is the \emph{doubling} of $I^d_N$. For a vertex $f$ of the doubling, define $f_0,\dots,f_{d-1}\in\mathbb{Z}$ by the formula $f=f_0\mathrm{e_0}+\dots+f_{d-1}\mathrm{e}_{d-1}$. 
The face of the initial grid \clarity{$I^d_N$} with the center $f$ is denoted by $f$ as well.

Let $T$ be a partially symmetric type $(1,1)$ tensor, $g$ be a \gram{non-boundary} hyperface of the doubling, $\mathrm{e}_l\perp g$, $f=\max g$.
The \emph{$k$-th component of the flux of $T$ across $g$
in \red{the} positive normal direction}~is
$$
\langle T,g\rangle_k=\frac{1}{2}(-1)^
{\left(l+1+\sum\limits_{\min\{k,l\}\le m\le\max\{k,l\}}f_m\right)}
\cdot
\begin{cases}
 -T((f-\mathrm{e}_k)\times(f+\mathrm{e}_l) ),
 & \mbox{if } l\ne k,\, 2\nmid f_k, \,2\nmid f_l; \\
 T((f+\mathrm{e}_l-\mathrm{e}_k)\times f),
 & \mbox{if } l\ne k,\, 2\nmid f_k,\, 2\mid f_l; \\
 T(f\times(f+\mathrm{e}_l-\mathrm{e}_k)),
 & \mbox{if } l\ne k,\, 2\mid f_k,\, 2\nmid f_l; \\
 -T((f+\mathrm{e}_l)\times(f-\mathrm{e}_k)),
 & \mbox{if } l\ne k,\, 2\mid f_k,\, 2\mid f_l; \\
 T(f\times f)-T((f+\mathrm{e}_l)\times(f-\mathrm{e}_k)),
 & \mbox{if } l=k.
\end{cases}
$$
\remove{clarity}{}

Let $L$ be a type $(0,1)$ tensor, $g$ be a $d$-dimensional face of the doubling, $f=\max g$. Denote
$$
\langle L,g\rangle_k:=
\frac{1}{2}(-1)^{1+\sum_{m<k}f_m}\cdot
\begin{cases}
  L(f\times (f-\mathrm{e}_k)),
 & \mbox{if } 2\mid f_k; \\
  L((f-\mathrm{e}_k)\times f),
 & \mbox{if } 2\nmid f_k.
\end{cases}
$$
\end{definition}

\begin{proposition} \label{prop-flux}
The flux of a partially symmetric type $(1,1)$ tensor across a hyperface $h$ of the grid \clarity{$I^d_N$} (see Definition~\ref{def-flow}) is the sum of fluxes across all the hyperfaces of the doubling contained in~$h$.
\end{proposition}

\begin{proof} Compare the $k$-th components of the fluxes. Take $\mathrm{e}_l\perp h$. Consider the $2$ cases: $l\!=\!k$ and $l\!\ne\! k$.

For $l=k$, the map $g\mapsto \max g$ is a 1--1 map
between the set of hyperfaces of the doubling contained in $h$ and the set of faces of the initial grid $I^d_N$ contained in $h$ and containing $\max h$. (Recall that the vertex $\max g$ is identified with the face~$f$ of the initial grid with the center at $\max g$.) Since $\dim\mathrm{Pr}(f,k,k)=0=f_k\pmod2$, by Definitions~\ref{def-flow} and~\ref{def-flow-doubling} the case $l=k$ follows.

For $l\ne k$, the map
$g\mapsto
\begin{cases}
 \max g, &\mbox{if }2\nmid (\max g)_k,\\
 \max g-\mathrm{e}_k,&\mbox{if }2\mid (\max g)_k;
\end{cases}
$
is a 2--1 map between the set of hyperfaces of the doubling in $h$ and the set of faces $f$ of the initial grid $I^d_N$ such that $f\subset h$, $f\ni\max h$, $f\parallel\mathrm{e}_k$. The contribution of a pair of hyperfaces mapped to the same face $f$ to the sum of fluxes is
\begin{align*}
&\tfrac{1}{2}(-1)^{\left(l+1+\sum\limits_{\min\{k,l\}\le m\le\max\{k,l\}}f_m\right)}
T((f+\mathrm{e}_l-\mathrm{e}_k)\times f)+\\+&
\tfrac{1}{2}(-1)^{\left(l+1+\sum\limits_{\min\{k,l\}\le m\le\max\{k,l\}}(f_m+\deltaup_{mk})\right)}
[-T((f+\mathrm{e}_k+\mathrm{e}_l)\times(f+\mathrm{e}_k-\mathrm{e}_k))]
=\\=
&\tfrac{1}{2}(-1)^{\dim\mathrm{Pr}(f,k,l)+l+1}[T( (f+\mathrm{e}_l-\mathrm{e}_k)\times f)
+T((f+\mathrm{e}_l+\mathrm{e}_k)\times f)]
\end{align*}
because $2\nmid f_k$ and $2\mid f_l$ by the assumptions $f\subset h\perp \mathrm{e}_l$ and $f\parallel \mathrm{e}_k$. Summation over all such pairs proves the case $l\ne k$.
\end{proof}

Now let us prove an analogue of the Stokes formula; cf.~\eqref{eq-discrete-Poynting}.
For that we need a lemma.

\begin{lemma}\label{l-grid-boundary}
For each $k$-dimensional face $f$ of the $d$-dimensional grid $I^d_N$ denote by $[f]\in C^k(I^d_N;\mathbb{R})$ the field, which equals $1$ at $f$, and equals $0$ at all the other faces.
Then
\begin{align*}
  \partial [f] &=  \sum_{l:2\nmid f_l}
  (-1)^{\sum_{0\le m\le l}f_m}\cdot ([f-\mathrm{e}_l]-[f+\mathrm{e}_l]);\\
  \delta [f] &=  \sum_{l:2\mid f_l}
  (-1)^{\sum_{0\le m\le l}f_m}\cdot ([f-\mathrm{e}_l]-[f+\mathrm{e}_l]).
\end{align*}
\end{lemma}

\begin{proof}
  This is a direct computation using Definition~\ref{def-boundary}. It suffices to prove that $f$ and $f-\mathrm{e}_l$ are cooriented, if and only if $2\mid\sum_{0\le m\le l}f_m$. Assume that $2\mid f_l$; the opposite case is analogous. A positive basis in $f$ is the sequence formed by all the vectors $\mathrm{e}_m$ such that $2\nmid f_m$ in a natural order. A positive basis in $f-\mathrm{e}_l$ is obtained by insertion of $\mathrm{e}_l$ into the sequence. Adding the outer normal to the former basis means adding $\mathrm{e}_l$ at the beginning of the sequence instead. Since moving $\mathrm{e}_l$ to the beginning of the sequence requires $\sum_{0\le m< l}f_m\pmod{2}$ transpositions, the lemma follows.
\end{proof}

\begin{proposition}[the Stokes Formula] \label{prop-Stokes} 
Let $0\le k<d$.\remove{clarity}{} For each partially symmetric type $(1,1)$ tensor $T$ and each $d$-dimensional face $g$ of the doubling of ${I}^d_N$ we have $\langle T,\partial g\rangle_k=\langle \partial T,g\rangle_k$.
\end{proposition}

\begin{proof} This is a direct computation; a technical difficulty is signs.
Set $f=\max g$.
Assume that $2\mid f_k$; the opposite case is discussed at the end of the proof.
For any fields $\phi$ and $\psi$ denote $T(\psi\times \phi)=\sum_{e,f}T(e\times f)\psi(e)\phi(f)$.
Then by Definition~\ref{def-tensor} we have
$
\partial T(e\times f)=T([e]\times \partial[f])
  +T(\delta [e]\times [f])
$
and by Lemma~\ref{l-grid-boundary} we have \fix{}
\begin{align*}
  \red{2(-1)^{1+\sum_{m<k}f_m}}\langle \partial T,g\rangle_k&=\partial T(f\times (f-\mathrm{e}_k))
  =T([f]\times \partial[f-\mathrm{e}_k])
  +T(\delta [f]\times [f-\mathrm{e}_k])\\
  &= \sum_{l:2\nmid f_l-\deltaup_{kl}}(-1)^{\sum_{m\le l}(f_m-\deltaup_{mk})}\cdot
  \left[T(f \times (f-\mathrm{e}_k-\mathrm{e}_l))
  -T(f\times (f-\mathrm{e}_k+\mathrm{e}_l))
  \right]\\
  &+
  \sum_{l:2\mid f_l}(-1)^{\sum_{m\le l}f_m}\cdot  \left[T((f-\mathrm{e}_l)\times(f-\mathrm{e}_k))
  -T((f+\mathrm{e}_l)\times(f-\mathrm{e}_k))\right].
\end{align*}
It remains to show that here the $l$-th summand multiplied by $(-1)^{1+\sum_{m<k}f_m}$ equals twice the difference of the fluxes across the two opposite hyperfaces of $g$ orthogonal to $\mathrm{e}_l$ multiplied by $(-1)^l$. (The latter sign factor is required to get the right contribution of the two faces into the whole flux across $\partial g$ in the positive normal direction; see Lemma~\ref{l-grid-boundary} for $k=d$).
Denote $f'=f-\mathrm{e}_l$, $k'=\min\{k,l\}$, $l'=\max\{k,l\}$.
Denote by $g+\mathrm{e}_l/2$ and $g-\mathrm{e}_l/2$ the hyperfaces of $g$ orthogonal to $\mathrm{e}_l$
such that $\max (g+\mathrm{e}_l/2)=f$ and $\max (g-\mathrm{e}_l/2)=f'$ respectively.

Consider the following $3$ cases: 1) $l=k$; \quad 2) $l\ne k$ and \clarity{$2\nmid f_l$}; \quad 3) $l\ne k$ and \red{$2\mid f_l$}.

For $l\!=\!k$ (hence $2\! \mid\! f_k\!=\!f_l$) the $l$-th summands in the two sums multiplied by $(-1)^{1+\sum_{m<k}f_m}$ add~up~to
\begin{align*}
&(-1)^{1+\sum_{m<k}f_m}(-1)^{\sum_{m\le k}(f_m-\deltaup_{mk})}\cdot
\left[T(f\times (f-2\mathrm{e}_k))-T(f\times f)\right]+\\+
&(-1)^{1+\sum_{m<k}f_m}(-1)^{\sum_{m\le k}f_m}\cdot
\left[ T((f-\mathrm{e}_k)\times(f-\mathrm{e}_k))
-T((f+\mathrm{e}_k)\times (f-\mathrm{e}_k))\right]=\\=
&(-1)^{f_k+1}\cdot \left[T(f\times f)-T((f+\mathrm{e}_k)\times(f-\mathrm{e}_k))\right]-\\-
&(-1)^{f_k}\cdot
\left[T((f-\mathrm{e}_k)\times (f-\mathrm{e}_k))     -T((f-\mathrm{e}_k+\mathrm{e}_k)\times(f-2\mathrm{e}_k)) \right]
=\\=
&(-1)^k\cdot(-1)^{k+1+f_k}\cdot \left[T(f\times f)-T((f+\mathrm{e}_k)\times(f-\mathrm{e}_k))\right]
-\\-
&(-1)^k\cdot (-1)^{k+1+f'_k}\cdot \left[T(f'\times f')     -T(f'+\mathrm{e}_k)\times (f'-\mathrm{e}_k)) \right]
=\\=
&(-1)^k 2\langle T,g+\mathrm{e}_k/2\rangle_k-
(-1)^k 2\langle T,g-\mathrm{e}_k/2\rangle_k;
\end{align*}
see Definition~\ref{def-flow-doubling} applied for $l=k$.
We have found the contribution of the $l$-th summands for~$l=k$.

For $l\ne k$ and $2\nmid f_l$ the $l$-th summand multiplied by $(-1)^{1+\sum_{m<k}f_m}$ is
\begin{align*}
&(-1)^{1+\sum_{m<k}f_m}(-1)^{\sum_{m\le l}(f_m-\deltaup_{mk})}\cdot \left[T(f\times(f-\mathrm{e}_k-\mathrm{e}_l))
  -T(f\times(f-\mathrm{e}_k+\mathrm{e}_l))\right]=\\
  \overset{(*)}{=}
&(-1)^{1+\sum_{k'\le m\le l'}f_m}\cdot \left[
T(f\times(f+\mathrm{e}_l-\mathrm{e}_k))-
T((f-\mathrm{e}_l+\mathrm{e}_l)\times(f-\mathrm{e}_l-\mathrm{e}_k))\right]=\\=
&(-1)^l\cdot (-1)^{l+1+\sum_{k'\le m\le l'}f_m}\cdot
T(f\times(f+\mathrm{e}_l-\mathrm{e}_k))-\\-
&(-1)^l\cdot (-1)^{l+1+\sum_{k'\le m\le l'}f'_m}\cdot [-T((f'+\mathrm{e}_l)\times (f'-\mathrm{e}_k))]
=\\=
&(-1)^l 2\langle T,g+\mathrm{e}_l/2\rangle_k
-
(-1)^l 2\langle T,g-\mathrm{e}_l/2\rangle_k;
\end{align*}
see Definition~\ref{def-flow-doubling} applied for
$l\ne k$, $2\mid f_k$, $2\nmid f_l$ and $l\ne k$, $2\mid f'_k$, $2\mid f'_l$. Here (*) follows from
$$
1+\sum_{m<k}f_m+\sum_{m\le l}(f_m-\deltaup_{mk})=
\begin{cases}
  \sum_{k\le m\le l}f_m, & \mbox{if }k<l; \\
  \sum_{l <  m <  k}f_m+1, & \mbox{if }k>l;
\end{cases}
\qquad= \sum_{k'\le m\le l'}f_m\pmod2,
$$
where we used the conditions $2\nmid f_l$ and $2\mid f_k$ to change the range of summation over $m$.


For $l\ne k$ and $2\mid f_l$ the $l$-th summand multiplied by $(-1)^{1+\sum_{m<k}f_m}$ is
\begin{align*}
&(-1)^{1+\sum_{m<k}f_m}(-1)^{\sum_{m\le l}f_m}\cdot  \left[T((f-\mathrm{e}_l)\times (f-\mathrm{e}_k))- T((f+\mathrm{e}_l)\times (f-\mathrm{e}_k))\right]=\\=
&(-1)^{\sum_{k'\le m\le l'}f_m}\cdot
\left[T((f+\mathrm{e}_l)\times(f-\mathrm{e}_k))-
T((f-\mathrm{e}_l)\times(f-\mathrm{e}_l+\mathrm{e}_l-\mathrm{e}_k))\right]=\\=
&(-1)^l\cdot (-1)^{l+1+\sum_{k'\le m\le l'}f_m}\cdot
[-T((f+\mathrm{e}_l)\times(f-\mathrm{e}_k))]-\\-
&(-1)^l\cdot (-1)^{l+1+\sum_{k'\le m\le l'}f'_m}\cdot
T(f'\times(f'+\mathrm{e}_l-\mathrm{e}_k))
=\\=
&(-1)^l 2\langle T,g+\mathrm{e}_l/2\rangle_k-
(-1)^l 2\langle T,g-\mathrm{e}_l/2\rangle_k.
\end{align*}
Summation of the expressions obtained in the three cases completes the proof in the case when $2\mid f_k$.

For $2\nmid f_k$ the proof is analogous and starts from the evaluation of \fix{} $\red{2(-1)^{1+\sum_{m<k}f_m}}\langle \partial T,g\rangle_k=\partial T( (f-\mathrm{e}_k)\times f)$. For $l=k$ one ends up with an expression involving $T((f-\mathrm{e}_k)\times(f+\mathrm{e}_k))$
rather than $T((f+\mathrm{e}_k)\times(f-\mathrm{e}_k))$. But the latter two values are equal because $T$ is partially symmetric.
\end{proof}

\begin{theorem}[\red{Restatement of Theorem~\ref{prop-global-momentum-conservation}}] \label{prop-global-momentum-conservation-repeated}
\edit{R12P3}{If a partially symmetric type $(1,1)$ tensor $T$ is conserved, then for each $d$-dimensional face $g$ disjoint with $\partial{I}^d_N$ and each $k$ 
we get $\langle T,\partial g\rangle_k=0$.}
\end{theorem}

\begin{proof}[Proof of \red{Theorem}~\ref{prop-global-momentum-conservation-repeated}] This follows directly from Propositions~\ref{prop-flux} and~\ref{prop-Stokes} because \clarity{$g$} can be filled by $d$-dimensional faces of the doubling.
\end{proof}

\begin{remark} \label{rem-conservation-apart-boundary} \clarity{Here the assumption that $T$ is conserved can be relaxed to $[\partial T](e\times f)=0$ for all faces~$e,f\not\subset\partial{I}^d_N$, because the boundary faces do not contribute to the flux when $g\cap \partial{I}^d_N=\emptyset$.}
\end{remark}

\subsection{Identities}

For the sequel we need several identities for cochain operations, most of which are well-known. \edit{R12P2}{Throughout this subsection $M$ is an arbitrary simplicial or cubical complex unless otherwise indicated.}

\begin{definition}\label{def-pairing}
The \emph{pairing} of fields $\phi,\psi\in C^k({M};\mathbb{C}^{m\times n})$, where $m=1$ or $m=n$,  is defined by
$$
\langle\phi,\psi\rangle=
\mathrm{Re}\,\mathrm{Tr} \sum_{\text{$k$-dimensional faces }f} \phi(f)\psi^*(f)=
\epsilon\,\mathrm{Re}\,\mathrm{Tr}[\phi\frown\psi^*]=
\epsilon\,\mathrm{Re}\,\mathrm{Tr}[\phi\spleen \psi^*].
$$
\end{definition}

Given $U\in C^1({M};G)$, denote by $C^1({M};T_U G)$ the set of all $\Delta\in C^1({M};\mathbb{C}^{n\times n})$ such that $\Delta(e)$ belongs to the tangent space $T_{U(e)}G$ for each edge $e$.
For $\phi\in C_k({M};\mathbb{C}^{n\times m})$, where $m=1$ or $m=n$, denote
\begin{equation}\label{eq-covar-boundarydual}
  \check D_A^* \phi=(D_A^*\phi^*)^*=
  \partial\phi+(-1)^k A\spleen\phi+\deltaup_{mn}\cdot \phi\frown A.
\end{equation}

\begin{lemma}[\red{Pairing} nondegeneracy] \label{l-nondegeneracy-general}\label{l-nondegeneracy-gauge}
  \gram{Fix} $\phi\in C_k({M};\mathbb{C}^{m\times n})$,  $\psi\in C_0({M};\mathbb{C}^{n\times n})$,
  $\chi\in C_1({M};\mathbb{C}^{n\times n})$.

  If $\langle \phi,\Delta\rangle=0$ for each $\Delta\in C^k({M};\mathbb{C}^{m\times n})$, then $\phi=0$.

  If $\langle \psi,\Delta\rangle=0$ for each $\Delta\in C^0({M};T_1 G)$, then $\mathrm{Pr}_{T_1 G}\psi=0$.

  If $\langle \chi,\Delta\rangle=0$ for each $\Delta\in C^1({M};T_U G)$, then $\mathrm{Pr}_{T_U G}\chi=0$.
\end{lemma}

\begin{proof} For the first assertion, take $\Delta=\phi$.  Then $0=\langle\phi,\phi\rangle=\sum_{f}\mathrm{Re}\,\mathrm{Tr}[\phi^*(f)\phi(f)]=\sum_f\sum_{i,j=1}^{m,n}|\phi_{ij}(f)|^2$. Thus $\phi=0$.

For the third assertion, take $\Delta=\mathrm{Pr}_{T_U G}\chi$. Then $0=\langle \chi,\mathrm{Pr}_{T_U G}\chi\rangle=
\sum_{e}\langle \chi(e),\mathrm{Pr}_{T_{U(e)}G}\chi(e)\rangle =
\sum_{e}\langle\mathrm{Pr}_{T_{U(e)}G}\chi(e),\mathrm{Pr}_{T_{U(e)}G}\chi(e)\rangle $, where the sums are over all edges $e$,
because $\mathrm{Pr}_{T_{U(e)}G}$ is an orthogonal projection. Since the pairing $\langle \cdot,\cdot\rangle$ on $\mathbb{C}^{n\times n}$ is nondegenerate, it follows that $\mathrm{Pr}_{T_U G}\chi=0$.

The second assertion is proved analogously.
\end{proof}



\begin{lemma}\label{l-boundary-cubical}
In a \emph{cubical} complex ${M}$, for each $U\in C^1({M};G)$ and $\Phi\in C^k({M};\mathbb{C}^{n\times n})$ we have
  \begin{align*}
  D_A\Phi&=U\smile \Phi-(-1)^k\Phi\smile U; & F&=U\smile U;\\
  \check D^*_A\Phi&=\Phi\frown U+(-1)^k U\spleen \Phi. & &
  \end{align*}
The two identities in the 1st column hold for a \emph{simplicial} complex $M$ for $k=0$ and $k=1$ respectively.
\end{lemma}

\begin{proof}
By Definitions~\ref{def-boundary} and~\ref{def-cap-general} it follows that
\begin{align*}
  \hspace{-0.4cm}[\delta\Phi](a\dots c) &=  \sum_{\substack{b:\dim(a\dots b)=1,\\\dim(b\dots c)=k}}\langle a,b,c\rangle\Phi(b\dots c)
  - (-1)^{k}\sum_{\substack{b:\dim(a\dots b)=k,\\\dim(b\dots c)=1}}\langle a,b,c\rangle\Phi(a\dots b)
  =1\smile\Phi-(-1)^{k}\Phi\smile 1;\\
  \hspace{-0.4cm}[\partial\Phi](b\dots c) &=\sum_{\substack{a:\dim(a\dots b)=1,\\\dim(a\dots c)=k}}\langle a,b,c\rangle\Phi(a\dots c)
  + (-1)^{k}\sum_{\substack{d:\dim(b\dots d)=k,\\\dim(c\dots d)=1}}\langle b,c,d\rangle\Phi(b\dots d)
  =\Phi\frown 1+(-1)^{k}1\spleen\Phi,
\end{align*}
where $1$ is the unit gauge group field and the sums are over the vertices such that there exist faces $a\dots b,b\dots c\subset a\dots c$ or $b\dots c,c\dots d\subset b\dots d$. Using~\eqref{eq-covar-coboundary-gauge}--\eqref{eq-covar-boundary-gauge},
we get the 
\clarity{identities in the first column. The identity $F=U\smile U$ holds by Definitions~\ref{def-curvature} and~\ref{def-cap-general}; see Figure~\ref{fig-grid} to the right.}
\end{proof}


\begin{lemma}\label{l-identities}\label{l-cap-delta-covar}
  (Cf.~\cite{Dimakis-etal-94}) For each $\phi\in C^k({M};\mathbb{C}^{p\times q}),\psi\in C^l({M};\mathbb{C}^{q\times r}),\chi\in C^m({M};\mathbb{C}^{r\times s})$ we have
\begin{align*}
  \delta\delta&=0; & \delta(\phi\smile\psi)&=(\delta\phi)\smile \psi+(-1)^{\dim \phi}\phi\smile\delta\psi;& (\phi\smile\psi)\smile\chi &= \phi\smile(\psi\smile\chi);
  \\
  \partial\partial&=0; &
  \partial(\phi\frown \psi) &= (-1)^{\dim\psi}(\partial \phi\frown \psi-\phi\frown \delta \psi);&
  (\phi\frown\psi)\frown\chi &= \phi\frown(\psi\smile\chi);
  \\
  \epsilon\partial &= 0; &
  \partial(\phi\spleen  \psi) &= \phi\spleen \partial  \psi+(-1)^{\dim\psi-\dim\phi}\delta\phi\spleen   \psi;& \phi\spleen (\psi\spleen \chi) &= (\phi\smile\psi)\spleen\chi;
  \\
  & & (\phi\smile\psi)^*&=
  \begin{cases}
    \psi^*\frown\phi^*, & \mbox{if } \dim\phi=0;\\
    \psi^*\spleen\phi^*,  & \mbox{if } \dim\psi=0;
  \end{cases}
  & (\phi\spleen \psi)\frown\chi &= \phi\spleen(\psi\frown\chi).
\end{align*}
For each $\phi\in C^{\red{k}}({M};\mathbb{C}^{n\times m}),\psi\in C^{\red{l}}({M};\mathbb{C}^{m\times n})$, $U\in C^1({M};G)$, where $m=1$ or $m=n$, we have \fix{}
\begin{align*}
  D_A D_A\psi&= -\psi\smile F +\deltaup_{mn}\cdot F\smile\psi; &
  D_A(\phi\smile \psi) &= D_A \phi\smile\psi +(-1)^{\dim\phi}\phi\smile D_A \psi, \mbox{\red{ if $m=n$;}}\\
  \check D^*_A \check D^*_A\phi&= -F\spleen \phi +\deltaup_{mn}\cdot \phi\frown F; &
  \check D^*_A(\phi\frown \psi) &= (-1)^{\dim\psi}(\check D^*_A \phi\frown \psi-\phi\frown D_A \psi);\\
  \mathrm{Re}\,\mathrm{Tr}\,\epsilon\,\check D^*_A\phi&=0, \mbox{ if $m=n$ and $\dim\phi=1$}; &
  \check D^*_A(\phi\spleen \psi) &=  \phi\spleen \check D^*_A\psi+(-1)^{\dim\psi-\dim\phi}D_A \phi\spleen \psi, \mbox{\red{ if $m=n$.}}
\end{align*}
For each $\phi\in C^k({M};\mathbb{C}^{m\times n}),\psi\in C^l({M};\mathbb{C}^{n\times n}\text{ or }\mathbb{C}^{m\times m}),\chi\in C_{k+l}({M};\mathbb{C}^{m\times n})$, $U\in C^1({M};G)$, where $m=1$ or $m=n$ (and $l=1$ for the identities in the 1st and 3rd column below), we have:
\begin{align*}
  \langle\chi,\delta\phi\rangle &= \langle \partial\chi,\phi\rangle;
  &
  \langle\chi,\psi\smile\phi\rangle &= \langle(\chi^*\frown\psi)^*,\phi\rangle;
  &
  \mathrm{Re}\,\mathrm{Tr}\, D^*_A\psi&=\partial \,\mathrm{Re}\,\mathrm{Tr}\, [U^*\cdot\psi];
  \\
  \langle\chi,D_A\phi\rangle &= \langle D^*_A\chi,\phi\rangle;&
  \langle\chi,\phi\smile\psi\rangle &= \langle(\psi\spleen \chi^*)^*,\phi\rangle;
  &
  \mathrm{Pr}_{T_1 G}\, D^*_A \psi&=D^*_A\mathrm{Pr}_{T_U G}\, \psi.
  \end{align*}
  In the 3rd column, \fix{$m=n$ and} ``$\cdot$'' is the edgewise product: $[U^*\cdot\psi](e):=U^*(e)\psi(e)$ for each edge $e$.
\end{lemma}

\begin{proof}
The identities involving neither \gram{cop product} nor covariant (co)boundary are well-known in the case when the functions assume values in a commutative ring; cf.~\cite{Dimakis-etal-94}. Without the commutativity the proof is literally the same.  Let us prove the remaining identities.


For an ordered $4$-ple of faces $a\dots b, b\dots c, c\dots d\subset a\dots d$ write $\langle a,b,c,d\rangle=+1$, if the ordered set consisting of positive bases in $a\dots b$, $b\dots c$, $c\dots d$ is a positive basis in $a\dots d$. Otherwise write
$\langle a,b,c,d\rangle= -1$. Clearly, $\langle a,b,c,d\rangle=\langle a,b,c\rangle\langle a,c,d\rangle=\langle a,b,d\rangle\langle b,c,d\rangle$.
Thus by Definition~\ref{def-cap-general}
\begin{align*}
[\phi\spleen (\psi\spleen \chi)](a\dots b)
&=\sum_{c:\dim(b\dots c)=k,\dim(a\dots c)=m-l}\langle a,b,c\rangle \phi(b\dots c)[\psi\spleen\chi](a\dots c)\\
&=\sum_{c,d:\dim(b\dots c)=k, \dim(c\dots d)=l,\dim(a\dots d)=m}\langle a,b,c\rangle
\langle a,c,d\rangle \phi(b\dots c)\psi(c\dots d)\chi(a\dots d)\\
&=\sum_{c,d:\dim(b\dots c)=k, \dim(c\dots d)=l,\dim(a\dots d)=m}\langle a,b,d\rangle
\langle b,c,d\rangle \phi(b\dots c)\psi(c\dots d)\chi(a\dots d)\\
&=[(\phi\smile\psi)\spleen \chi](a\dots b).
\end{align*}

Setting $m=k+l$, changing the notation $\chi$ to $\chi^*$, and applying the operator $\epsilon\,\mathrm{Re}\,\mathrm{Tr}$, we obtain $\langle(\psi\spleen \chi^*)^*,\phi\rangle=\langle\chi,\phi\smile\psi\rangle $. Taking $\psi=A$, $\phi\in C^k({M};\mathbb{C}^{m\times n})$, $\chi\in C_{k+1}({M};\mathbb{C}^{m\times n})$, multiplying by $(-1)^{\dim\phi}=-(-1)^{\dim\chi}$, adding the known identity $\langle \partial\chi,\phi\rangle =\langle\chi,\delta\phi\rangle$
(and for $m=n$ also the known identity $\langle (\chi\frown\psi^*)^*,\phi\rangle= \langle\chi,\phi\smile\psi\rangle$), and
using~\eqref{eq-covar-coboundary-gauge}--\eqref{eq-covar-boundary}, we get
$\langle D^*_A\chi,\phi\rangle= \langle\chi,D_A\phi\rangle$.

The formula for $(\phi\spleen \psi)\frown\chi$ is proved analogously.

Next, the formula for
$\check D^*_A(\phi\frown \psi)$ for a cubical complex and $m=n$ follows from
\begin{align*}
  (-1)^l\check D^*_A(\phi\frown \psi)
  &=(-1)^l(-1)^{k-l}U\spleen (\phi\frown \psi)+ (-1)^l(\phi\frown \psi)\frown U\\
  &=
  (-1)^{k}(U\spleen  \phi)\frown\psi
  +(\phi\frown U)\frown \psi
  -\phi\frown (U\smile \psi)
  +(-1)^l\phi\frown (\psi\smile U)\\
  &=(\check D_A^*\phi)\frown \psi-\phi\frown D_A\psi,
\end{align*}
where we used Lemma~\ref{l-boundary-cubical} and the identities not involving (covariant) (co)boundary.
Alternatively, the formula for $\check D^*_A(\phi\frown \psi)$ can be deduced from the formula for $\delta(\phi\smile\psi)$ by pairing with an arbitrary field $\Delta$ and applying Lemma~\ref{l-nondegeneracy-general} and the identities from the paragraph before the previous one; this works for a simplicial complex and for $m=1$ as well.

The formulae for $D_A(\phi\smile \psi)$, $\check D^*_A(\phi\spleen \psi)$, $D_A D_A$, $\check D^*_A\check D^*_A$  are proved analogously.

Finally, for each vertex $v$ by Lemma~\ref{l-boundary-cubical} we have (where $\langle U, \psi\rangle$ is the edgewise scalar product)
  \begin{align*}
  \hspace{-0.6cm}[\mathrm{Re}\,\mathrm{Tr}\, D^*_A\psi](v)&=\mathrm{Re}\,\mathrm{Tr}[\psi^*\frown U - U\spleen  \psi^*]^*(v)=\sum_{e:\max e=v}\langle \psi(e),U(e)\rangle- \sum_{e:\min e=v}\langle \psi(e),U(e)\rangle=[\partial \langle U, \psi\rangle](v);\\
  D^*_A\mathrm{Pr}_{T_U G} \psi
  &=((\mathrm{Pr}_{T_U G} \psi)^*\frown U-U\spleen(\mathrm{Pr}_{T_U G} \psi)^*)^*
  =\mathrm{Pr}_{T_1 G}(\psi^*\frown U)^* -\mathrm{Pr}_{T_1 G}(U\spleen  \psi^*)^*
  =\mathrm{Pr}_{T_1 G} D^*_A \psi.
  \end{align*}
  Applying the operator $\epsilon$ we get
  $
  \mathrm{Re}\,\mathrm{Tr}\,\epsilon\,\check D^*_A\psi=
  \epsilon\,\mathrm{Re}\,\mathrm{Tr}\,D^*_A\psi^*=
  \epsilon\,\partial\,\mathrm{Re}\,\mathrm{Tr}[U\cdot\psi]=0.
  $
\end{proof}


\subsection{Generalizations}

Now we 
\red{prove} the results of \S\ref{sec-general}. \edit{R12P2}{Starting from Lemma~\ref{l-lagrangian-functional-derivative-covar} below, the proof} 
 is parallel to that of~\S\ref{ssec-proofs-basic}.

\begin{proof}[Proof of Proposition~\ref{l-Bianchi}]
\move{R12P3}{}
By the \red{formulae} of Lemma~\ref{l-boundary-cubical} for $F$ and $D_A\Phi$ in the case when \clarity{$(\Phi,A)=(A,0)$,} 
we get $F=
(1+A)\smile (1+A)
=0+D_0 A+A\smile A=\delta A+A\smile A$.
By Lemma~\ref{l-boundary-cubical} and the associativity of the \gram{cup product}, $D_A F=U\smile F-F\smile U=U\smile (U\smile U)-(U\smile U)\smile U=0$. By Lemma~\ref{l-boundary-cubical} and the 3rd column of Table~\ref{tab-products} we get~\eqref{eq-def-boundary-particular}.

Let us prove~\eqref{eq-def-unmotivated}. \remove{R12P2}{} 
For each \red{face} $f\supset e$ we have either $\min f=\min e$ or $\max f=\max e$ \edit{R12P2}{(cf.~Definition~\ref{def-spacetime}).} Consider a face $f=abcd$ containing $e=ab$ such that $\min f=a$. Then $U(e)-U(\partial f -e)=U(ab)-U(adcb)=\left(F(abcd)^*U(bc)\right)^*$. Applying $\#$ and summing the obtained expression over all such faces $f$, we get $(\# F^*\frown U)^*$. Analogous sum over all the faces $f$ such that $\max f=b$ gives $(U\spleen \# F^*)^*$. Then Lemma~\ref{l-boundary-cubical} implies~\eqref{eq-def-unmotivated}.

\edit{R12P2}{For a cubical complex the proof is the same, only one drops all $\#$-operators. For a simplicial complex, in addition, references to Lemma~\ref{l-boundary-cubical} should be replaced by a direct checking.}
\end{proof}

\begin{proof}[Proof of Proposition~\ref{l-differentiate-Lagrangian}]
  This is a straightforward computation using the explicit expression for the function $L_v$ given in the middle part of Table~\ref{tab-derivatives}. There we use notation from Definition~\ref{def-local} and \remove{R12P2}{}
  $$g(k,l)=\begin{cases}
                    (-1)^k, & \mbox{if } f_{l}\parallel (1,0,\dots,0), \\
                    (-1)^{k-1}, & \mbox{if } {f}_{l}\perp (1,0,\dots,0).
  \end{cases}$$
\remove{R12P2}{}
\red{For a cubical/simplicial complex the proof is the same, only one drops all $\#$ and $g(k,l)$.}
\end{proof}

\begin{lemma}[Lagrangian functional derivative]\label{l-lagrangian-functional-derivative-covar}
For a local Lagrangian $\mathcal{L}\colon C^k({M};\mathbb{C}^{1\times n})\times C^1({M};\mathbb{C}^{n\times n})\to C_0({M};\mathbb{R})$
and arbitrary fields $\phi,\Delta\in C^k({M};\mathbb{C}^{1\times n})$, $U\in C^1({M};G)$ we have
$$
 \left.\frac{\partial \mathcal{L}[\phi+ t\Delta,U]}{\partial t}\right|_{ t=0}=
 \mathrm{Re}\,\mathrm{Tr}\left[
 \left(\frac{\partial\mathcal{L}[\phi,U]}{\partial\phi}+
 \check D^*_A\frac{\partial\mathcal{L}[\phi,U]}{\partial(D_A\phi)}\right)
 \frown \Delta
 -(-1)^k \check D^*_A\left(
 \frac{\partial\mathcal{L}[\phi,U]}{\partial(D_A\phi)}
 \frown\Delta\right)\right].
$$
\end{lemma}

\begin{proof}
  This is proved literally as Lemma~\ref{l-Lagrangian-functional-derivative} with $\delta$ and $\partial$ replaced by $D_A$ and $\check D^*_A$ respectively, and $\mathrm{Re}\,\mathrm{Tr}$ applied to each summand. Instead of~\eqref{eq-by-parts} use the formula for $\check D^*_A(\phi\frown\psi)$ from Lemma~\ref{l-identities}.
\end{proof}

\begin{proof}[Proof of Theorem~\ref{th-Euler-Lagrange-covar}]
  A  field $\phi$ is \edit{R1P4}{an extremal}, if and only if $\left.\frac{\partial \mathcal{S}[\phi+ t\Delta,U]}{\partial t}
  \right|_{ t=0}
  =0$ for each $\Delta\in C^k({M},\mathbb{C}^{1\times n})$.
  By Lemmas~\ref{l-lagrangian-functional-derivative-covar} and~\ref{l-nondegeneracy-general} this is equivalent to~\eqref{eq-Euler-Lagrange-covar} because $\epsilon\,\mathrm{Re}\,\mathrm{Tr}\, \check D^*_A=0$ by Lemma~\ref{l-identities}.
\end{proof}

\begin{lemma}[Lagrangian functional derivative]\label{l-lagrangian-functional-derivative-gauge}
For a local Lagrangian $\mathcal{L}\colon C^1({M};\mathbb{C}^{n\times n})\to C_0({M};\mathbb{R})$
and arbitrary fields $U\in C^1({M};G)$, $\Delta\in C^1({M};T_U G)$ we have
$$
 \left.\frac{\partial \mathcal{L}[U+ t\Delta]}{\partial t}\right|_{ t=0}=
 \mathrm{Re}\,\mathrm{Tr}\left[
 \left(\frac{\partial\mathcal{L}[U]}{\partial U}+
 \check D^*_A\frac{\partial\mathcal{L}[U]}{\partial (F[U])}\right)
 \frown \Delta
 + \check D^*_A\left(
 \frac{\partial\mathcal{L}[U]}{\partial (F[U])}
 \frown\Delta\right)\right].
$$
\end{lemma}

\begin{proof}
This is proved analogously to Lemma~\ref{l-Lagrangian-functional-derivative} with $\phi$ and $\delta\phi$ replaced by $U$ and $F=\delta A+A\smile A$ (see Proposition~\ref{l-Bianchi}), using  the formula for $\check D^*_A(\phi\frown\psi)$ from Lemma~\ref{l-identities} instead of~\eqref{eq-by-parts}, and
\begin{align*}
  \left.\frac{\partial}{\partial t} F[U+ t\Delta]
  \right|_{ t=0}&=\left.\frac{\partial}{\partial t} [\delta (U+t\Delta-1)+(U+t\Delta-1)\smile(U+t\Delta-1)]
  \right|_{ t=0}\\
  &=\delta\Delta+(U-1)\smile \Delta+\Delta\smile(U-1)=D_A\Delta.\\[-1.5cm]
\end{align*}
\end{proof}

\begin{proof}[Proof of Theorem~\ref{th-Euler-Lagrange-gauge}]
  A gauge group field $U$ is \edit{R1P4}{an extremal}, if and only if $\left.\frac{\partial \mathcal{S}[U+ t\Delta]}{\partial t}
  \right|_{ t=0}
  =0$ for each $\Delta\in C^1({M},T_{U}G)$.
  By Lemmas~\ref{l-lagrangian-functional-derivative-gauge}, \ref{l-identities}, and~\ref{l-nondegeneracy-gauge} this is equivalent to~\eqref{eq-Euler-Lagrange-gauge}.
\end{proof}

\begin{proof}[Proof of Theorem~\ref{th-Noether-covar}]
  This follows from $\partial \langle j[\phi,U],U\rangle=\mathrm{Re}\,\mathrm{Tr}\, D^*_Aj[\phi,U]=\left.\frac{\partial\mathcal{L}[\phi+t\Delta,U]}{\partial t}\right|_{t=0}=0$.
  Here the 1st equality is given by Lemma~\ref{l-identities}. The 2nd one is proved
  as in the proof of Theorem~\ref{th-Noether-repeated} with $\delta$, $\partial$ replaced by $D_A$, $\check D^*_A$, and $\mathrm{Re}\,\mathrm{Tr}$ applied to each summand. The 3rd one is~\eqref{eq-invariance}.
\end{proof}

\begin{remark} \label{rem-subset} If~\eqref{eq-invariance} holds in a subset of $M$, then the current $\langle j[\phi,U],U\rangle$ is conserved on the subset.
\end{remark}

\begin{lemma}[Lagrangian functional derivative] \label{l-lagrangian-functional-derivative-charge}
For a local Lagrangian $\mathcal{L}\colon C^k({M};\mathbb{C}^{1\times n})\times C^1({M};\mathbb{C}^{n\times n})\to C_0({M};\mathbb{R})$
and arbitrary fields $\phi\in C^k({M};\mathbb{C}^{1\times n})$ and $U,\Delta\in C^1({M};\mathbb{C}^{n\times n})$ we have
$$
 \left.\frac{\partial \mathcal{L}[\phi,U+ t\Delta]}{\partial t}\right|_{t=0} =
 \mathrm{Re}\,\mathrm{Tr}\left[
 \left(\frac{\partial\mathcal{L}[\phi,U]}{\partial(D_A\phi)}
 \frown \phi\right)\frown \Delta\right]
 \qquad\mbox{and}\qquad
 \frac{\partial \mathcal{L}[\phi,U]}{\partial U} =\frac{\partial\mathcal{L}[\phi,U]}{\partial(D_A\phi)}
 \frown \phi.
$$
\end{lemma}

\begin{proof}
Analogously to the proof of Lemma~\ref{l-Lagrangian-functional-derivative} using \eqref{eq-covar-coboundary} and Lemma~\ref{l-cap-delta-covar} we get
  \begin{align*}
  \left.\frac{\partial \mathcal{L}[\phi,U+ t \Delta]}{\partial t}
  \right|_{ t=0}
  &=
  \mathrm{Re}\,\mathrm{Tr}\,
  \left[
  \frac{\partial \mathcal{L}[\phi,U]}{\partial \phi}\frown \frac{\partial \phi}{\partial t}+
  \frac{\partial\mathcal{L}[\phi,U]}{\partial(D_A\phi)}\frown \left.\frac{\partial (D_{A[U+ t\Delta]}\phi)}{\partial t}
  \right|_{ t=0}\right]\\
  &=0+\mathrm{Re}\,\mathrm{Tr}\,\left[
  \frac{\partial\mathcal{L}[\phi,U]}{\partial(D_A\phi)}\frown
  \left.\frac{\partial [\delta\phi+\phi\smile(U-1+t \Delta)]}{\partial t}
  \right|_{ t=0}\right]\\
  &=
  \mathrm{Re}\,\mathrm{Tr}\,\left[
  \frac{\partial\mathcal{L}[\phi,U]}{\partial(D_A\phi)}\frown
  (\phi\smile \Delta)\right]
  =
  \mathrm{Re}\,\mathrm{Tr}\,\left[
  \left(\frac{\partial\mathcal{L}[\phi,U]}{\partial(D_A\phi)}\frown \phi\right)\frown\Delta\right].
\end{align*}
A local Lagrangian $\mathcal{L}[\phi,U]$ is also local with respect to $U$ and does not depend on $F[U]$. Since $\Delta\in C^1({M};\mathbb{C}^{n\times n})$ is arbitrary, by Lemmas~\ref{l-lagrangian-functional-derivative-gauge} and~\ref{l-nondegeneracy-general} it follows that $\frac{\partial \mathcal{L}[\phi,U]}{\partial U} =\frac{\partial\mathcal{L}[\phi,U]}{\partial(D_A\phi)}
 \frown \phi$.
\end{proof}

\begin{lemma}[Infinitesimal form of gauge invariance]\label{l-infinitesimal-gauge}
For each gauge invariant differentiable function $
\mathcal{L}\colon C^k({M};\mathbb{C}^{1\times n})\times C^1({M};\mathbb{C}^{n\times n})\to C_0({M};\mathbb{R})
$ and each $\Delta\in C^0({M},T_1 G)$ we have
$$
\left.\frac{\partial}
  {\partial t} \mathcal{L}[\phi+t \phi\smile\Delta,U+ t D_A\Delta ]\right|_{ t=0}=0.
$$
\end{lemma}

\begin{proof}
Since $\mathcal{L}[\phi,U]$ is gauge invariant and differentiable, by Lemma~\ref{l-boundary-cubical} up to first order in $t$
\begin{align*}
\mathcal{L}[\phi,U]
&=\mathcal{L}[\phi\smile\exp(t\Delta),\exp(-t\Delta)\smile U\smile \exp(t\Delta)]\\
&=\mathcal{L}[\phi+t\phi\smile\Delta,U+t(U\smile \Delta-\Delta\smile U)]+o(t)\\
&=\mathcal{L}[\phi+t\phi\smile\Delta, U+ t D_A\Delta]+o(t) \qquad\mbox{as }t\to 0.
\end{align*}
\clarity{Subtracting $\mathcal{L}[\phi,U]$ from both sides and dividing by $t$,}
we get the required result.
\end{proof}

\begin{lemma}[Local covariant constants] \label{l-const}
  For each $U\in C^1({M};G)$, $g_0\in T_1G$, and each vertex $v$ there is $g\in C^0({M};T_1G)$ such that $g(v)=g_0$ and $[D_A g](uv)=0$ for each neighbor $u$ of $v$.
\end{lemma}

\begin{proof}
  Set $g(v)=g_0$, $g(u)=U(uv)g(v)U(vu)$ at each neighbor $u$ of $v$, and let $g$ be arbitrary at the other vertices. By Lemma~\ref{l-boundary-cubical} we have $[D_A g](uv)=
  U(uv)g(v)-U(uv)g(v)U(vu)U(uv)=0$.
\end{proof}

\begin{proof}[Proof of Theorem~\ref{th-charge-conservation}] Take an arbitrary vertex $v$ and $g_0\in T_1G$. Let $g\in C^0({M};T_1G)$ be given by Lemma~\ref{l-const}. Apply Lemma~\ref{l-infinitesimal-gauge} for $\Delta=g$.
Since $D_A g(uv)=0$ for each neighbor $u$ of $v$, we obtain that equation~\eqref{eq-invariance} holds at the vertex $v$ with $\Delta=\phi\smile g$ (notice that the connection in~\eqref{eq-invariance} does not depend on $t$). By Theorem~\ref{th-Noether-covar}, Remark~\ref{rem-subset}, and Lemma~\ref{l-identities}, we have \clarity{}
\begin{multline*}
0=\partial \mathrm{Re}\,\mathrm{Tr}\,\left[ \left(\frac{\partial\mathcal{L}[\phi,U]}{\partial(D_A\phi)}
\frown(\phi\smile g)\right)\cdot U\right](v)
=
\mathrm{Re}\,\mathrm{Tr}\, \bigg[\check D^*_A
\bigg(\underbrace{\left(\frac{\partial\mathcal{L}[\phi,U]}{\partial(D_A\phi)}
\frown\phi\right)}_{\red{j[\phi,U]^*}}\frown g\bigg)\bigg](v)\\
=
\mathrm{Re}\,\mathrm{Tr}\, \left[\check D^*_A
\red{j[\phi,U]^*}\frown g-
\red{j[\phi,U]^*}\frown D_A g\right](v)
=\mathrm{Re}\,\mathrm{Tr}\, \left[D^*_A
\red{j[\phi,U]}(v)\cdot g_0\red{^*}\right].
\end{multline*}
Here we used that $[D_Ag](uv)=0$ for each edge $uv$ containing $v$. Since the vertex $v$ and $g_0\in T_1 G$ are arbitrary, by Lemma~\ref{l-nondegeneracy-gauge} it follows that $\mathrm{Pr}_{T_1G}\,D^*_A \red{j[\phi,U]}
=0$. \red{Then} by Lemma~\ref{l-identities} we have
$D^*_A \mathrm{Pr}_{T_UG}\, j[\phi,U]=0$.
\move{clarity}{} \red{Finally}, by Lemma~\ref{l-lagrangian-functional-derivative-charge} we have $\frac{\partial\mathcal{L}[\phi,U]}{\partial(D_A\phi)}
\frown\phi=\frac{\partial\mathcal{L}[\phi,U]}{\partial U}$. 
\end{proof}

\begin{proof}[Proof of Theorem~\ref{th-charge-conservation-gauge}]
  Denote $\mathcal{S}[U]=\epsilon\mathcal{L}[U]$ and $\mathcal{S}'[U]=\epsilon\mathcal{L}'[U]$.
  Take arbitrary $\Delta\in C^0({M},T_1 G)$.
  By Lemmas~\ref{l-infinitesimal-gauge} (with $\mathcal{L}[\phi,U]$ replaced by $\mathcal{L}'[U]$) and~\ref{l-cap-delta-covar} we get
  $$
  \left.\frac{\partial}{\partial t} \mathcal{S}[U+ t D_A\Delta ]\right|_{ t=0}=
  \left.\frac{\partial}
  {\partial t}\left( \mathcal{S}'[U+ t D_A\Delta ]+\langle j, U+t D_A\Delta\rangle\right)\right|_{ t=0}
  =0+\langle j, D_A\Delta\rangle=\langle D^*_A j, \Delta\rangle.
  $$
  For a \edit{R1P4}{stationary} gauge group field $U$ the \red{left} side vanishes, because $D_A\Delta= U\smile \Delta-\Delta\smile U\in C^1({M},T_U G)$ is a possible variation of $U$.
  Thus $\langle D^*_A j, \Delta\rangle=0$ for arbitrary $\Delta\in C^0({M},T_1 G)$.
  By Lemmas~\ref{l-nondegeneracy-gauge} and~\ref{l-identities} we get $0=\mathrm{Pr}_{T_1 G} D^*_A j=D^*_A \mathrm{Pr}_{T_U G} j$, as required.
\end{proof}

\begin{proof}[Proof of Proposition~\ref{l-gauge-covariance}]
Let us present the proof for a cubical complex.
For a simplicial complex the argument is literally the same, only each instance of the fourth vertex ``$d$'' is just removed.

Since the group $G$ consists of unitary matrices, for each edge $uv$ and each face $abcd$ we have
\begin{align*}
  A[g^*\smile U\smile g](uv)&=
  g^*(u)U(uv)g(v)-1\\&=
  g^*(u)(U(uv)-1)g(v)+g^*(u)(g(v)-g(u))=
  [g^*\smile A[U]\smile g+g^*\smile \delta g](uv);\\
  F[g^*\smile U\smile g](abcd)&=
  [g^*\smile U\smile g](abc)-[g^*\smile U\smile g](adc)\\&=
  g^*(a)U(ab)g(b)g^*(b)U(bc)g(c)-g^*(a)U(adc)g(c)=
  [g^*\smile F[U]\smile g](abcd).
\end{align*}
Now, using \eqref{eq-covar-coboundary}--\eqref{eq-covar-boundary} and Lemma~\ref{l-identities} we get
\begin{align*}
  D_{A[g^*\smile U\smile g]}(\phi\smile g)&=
  \delta(\phi\smile g)-(-1)^k \phi\smile g \smile [g^*\smile A[U]\smile g+g^*\smile \delta g]\\&=
  (\delta\phi)\smile g +(-1)^{k}\phi\smile\delta g-(-1)^k \phi\smile (g\smile g^*)\smile [A[U]\smile g+\delta g]\\&=
  (D_{ A[U]}\phi)\smile g;\\
  \left(D^*_{A[g^*\smile U\smile g]}(\phi\smile g)\right)^*&=
  \partial(\phi\smile g)^*+(-1)^k [g^*\smile A[U]\smile g-\delta g^*\smile g]\spleen(\phi\smile g)^*  \\&=
  \partial(g^*\spleen\phi^*)+(-1)^k [g^*\smile A[U]\smile g-\delta g^*\smile g]\spleen (g^*\spleen\phi^*)  \\&=
  g^*\spleen\partial\phi^*+(-1)^k \delta  g^*\spleen\phi^*
  +(-1)^k (g^*\smile A[U]-\delta g^*)\spleen(g\spleen (g^*\spleen\phi^*))\\&=
  g^*\spleen(\partial\phi^*+(-1)^k A[U]\spleen\phi^*)
  =\left(D^*_{ A[U]}\phi\smile g\right)^*.
  \end{align*}
The formulae involving $\Phi\in C^k({M};\mathbb{C}^{n\times n})$ are proved analogously. Gauge invariance of the Lagrangians
\edit{R12P2}{in rows 3, 4, and 6 of} Table~\ref{tab-derivatives} is a straightforward consequence.
\end{proof}

\subsection{Proofs of examples}
\label{ssec-example-proofs}

Now we apply the general results of \edit{R2P1}{\S\ref{ssec-statements}, \S\ref{ssec-summit}, and \S\ref{ssec-general-connections}} to prove particular results of \S\ref{sec-examples}.
\remove{R2P1}{}

\begin{proof}[Proof of Corollary~\ref{cor-networks-momentum-conservation}]
This follows \edit{R2P1}{from Theorem~\ref{prop-global-momentum-conservation} and
Remark~\ref{rem-conservation-apart-boundary}}
applied to 
$T[\phi]=\delta\phi\times\delta\phi$,
\red{because $\partial T[\phi]=
\partial\delta\phi\times\delta\phi+\delta\phi\times\delta\delta\phi
=0$ apart $\partial I^2_N$, since $\partial\delta\phi=s$  and $\delta\delta=0$.}
\remove{R12P3}{}
\end{proof}

\begin{proof}[Proof of Theorem~\ref{th-networks-convergence}]
First let us prove the ``convergence'' of $F_N$ to $\mathrm{F}$.
It is convenient to modify the grid slightly.
Consider the auxiliary grid $M$ obtained by dissection of $\mathrm{I}^2 $ into $(N+1)^2$ equal squares and its dual $N\times N$ grid $M'$ with the vertices at face-centers of $M$. Consider all the discrete fields in question as defined on $M'$ instead of the initial $N\times N$ grid; this does not affect approximation.

Let $F'_N$ be the function on vertices of $M$ such that $\partial\delta F'_N=0$ apart $\partial \mathrm{I}^2 $ and $F'_N=\mathrm{F}$ on $\partial \mathrm{I}^2 $. The restriction of $F'_N$ to \gram{non-boundary} vertices can be considered as a function on faces of $M'$. Actually, it is a magnetic field on $M'$ generated by the source $s_N$ (in particular, it exists by Proposition~\ref{prop-magnetic-existence}). Indeed, the condition $\partial\delta F'_N=0$ implies that it is a magnetic field generated by \emph{some} source. The source is exactly $s_N$ because for each boundary vertex $v$ of the initial $N\times N$ grid we have $F'_N(v_+)-F'_N(v_-)=\mathrm{F} (v_+)-\mathrm{F}(v_-) =\int_{v_-v_+}\mathrm{s}\,\mathrm{dl}=s_N(v)$,
where $v_-,v,v_+$ are in the counterclockwise order along $\partial \mathrm{I}^2 $. By Propositions~\ref{prop-current-existence} and~\ref{prop-magnetic-existence} the function $F'_N-F_N$ on faces of $M'$ is a constant (depending on $N$).

By \cite[Proposition~3.3]{Chelkak-Smirnov-08} on the set of vertices $v$ \clarity{of $M$} at distance $\ge r$ from $\partial \mathrm{I}^2 $, we have $F'_N(v)\rightrightarrows \mathrm{F} (v)$ as $N\to\infty$. In particular,
for one of the faces $f_N$ closest to $c:=(\frac{1}{2},\frac{1}{2})$ we have
$F'_N (f_N)\to \mathrm{F} (c)=0=F_N(f_N)$ as $N\to\infty$. Since $F'_N-F_N$ is a constant  function on \clarity{faces $f$ of $M'$,} it follows that $$
F_N(f)\rightrightarrows F'_N(f)\rightrightarrows \mathrm{F} (\max f)\rightrightarrows N^2\int_f\mathrm{F} \,\mathrm{dS}.$$

The convergence of $j_N=-\partial F_N$
follows immediately from the second part of \cite[Proposition~3.3]{Chelkak-Smirnov-08}.

To prove the convergence of $\phi_N$, join a vertex $v$ with the vertex $u$ closest to $c$ such that $\phi_N(u)=0$ by a shortest grid path $uv$. By the convergence of $j_N$ we get
$
\phi_N(v)\!=\!\sum_{e\subset uv}\langle uv,e \rangle j_N(e)
\rightrightarrows \int_{cv} \vec{\mathrm{ j}}\cdot\mathrm{d}\vec{\mathrm{l}}={\phiup}(v).
$

The convergence of the other fields is a straightforward consequence. For instance, let  $e=uv$ be a horizontal edge with the midpoint $e'$ and $f\supset e$ be a face with the center $f'$. Then
\begin{align*}
  N L_N(e'f') &= \tfrac{N}{2} j_N(e)F_N(f)
     \rightrightarrows F_N(f)\tfrac{N}{2}\!\int_e \vec{\mathrm{ j}}\cdot\mathrm{d}\vec{\mathrm{l}}
     \rightrightarrows \mathrm{F} (e')\vec{\mathrm{ j}}(e')\cdot \tfrac{N}{2}\!\int_e \mathrm{d}\vec{\mathrm{l}}
     = *\vec{\mathrm{ j}}(e')\mathrm{F} (e')\cdot
     N\!\int_{e'f'}\!\!\mathrm{d}\vec{\mathrm{l}}
     \rightrightarrows N\!\int_{e'f'}\!\!\vec{\mathrm{ L}}\cdot\mathrm{d}\vec{\mathrm{l}},\\
  N^2\sigma_{N,2}(uv) &=
     -\tfrac{N^2}{2}  [\delta\phi(uv)^2-\delta\phi(vv_+)\delta\phi(v_-v)]
     \rightrightarrows \tfrac{1}{2}\tfrac{\partial\phiup}{\partial \mathrm{x}_2}(v)^2-\tfrac{1}{2}\tfrac{\partial\phiup}{\partial \mathrm{x}_1}(v)^2=\sigmaup_{22}(v)\rightrightarrows N\!\int_e\left(\sigmaup_{22}\,\mathrm{dx}^1-
     \sigmaup_{21}\,\mathrm{dx}^2\right),
\end{align*}
as required (in the latter formula the notations $v_+$ and $v_-$ from Definition~\ref{def-Maxwell-stress} are used).
\end{proof}


\begin{proof}[Proof of Corollary~\ref{cor-Maxwell}]
  \red{Use} Theorem~\ref{th-Euler-Lagrange} and Proposition~\ref{l-differentiate-Lagrangian};
  \remove{clarity}{}
   see rows \red{1--2} of Table~\ref{tab-derivatives}.
\end{proof}

\begin{proof}[Proof of Theorem~\ref{cor-free}]
\move{R2P1}{}Tensor~\eqref{eq-energy-conservation} is partially symmetric for this particular Lagrangian $\mathcal{L}[\phi]$; \red{see}~rows 2--3 of Table~\ref{tab-derivatives}.
Thus the \clarity{result} follows from Theorems~\ref{th-energy-conservation}, \ref{prop-global-momentum-conservation}, \edit{R2P1}{and Remark~\ref{rem-conservation-apart-boundary}}.
\end{proof}

\begin{proof}[Proof of Proposition~\ref{prop-quadratic-lagrangian}]
Consider the cases when $l\ne k$ and $l=k$ separately.

For $l\ne k$ the only nonvanishing contribution to the flux of $T$ comes from the edge $f=v-\mathrm{e}_k$ because $f\not\parallel \mathrm{e}_k$ otherwise. We have $\dim\mathrm{Pr}(f,k,l)=1$. Thus by Definition~\ref{def-flow} and~\eqref{eq-energy-conservation} we get 
\begin{align*}
\hspace{-0.8cm}(-1)^{l}\langle T,h\rangle_k &=\tfrac{1}{2}(-1)^{l}(-1)^{l}
\left[
T((v+\mathrm{e}_l)\times(v-\mathrm{e}_k))
+T((v+\mathrm{e}_l-2\mathrm{e}_k)\times(v-\mathrm{e}_k))
\right]\\&=
\frac{1}{2}\left[
\frac{\partial\mathcal{L}}{\partial(\delta\phi)}
(v+\mathrm{e}_l)+
\frac{\partial\mathcal{L}}{\partial(\delta\phi)}
(v+\mathrm{e}_l-2\mathrm{e}_k)
\right]
\delta\phi(v-\mathrm{e}_k).
\intertext{
\quad For $l=k$ the contribution to the flux comes from $f=v$ and $f=v-\mathrm{e}_m$ for each $m\ne k$. Thus} 
\hspace{-0.8cm}(-1)^k\langle T,h\rangle_k &=\frac{1}{2}(-1)^k(-1)^{k+1}
\left[
T(v\times v)-T((v+\mathrm{e}_k)\times (v-\mathrm{e}_k))+\sum_{m\ne k}T((v-\mathrm{e}_m)\times (v-\mathrm{e}_m))
\right]\\
&=-\frac{1}{2}
\left[
\frac{\partial\mathcal{L}}{\partial\phi}(v)
\phi(v)-
\frac{\partial\mathcal{L}}{\partial(\delta\phi)}(v+\mathrm{e}_k)[\delta\phi](v-\mathrm{e}_k)
+\sum_{m\ne k} \frac{\partial\mathcal{L}}{\partial(\delta\phi)}
(v-\mathrm{e}_m)[\delta\phi](v-\mathrm{e}_m)
\right]\\
&=\frac{1}{2}\left[
\frac{\partial\mathcal{L}}{\partial(\delta\phi)}
(v+\mathrm{e}_k)+
\frac{\partial\mathcal{L}}{\partial(\delta\phi)}
(v-\mathrm{e}_k)
\right]\delta\phi(v-\mathrm{e}_k)
-\frac{1}{2}
\frac{\partial\mathcal{L}}{\partial\phi}(v)
\phi(v)
-\frac{1}{2}\sum_{m} \frac{\partial\mathcal{L}}{\partial(\delta\phi)}
(v-\mathrm{e}_m)\delta\phi(v-\mathrm{e}_m)
\\
&=\frac{1}{2}\left[
\frac{\partial\mathcal{L}}{\partial(\delta\phi)}
(v+\mathrm{e}_k)+
\frac{\partial\mathcal{L}}{\partial(\delta\phi)}
(v+\mathrm{e}_k-2\mathrm{e}_k)
\right]\delta\phi(v-\mathrm{e}_k)
-\mathcal{L}[\phi](v).
\end{align*}
The latter equality is proved as follows. Since $\mathcal{L}$ is homogeneous quadratic, it follows that $
\frac{\partial L_v}{\partial\phi_1}\phi_1+
\frac{\partial L_v}{\partial\phi'_1}\phi'_1+\dots+
\frac{\partial L_v}{\partial\phi'_d}\phi'_d=2L_v(\phi_1,\phi'_1,\dots,\phi'_d).
$
Hence $
\frac{\partial\mathcal{L}}{\partial\phi}\frown\phi+
\frac{\partial\mathcal{L}}{\partial(\delta\phi)}\frown\delta\phi
=2\mathcal{L}[\phi],
$
as~required.
\end{proof}

\begin{proof}[Proof of Proposition~\ref{th-Maxwell-approximation}]
First note that $F_N(f)=\mathrm{F}_{mn}(\max f)\rightrightarrows\mathrm{F}_{mn}(\max h)$ on the set of all pairs $(f,h)$ having common vertices, because $\mathrm{F}_{mn}$ is continuous on $\mathrm{I}^d$, hence uniformly continuous.

\clarity{Denote $v=\max h$}. Consider the cases when $l=k$ and $l\ne k$ separately.

Assume that $l=k$. For a $1$- or $2$-dimensional face $f\subset h\perp\mathrm{e}_k$ we have $\dim\mathrm{Pr}(f,k,k)=0$. Thus
\begin{align*}
(-1)^k\langle T'_N,h\rangle_k
&=-\frac{1}{2}\left[\sum_{f:f\subset h,f\ni\red{v},\dim f=2} T'_N(f\times f)-
\sum_{f:f\subset h,f\ni\red{v},\dim f=1} T'_N((f+\mathrm{e}_k)\times(f-\mathrm{e}_k))\right]\\
&=\frac{1}{2}\left[\sum_{f:f\subset h,f\ni\red{v},\dim f=2} \#F_N(f)F_N(f)-
\sum_{f:f\subset h,f\ni\red{v},\dim f=1} \#F_N(f+\mathrm{e}_k)F_N(f-\mathrm{e}_k)\right]\\
&\rightrightarrows \frac{1}{2}\left[\sum_{m,n\ne k:m<n}\mathrm{F}^{mn}\mathrm{F}_{mn}-
  \sum_{m\ne k}\mathrm{F}^{km}\mathrm{F}_{km}
  \right](\red{v})
\\
&=  \left[\frac{1}{4}\sum_{m,n}\mathrm{F}^{mn}\mathrm{F}_{mn}-
  \sum_{m}\mathrm{F}^{km}\mathrm{F}_{km}
  \right](\red{v})
\\
&=\mathrm{T}^k_k(\red{v}).
\intertext{\quad Assume that $l\ne k$. For a $2$-dimensional face $f\parallel \mathrm{e}_k,\mathrm{e}_m$, where $m\ne k,l$, we have $\dim\mathrm{Pr}(f,k,l)=2$ or $1$ depending on if
$m$ is between $k$ and $l$ or not. Thus}
(-1)^l\langle T'_N,h\rangle_k
&=\frac{1}{2}
\sum_{\substack{f:f\subset h,f\ni\red{v},\\
                \dim f=2,f\parallel \mathrm{e}_k}}
(-1)^{\dim\mathrm{Pr}(f,k,l)}
\left[\#F_N(f+\mathrm{e}_l-\mathrm{e}_k))+
\#F_N(f+\mathrm{e}_l+\mathrm{e}_k)\right]F_N(f)
\\
&\rightrightarrows
-\sum_{m\ne k}\mathrm{sgn}(m-k)\mathrm{sgn}(m-l)
  \mathrm{F}^{\min\{l,m\},\max\{l,m\}}(\red{v})
  \mathrm{F}_{\min\{k,m\},\max\{k,m\}}(\red{v})
  \\
&=-\sum_{m\ne k}\mathrm{F}^{lm}(\red{v})\mathrm{F}_{km}(\red{v})
\\
&= \mathrm{T}^l_k(\red{v}).\\[-1.7cm]
\end{align*}
\end{proof}



\begin{proof}[Proof of Proposition~\ref{l-Wilson}]
Let $abcd$ be a face with the vertices listed in the order compatible with the positive orientation of its boundary (\clarity{see} Definition~\ref{def-boundary}), starting from the minimal vertex. Then
\begin{align*}
\mathrm{Re}\mathrm{Tr}\,[\#F^*(abcd) F(abcd)]&=\# \mathrm{Re}\mathrm{Tr}\,[(U(abc)-U(adc))^*(U(abc)-U(adc))]\\&=
\# \mathrm{Re}\mathrm{Tr}\,[U(cbabc)-U(cdabc)-U(cbadc)+U(cdadc)]\\&=
\# \mathrm{Re}\mathrm{Tr}\,[1-U(abcda)-U(abcda)^*+1]\\&=
2\#(n-\mathrm{Re}\mathrm{Tr}\,U(abcda)).
\end{align*}
Multiplying by $-1/2$ and summing over all the faces $abcd$,
we get the required expression.
\end{proof}

\begin{proof}[Proof of Corollary~\ref{cor-gauge}]
\clarity{Use notation $F$, $D^*_A$, $\mathcal S[U]$ from Definitions~\ref{def-curvature}, \ref{def-covar-boundary-gauge}, and Proposition~\ref{l-Wilson}. By Propositions~\ref{l-Bianchi} and~\ref{l-Wilson}, this notation is compatible with \eqref{eq-def-gauge}--\eqref{eq-def-unmotivated}. Then}
the Yang--Mills equation follows from Theorem~\ref{th-Euler-Lagrange-gauge} \red{and} \ref{l-differentiate-Lagrangian}; see rows \edit{R12P2}{5--6} of Table~\ref{tab-derivatives}. 
Proposition~\ref{l-gauge-covariance} and Theorem~\ref{th-charge-conservation-gauge} imply charge conservation.
\end{proof}

\begin{proof}[Proof of Corollary~\ref{prop-gauge-invariance-Wilson}]
This follows directly from Propositions~\ref{l-Wilson} and \ref{l-gauge-covariance} (see \edit{R12P2}{row~5} of Table~\ref{tab-derivatives}) because $\mathrm{Re}\,\mathrm{Tr}[j^*\frown U]$ is preserved under simultaneous gauge transformation of $U$ and~$j$.
\end{proof}

\remove{R2P1}{}

\section{Open problems}
\label{sec-open}

\begin{itemize}
  \item Expand the suggested discretization algorithm to:
      \begin{itemize}
      \item quantum field theories via path integral formalism;
      \item general relativity via discretizing the \gram{raising-index} operator $\sharp$  for nonflat spacetimes;
      \item hydrodynamics via discretizing the fluid energy-momentum tensor.
      \end{itemize}
  \item Extend the suggested discretization algorithm to involve the following conservation laws:
      \begin{itemize}
      \item energy conservation in nontrivial connection via making the \gram{cross product} gauge invariant;
      \item angular momentum conservation via discretizing the radius vector;
      \item integral-form energy conservation in general complexes via discretizing tensor integration.
      \end{itemize}
  \item Prove the conservation of the discrete covariant chiral current. Generally, is the covariant current from Theorem~\ref{th-Noether-covar} times $i$ conserved for each gauge invariant Lagrangian satisfying~\eqref{eq-invariance}?
  \item Prove analogous conservation laws in statistical field theory. E.g., is the expectation of a covariant current conserved, if the gauge group field is random with the probability density proportional to the exponential of the action from Definition~\ref{def-connection}?
  \item Apply the discretization algorithm to characteristic classes to obtain invariants of piecewise-linear homeomorphisms or rational homotopy type. 
  \item Constuct a ``second-generation'' discretization algorithm for field theories, in which not only spacetime, but also the set of field values becomes discrete; e.g., as in the Feynman checkerboard.
  \item Prove that the discussed discrete field theories approximate continuum ones in a sense. Even no analogue of Theorem~\ref{th-networks-convergence} for planar graphs with faces not being inscribed is known \cite{Chelkak-Smirnov-08,Werness-14}.
  \item State and prove a ``reciprocal Noether theorem'' giving a symmetry of the continuum limit for each discrete conservation law.
  \item Find an experimentally measurable quantity in our discretization not converging to the continuum counterpart; this would make the discretization \emph{refutable} against the continuum theory.
\end{itemize}

\comment
\subsection*{Conclusions}

We have introduced a new general discretization algorithm for field theories (see \S\ref{ssec-main-tools}), in many cases leading to both approximation of continuum theory and exact conservation laws.
The latter are produced by a new discrete Noether theorem relating them to continuum symmetries (Theorems~\ref{th-Noether} and~\ref{th-Noether-covar}). Compared to known results, the new Noether theorem is simple enough to write the resulting conservation laws explicitly as one-line formulae (using only standard topological notation) in numerous examples. Since discrete spacetime has no continuous symmetries, exact energy conservation is obtained separately by a different method not based on a symmetry, extending~\cite{Dorodnitsyn-04} (Theorems~\ref{th-energy-conservation} and~\ref{cor-free}). For that purpose, a new discretization of tensor calculus involving non-antisymmetric tensors is applied (see~\S\ref{ssec-Electrodynamics}), although it has serious limitations (see \S\ref{ssec-limitations}). Approximation of continuum theory is established in many examples (Theorem~\ref{th-networks-convergence}, Propositions~\ref{th-Maxwell-approximation},
\remove{R12P3}{}
and Remark~\ref{rem-curvature}),
some of them slightly extending the known results on electrical networks \cite{Chelkak-Smirnov-08} (see Theorem~\ref{th-networks-convergence}) and gauge theory \cite{Dimakis-etal-94} (see Corollary~\ref{cor-gauge}).
In particular, the new conserved discrete energy-momentum tensor approximates the continuum one at least for free fields.
Thus for a variety of field theories, the algorithm achieves the principles of discretization from~\S\ref{s:intro}.
\endcomment

\bigskip


\textbf{Acknowledgements.} The author is grateful to  E.Akhmedov, L.Alania, D.Arnold, A.Bossavit, V.Buchstaber, D.Chelkak, M.Chernodub, M.Desbrun, M.Gualtieri, F.G\"unther, I.Ivanov, M.Kraus, N.Mnev, F.M\"uller-Hoissen, S.Pirogov, P.Pylyavskyy, A.Rassadin, R.Rogalyov, I.Sabitov, P.Schr\"oder, I.Shenderovich, B.Springborn, A.Stern, S.Tikhomirov, S.Vergeles for useful discussions.

\ifarxiv{}{
\subsection*{Declarations}

\textbf{Funding.} The publication was prepared within the framework of the Academic Fund Program at the National Research University Higher School of Economics (HSE) in 2018-2019 (grant N18-01-0023) and by the Russian Academic Excellence Project ``5-100''. The author has also received support from the Simons--IUM fellowship.


\textbf{Data availability.} Data sharing not applicable to this article as no datasets were generated or analysed during the current study.

\textbf{Conflict of interests.} The author has no relevant financial or non-financial interests to disclose.
}

\arxiv
\appendix

\section{More examples} \label{sec-more-examples}

Here we give more examples: \S\ref{ssec-1dim} is a zero-knowledge illustration, and \S\S\ref{ssec-Klein-Gordon}--\ref{ssec-Dirac} use notions from \S\S\ref{s:intro}--\ref{sec-examples}. The proofs rely on \S\ref{s:intro}--\S\ref{s:proofs} and are given at the end of \S\ref{ssec-Klein-Gordon} and \S\ref{ssec-Dirac}.

\subsection{One-dimensional field theory}
\label{ssec-1dim}

\subsubsection*{Toy model}

Let us illustrate our main results in the trivial particular case of dimension $1$.

Consider a pipeline of $N$ identical pipes in series with sources at the two endpoints pumping incompressible fluid in and out; see Figure~\ref{fig-pipeline2} to the left. Let $s$ be the intensity of each source (measured in liters/second). The \emph{current} $j(\mathbf{k})$ through $\mathbf{k}$-th pipe (measured in litres/second) satisfies
\begin{itemize}
  \item \emph{Mass conservation law}: $j(\mathbf{1})=j(\mathbf{N})=s$ and  $j(\mathbf{k}+\mathbf{1})=j(\mathbf{k})$ for each $\mathbf{k}=\mathbf{1},\dots, \mathbf{N}-\mathbf{1}$.
\end{itemize}
This just means that $j(\mathbf{k})=s$ for $\mathbf{k}=\mathbf{1},\dots,\mathbf{N}$. Throughout \S\ref{ssec-1dim} we use bold font for pipe numbers.

Formally, we define $s\in\mathbb{R}$ to be a fixed number and the \emph{current} to be a function $j\colon\{\mathbf{1},\dots,\mathbf{N}\}\to\mathbb{R}$ satisfying the mass conservation. (There is no formal difference between symbols in different fonts.)



Let us state a least action principle for the toy model.
A \emph{potential} $\phi$ of the flow is a function $\phi\colon\{0,\dots,N\}\to\mathbb{R}$ such that $\phi(k-1)-\phi(k)=j(\mathbf{k})$ for each $k=1,\dots, N$. Clearly, it satisfies
\begin{itemize}
  \item \emph{the Laplace equation}:
   $\phi(k+1)-2\phi(k)+\phi(k-1)=0$ for each $k=1,\dots, N-1$;
  \item \emph{the least action principle}: among all functions on $\{0,\dots,N\}$, $\phi$ minimizes the functional
      \begin{equation*}
      \frac{1}{2}\sum_{k=1}^{N}(\phi(k)-\phi(k-1))^2-s\phi(0)+s\phi(N)
      \end{equation*}
\end{itemize}
The first term is the total fluid kinetic energy. The functional is the sum of the values of the function
      \begin{equation*}
      \mathcal{L}[\phi](k)=\frac{1}{2}(\underbrace{\phi(k)-\phi(k-1)}_{[\delta\phi](\mathbf{k})})^2-s(k)\phi(k),
      \quad\mbox{ where }\quad
      s(k):=\begin{cases}
        +s, & \mbox{if }k=0  \\
        0,  & \mbox{if }1\le k\le N-1 \\
        -s, & \mbox{if }k=N.
      \end{cases}
      \end{equation*}

\vspace{-0.6cm}

\subsubsection*{Generalization}

Such a ``least action'' formulation of the model has a straightforward generalization. The following definition is a particular case of Definition~\ref{def-local-particular} above.

A \emph{local Lagrangian} $\mathcal{L}$ is a self-map  of the set of all real-valued functions on $\{0,\dots,N\}$ such that
$$
\mathcal{L}[\phi](k)=L_k(\phi(k),\phi(k)-\phi(k-1))
$$
for some differentiable function $L_k\colon \mathbb{R}^2\to \mathbb{R}$. The $2$ arguments of~$L_k$ are denoted by $\phi$ and $\delta\phi$. Set
\begin{align*}
\frac{\partial\mathcal{L}[\phi]}{\partial\phi}&\colon \{0,\dots,N\}\to\mathbb{R},&
\frac{\partial\mathcal{L}[\phi]}{\partial\phi}(k)
&:= \left.\frac{\partial L_k(\phi,\delta\phi)}{\partial \phi}\right|_{\phi=\phi(k),\,\delta\phi=\phi(k)-\phi(k-1)};\\
\frac{\partial\mathcal{L}[\phi]}{\partial(\delta\phi)}&\colon \{\mathbf{1},\dots,\mathbf{N}\}\to\mathbb{R}, &
\frac{\partial\mathcal{L}[\phi]}{\partial(\delta\phi)}(\mathbf{k})
&:= \left.\frac{\partial L_k(\phi,\delta\phi)}{\partial (\delta\phi)}\right|_{\phi=\phi(k),\,\delta\phi=\phi(k)-\phi(k-1)}.
\end{align*}
We also set  $\frac{\partial\mathcal{L}}{\partial(\delta\phi)}(\mathbf{0})
=\frac{\partial\mathcal{L}}{\partial(\delta\phi)}(\mathbf{N}+\mathbf{1})=0$.
E.g., in the toy model:
$\frac{\partial\mathcal{L}[\phi]}{\partial\phi}(k)
= -s(k)$,
$\frac{\partial\mathcal{L}[\phi]}{\partial(\delta\phi)}(\mathbf{k})
= \delta\phi(\mathbf{k})$.

The following obvious proposition is a particular case of Theorem~\ref{th-Euler-Lagrange} above.

\begin{proposition}[the Euler--Lagrange equation] \label{prop-Euler-Lagrange}
Let $\mathcal{L}[\phi]$ be a local Lagrangian.
A function $\phi$ is stationary for the functional $\sum_{k=0}^N\mathcal{L}[\phi](k)$, if and only if for each $k=0,\dots,N$ we have
$$
\frac{\partial\mathcal{L}[\phi]}{\partial(\delta\phi)}(\mathbf{k})
-\frac{\partial\mathcal{L}[\phi]}{\partial(\delta\phi)}(\mathbf{k}+\mathbf{1})
+\frac{\partial\mathcal{L}[\phi]}{\partial\phi}(k)=0.
$$
\end{proposition}

E.g., in the toy model above, the Euler--Lagrange equation is the Laplace equation. That model had a built-in conservation law, hidden after the least-action formulation. The following obvious proposition reveals conservation laws hidden in the Lagrangian; it is a particular case of Theorem~\ref{th-Noether}.

\begin{proposition}[the Noether theorem]\label{prop-Noether}
If a local Lagrangian $\mathcal{L}[\phi]$ is \emph{invariant under an infinitesimal transformation} $\Delta (k)$, i.e.,
  \begin{equation*}
  \left.\frac{\partial}
  {\partial t}\mathcal{L}[\phi+ t\Delta]\right|_{ t=0}=0,
  \end{equation*}
  then for each stationary function $\phi$ for $\sum_{k=0}^N\mathcal{L}[\phi](k)$ the following function is \emph{conserved}, i.e. constant:
  \begin{equation*}
  j(\mathbf{k})=\frac{\partial\mathcal{L}[\phi]}{\partial(\delta\phi)}(\mathbf{k}) \Delta(k-1).
\end{equation*}
\end{proposition}

E.g., in the above toy model, apart the endpoints, the Lagrangian is invariant under the transformation $\phi\mapsto\phi-t$, where $t\in\mathbb{R}$. 
The resulting Noether conserved function is exactly $j(\mathbf{k})=\phi(k-1)-\phi(k)$.



\subsubsection*{Momentum conservation}

Let us state a less intuitive momentum conservation. The introduced discrete momentum tensor is a completely new object. First we give a heuristic motivation (cf.~\S\ref{ssec-Networks}), then a precise definition.

In the toy model above, momentum circulation is physically clear. The momentum of the fluid in the pipe $\mathbf{k}$ is proportional to $j(\mathbf{k})$. During time $\Delta t$, the volume proportional to $j(\mathbf{k})\Delta t$ moves to the next pipe. Thus the momentum flux through the vertex $k$ per unit time is proportional to $j(\mathbf{k})^2$. (We ignore pressure and do not care of the proportionality constant because this is just a heuristic anyway.)

Now consider a \emph{free field}, i.e., $\mathcal{L}[\phi](k)=[\delta\phi](\mathbf{k})^2+m^2\phi(k)^2$, where $m\ge 0$. Let $\phi$ be a stationary function, i.e. just a function satisfying the equation $\phi(k-1)-(2+m^2)\phi(k)+\phi(k+1)=0$ for each $0<k<N$. One expects the following properties of the momentum flux $\sigma(k)$ through a vertex $k$:
\begin{itemize}
\item
$\sigma(k)=j(\mathbf{k})^2$ for $m=0$, i.e., for a linear potential $\phi$;
\item
$\sigma(k)$ depends only on $\phi(k)$, $\delta\phi(\mathbf{k})$, $\delta\phi(\mathbf{k}+\mathbf{1})$, and is homogeneous quadratic in these values;
\item
$\sigma(k)=\mathrm{const}$ \ apart the endpoints, i.e., the momentum is conserved.
\end{itemize}
The simplest function $\sigma(k)$ satisfying these properties is (we skip a direct checking)
$$
\sigma(k)=\delta\phi(\mathbf{k}+\mathbf{1})\delta\phi(\mathbf{k})-m^2\phi(k)^2
=
\frac{1}{2}\frac{\partial\mathcal{L}}{\partial(\delta\phi)}(\mathbf{k}+\mathbf{1})
\delta\phi(\mathbf{k})-
\frac{1}{2}\frac{\partial\mathcal{L}}{\partial\phi}(k)\phi(k).
$$ 
\begin{remark}\label{rem-naive}
  A naive way to discretize the momentum flux would be to take the usual (continuum) momentum flux of a piecewise-linear extension of $\phi$. But the resulting quantity is not conserved in a reasonable sense.
  Our function $\sigma(k)$ is very different from such naive ``finite-element'' discretization.
\end{remark}

For an arbitrary Lagrangian, the formula for $\sigma(k)$ is not applicable literally but still suggestive. Since the formula involves the product of the values of $\delta\phi$ at distinct edges, it is reasonable to view it as a ``projection'' of a more fundamental quantity defined on the Cartesian square of the pipeline.



\begin{definition} \label{def-tensor-1D}
(This is a particular case of Definition~\ref{def-tensor}.)
The \emph{Cartesian square} of a path with $N$ edges is the grid $N\times N$; see Figure~\ref{fig-notation-cross}.
The vertices of the grid have form $k\times l$, where $k$ and $l$ are vertices of the path. The $1\times 1$ squares have form $\mathbf{k}\times \mathbf{l}$, where $\mathbf{k}$ and $\mathbf{l}$ are edges.

For functions $\psi,\phi$ on the set of vertices (respectively, edges) of the path denote by $\psi\times \phi$ the function on the vertices (respectively, $1\times 1$ squares) of the grid given by $[\psi\times \phi](k\times l)=\psi(k)\phi(l)$
(respectively, by $[\psi\times \phi](\mathbf{k}\times \mathbf{l})=\psi(\mathbf{k})\phi(\mathbf{l})$). A real-valued function on the disjoint union of the sets of vertices and $1\times 1$ squares of the grid is a \emph{type $(1,1)$ tensor}. (E.g., for the toy model, equation~\eqref{eq-energy-conservation} gives the tensor equal $s^2$ on each $1\times 1$ square and vanishing on each non-boundary vertex.)

A tensor $T$ is \emph{conserved}, if for each $0<k< N$ and $0<l\le N$ the following equation holds:
$$
T(k\times l)-T(k\times(l-1)) +T(\mathbf{k}\times\mathbf{l})
-T((\mathbf{k}+\mathbf{1})\times\mathbf{l})=0.
$$
I.e., we have one equation per \emph{vertical} non-boundary edge; see Figure~\ref{fig-notation-cross}.
\end{definition}

The following obvious corollary of Proposition~\ref{prop-Euler-Lagrange} is a particular case of Theorem~\ref{th-energy-conservation}.

\begin{proposition}[Momentum conservation] Let $\mathcal{L}[\phi]$ be a local Lagrangian and $\phi$ be a stationary function for the functional $\sum_{k=0}^{N}\mathcal{L}[\phi](k)$. Then the tensor given by~\eqref{eq-energy-conservation} is conserved.
\end{proposition}

Define the \emph{flux} of a tensor $T$ through a vertex $k$ by the formula $\frac{1}{2}T((\mathbf{k}+\mathbf{1})\times \mathbf{k})- \frac{1}{2}T(k\times k)$. E.g., for the free field, 
the flux of tensor~\eqref{eq-energy-conservation} equals exactly $\sigma(k)$. A tensor $T$ is \emph{symmetric}, if $T(k\times l)=T(l\times k)$ for all vertices or edges $k,l$. E.g., tensor \eqref{eq-energy-conservation} is symmetric essentially \emph{only} for the free field (despite being a tensor on $1$-dimensional spacetime). A conserved symmetric tensor has constant flux (this is a version of  Theorem~\ref{prop-global-momentum-conservation} above).
E.g., for the toy model, the flux of tensor~\eqref{eq-energy-conservation} is $j(\mathbf{k})^2/2$.


\subsection{The Klein--Gordon field}\label{ssec-Klein-Gordon}

The classical (not quantum) Klein--Gordon field does not describe a real physical field 
but serves as an example for more realistic models. Corollaries~\ref{cor-Klein-Gordon-charge-conservation}, \ref{cor-Klein-Gordon-charge-conservation-gauge},
\ref{cor-Klein-Gordon-full} and Proposition~\ref{th-Klein-Gordon-approximation} are new.

\subsubsection*{Basic model}

\begin{definition}
Fix a number $m\ge 0$ called \emph{particle mass}.
A complex-valued function $\phi$ on the set of vertices of $I^d_N$ is a \emph{Klein--Gordon field} of mass $m$, if the following equation holds apart $\partial I^d_N$:
\begin{itemize}
  \item \emph{the Klein--Gordon equation:} $-\partial\#\delta\phi+m^2\phi=0$.
\end{itemize}
\end{definition}

\begin{corollary} 
  \label{cor-Klein-Gordon}
  A complex-valued function $\phi$ on vertices of $I^d_N$ 
  is a Klein--Gordon field, if and only if
  among all the functions with the same values on $\partial I^d_N$, the function $\phi$ is stationary for the functional $\mathcal{S}[\phi]=\epsilon\mathcal{L}[\phi]$, where
  $$
  \mathcal{L}[\phi]=\mathop{\#}\!\delta \phi \frown\delta \phi^*-m^2\phi\frown \phi^*.
  $$
\end{corollary}

Here we impose a boundary condition, because the theory becomes trivial otherwise. The Lagrangian $\mathcal{L}[\phi]$ is \emph{globally gauge invariant}, i.e., $\mathcal{L}[\phi g]=\mathcal{L}[\phi]$ for each $g\in\mathbb{C}$ with $|g|=1$.

\begin{corollary}[Charge, energy, momentum conservation] \label{cor-Klein-Gordon-charge-conservation}
For a Klein--Gordon field $\phi$
the current  $j[\phi]:=-2\mathrm{Im}(\#\delta\phi^*\frown \phi)$
and the tensor $T[\phi]:=2\mathrm{Re}[\#\delta\phi^*\times \delta\phi-m^2\phi^*\times\phi]$
are conserved apart~$\partial I^d_N$.
\end{corollary}

\subsubsection*{Approximation}

The resulting current $j[\phi]$ and energy-momentum tensor $T[\phi]$ indeed approximate continuum ones.

In continuum theory, $\phiup$ is a smooth complex-valued function defined on 
$\mathrm{I}^d$. (Hereafter \emph{smooth} means $C^1$, and the derivative at the boundary $\partial\mathrm{I}^d$ means a one-sided derivative.) The \emph{current} and \emph{energy-momentum tensor} of $\phiup$ (for the metric signature $(+,-,...,-)$) are the vector and matrix fields
\begin{equation*}
\mathrm{j}^l=-2\mathrm{Im}\left[\phiup\,\partial^l \phiup^*\right]
\quad\mbox{ and }\quad
\mathrm{T}_k^l=2\mathrm{Re}\left[\partial^l \phiup^*\partial_k \phiup\right]
+\deltaup_k^l\left[-\partial^n \phiup^*\partial_n \phiup+m^2\phiup^*\phiup\right],
\vspace{-0.3cm}
\end{equation*}
where summation over $n$ is understood, and we denote
$\partial_n\phiup:=\frac{\partial \phiup}{\partial\mathrm{x}^n}$,  $\partial^n\phiup:=\begin{cases}
  +\partial_n\phiup, & \mbox{if $n=0$}; \\
  -\partial_n\phiup, & \mbox{if $n\ne 0$}.
\end{cases}$
\vspace{-0.3cm}

\begin{proposition}[Approximation property] \label{th-Klein-Gordon-approximation}
Let $\phiup$ be a smooth complex-valued field on $\mathrm{I}^d$. Dissect $\mathrm{I}^d$ into $N^d$ equal hypercubes and take the discrete field $\phi_N(v):=\phiup(v)$ on the vertices $v$ of the resulting grid. Let $\mathrm{j}^l$, $\mathrm{T}^l_k$ be the continuous current and energy-momentum tensor. Define $j_N=j[\phi_N]$, $T_N=T[\phi_N]$ by the same formulae as in Corollary~\ref{cor-Klein-Gordon-charge-conservation} except that $m$ is replaced by $m/N$. Take $0\le k,l<d$.
Then on the set of all edges $e\parallel \mathrm{e}_l$ and all hyperfaces $h\perp \mathrm{e}_l$ disjoint with $\partial \mathrm{I}^d$, we have
\begin{equation*}
Nj_N(e)\rightrightarrows \mathrm{j}^l(\max e)
\quad\mbox{ and }\quad
(-1)^l N^2 \langle T_N,h\rangle_k\rightrightarrows \mathrm{T}^l_k(\max h)\qquad\mbox{as }N\to\infty.
\end{equation*}
\end{proposition}

\begin{remark} \label{rem-weekness}
  The fields $\phiup$ and $\phi_N$ are \emph{not} necessarily Klein--Gordon fields (and typically $\phi_N$ cannot be such one, even $\phiup$ is). In particular, $j[\phi_N]$ and $T[\phi_N]$ are \emph{not} necessarily conserved.
%
\end{remark}

\subsubsection*{Coupling to a gauge field}

Interaction with a gauge field is introduced by replacement of (co)boundary by covariant (co)boundary. 
Let $U\in C^1({I}^d_N;G)$, $A=U-1$, $F$ be a gauge group field, the connection, and the curvature respectively.
Hereafter $C^k({I}^d_N;V)$ is the set of $V$-valued functions on the set of $k$-dimensional faces of ${I}^d_N$.

\begin{definition} \label{def-covar-boundary}
\label{def-covar-coboundary} 
For $\phi\in C^0(I^d_N;\mathbb{C}^{1\times n})$, $g\in C^0(I^d_N;G)$, and $\psi\in C_1(I^d_N;\mathbb{C}^{1\times n})$,
the \emph{gauge transformation} of $\phi$ by $g$, the \emph{covariant coboundary} of $\phi$, and the \emph{covariant boundary} of $\psi$ are the functions on vertices $b$ or edges $ab$, where $a<b$, given by
\begin{align*}
[\phi\smile g](b)&:=\phi(b)g(b),\\
[D_A\phi](ab)    &:=\phi(b)-\phi(a)U(ab),\\
[D^*_A\psi](b)   &:=\sum\nolimits_{\text{edges }
ab\text{ ending at }b}\psi(ab)-\sum\nolimits_{\text{edges }
bc\text{ starting at }b}\psi(bc)U(cb)
\end{align*}
A field $\phi\in C^0({I}^d_N;\mathbb{C}^{1\times n})$
is a \emph{Klein--Gordon field interacting with the gauge field}, if 
apart $\partial I^d_N$ we have
\begin{itemize}
  \item \emph{the Klein--Gordon equation in a gauge field}: $-D^*_A\#D_A\phi+m^2\phi=0$.
\end{itemize}
\end{definition}

\begin{corollary} 
  \label{cor-Klein-Gordon-gauge}
  A function $\phi\in C^0({I}^d_N;\mathbb{C}^{1\times n})$ is a Klein--Gordon field interacting with a gauge group field $U\in C^1({I}^d_N;G)$, if and only if among all the functions with the same values on $\partial I^d_N$, the function $\phi$ is stationary for the functional $\mathcal{S}[\phi,U]=\epsilon\mathcal{L}[\phi,U]$ for fixed $U$, where
  $$
  \mathcal{L}[\phi,U]=\#D_A \phi\frown (D_A \phi)^*-m^2\phi\frown \phi^*-\tfrac{1}{2}\mathrm{Re}\,\mathrm{Tr}[\#F^*\frown F].
  $$
  \end{corollary}

\begin{remark}
  Using row-vectors $\phi$ rather than column-vectors is essential to make $\mathcal{L}[\phi,U]$ a local Lagrangian with respect to the gauge group field $U$ as well. The third summand in $\mathcal{L}[\phi,U]$ can be dropped for fixed $U$ but becomes essential for dynamic $U$ in Corollary~\ref{cor-Klein-Gordon-full}.
\end{remark}

In what follows we use (co)chain products defined in Table~\ref{tab-products}; cf.~Definition~\ref{def-cap-general}.

\begin{corollary}[Gauge invariance] \label{prop-gauge-invariance-Klein-Gordon}
The Lagrangian $\mathcal{L}[\phi,U]$ from Corollary~\ref{cor-Klein-Gordon-gauge} is \emph{gauge invariant}, i.e., $\mathcal{L}[\phi\smile g,g^*\smile U \smile g]=\mathcal{L}[\phi,U]$ for each $\phi\in C^0(I^d_N;\mathbb{C}^{1\times n})$, $U\in C^1(I^d_N;G)$, $g\in C^0(I^d_N;G)$.
\end{corollary}


\begin{corollary}[Charge conservation] 
\label{cor-Klein-Gordon-charge-conservation-gauge}
For a Klein--Gordon field $\phi$ interacting with a gauge group field $U$
the covariant current $j[\phi,U]=-2\phi^*\smile \#D_A\phi\in C^1({I}^d_N;\mathbb{C}^{n\times n})$ is conserved apart $\partial I^d_N$, i.e.,  $D^*_A\mathrm{Pr}_{T_U G} j[\phi,U]=0$ apart $\partial I^d_N$. (Beware that the product of a column- and a row-vector is a~matrix.)
\end{corollary}

\begin{corollary}
\label{cor-Klein-Gordon-full}
A gauge group field $U$ is stationary for the functional
$\mathcal{S}[\phi,U]$ from Corollary~\ref{cor-Klein-Gordon-gauge}
for fixed $\phi\in C^0(I^d_N;\mathbb{C}^{1\times n})$, if and only if $U$ is generated by
the covariant current from Corollary~\ref{cor-Klein-Gordon-charge-conservation-gauge}.
\end{corollary}


\subsubsection*{Proofs}

\begin{proof}[Proof of Corollary~\ref{cor-Klein-Gordon}]
This follows from a version of Theorem~\ref{th-Euler-Lagrange} for
complex-valued fields and nonfree boundary conditions and Proposition~\ref{l-differentiate-Lagrangian} for $A=0$; see rows 3--4 of Table~\ref{tab-derivatives}.
\end{proof}

\begin{proof}[Proof of Corollary~\ref{cor-Klein-Gordon-charge-conservation}]
Since $\mathcal{L}[\phi]$ is globally gauge invariant, it follows that~\eqref{eq-invariance} holds for
$\Delta=i\phi$. By the versions of Theorems~\ref{th-Noether} and~\ref{th-energy-conservation} for
complex-valued fields and nonfree boundary conditions, it follows that the real parts of~\eqref{eq-Noether-current}--\eqref{eq-energy-conservation} are conserved apart the boundary, as required.
\end{proof}

\begin{proof}[Proof of Proposition~\ref{th-Klein-Gordon-approximation}]
Since $\phiup$ is $C^1$, we get $N[\delta\phi_N](e)\rightrightarrows \partial_l\phiup(\max e)$, $\phi_N(\min e)\rightrightarrows\phiup(\max e)$,
$$
Nj_N(e)=-2N\,\mathrm{Im}[\#\delta\phi_N^*\frown \phi_N](e)
\rightrightarrows -2\,\mathrm{Im}\left[
\partial^l\phiup^*(\max e)\phiup(\max e)
\right]
= \mathrm{j}^l(\max e).
$$
For $v\!:=\!\max h$, by a version of Proposition~\ref{prop-quadratic-lagrangian} for $\mathbb{C}$-valued fields and rows 3--4 of Table~\ref{tab-derivatives}, we get
\begin{align*}
  (-1)^lN^2\langle T_N,h\rangle_k
  &=
  N^2\mathrm{Re}\left[
  \left(\#\delta\phi_N^*(v+\mathrm{e}_l)+
  \#\delta\phi_N^*(v+\mathrm{e}_l-2\mathrm{e}_k)\right)
  \delta\phi_N(v-\mathrm{e}_k)
    \right]+\\&+
  N^2\deltaup^l_k
  \left[
  -\#\delta\phi_N\frown \delta\phi_N^* +\tfrac{m^2}{N^2}\phi_N\frown\phi_N^*
  \right](v)\\
  &\rightrightarrows 2\mathrm{Re}\left[\partial^l \phiup^*\partial_k \phiup\right](v)
  +\deltaup^l_k\left[-\partial^n \phiup^*\partial_n \phiup+m^2\phiup^*\phiup\right](v)
  =\mathrm{T}^l_k(v).\\[-1.6cm]
\end{align*}
\end{proof}

\begin{proof}[Proof of Corollary~\ref{cor-Klein-Gordon-gauge}]
Drop the last term (not depending on $\phi$) from the Lagrangian $\mathcal{L}[\phi,U]$. Then by a version of Theorem~\ref{th-Euler-Lagrange-covar} and rows 3--4 of Table~\ref{tab-derivatives} the corollary follows.
\end{proof}

\begin{proof}[Proof of Corollary~\ref{prop-gauge-invariance-Klein-Gordon}]
This follows directly from Proposition~\ref{l-gauge-covariance}; see rows 3--4 of Table~\ref{tab-derivatives}.
\end{proof}

\begin{proof}[Proof of Corollary~\ref{cor-Klein-Gordon-charge-conservation-gauge}]
This follows from Corollary~\ref{prop-gauge-invariance-Klein-Gordon}, a version of Theorem~\ref{th-charge-conservation} for nonfree boundary conditions, row 4 of Table~\ref{tab-derivatives}, and the formula for $(\phi\smile\psi)^*$ from Lemma~\ref{l-identities}.
\end{proof}

\begin{proof}[Proof of Corollary~\ref{cor-Klein-Gordon-full}]
For fixed $\phi$, 
the Lagrangian $\mathcal{L}[\phi,U]=:\mathcal{L}[U]$ from Corollary~\ref{cor-Klein-Gordon-gauge} is local with respect to $U$. By Lemma~\ref{l-lagrangian-functional-derivative-charge} and row~6 of Table~\ref{tab-derivatives} we get
$\frac{\partial\mathcal{L}[U]}{\partial U}=j[\phi,U]^*$ and $\frac{\partial\mathcal{L}[U]}{\partial (F[U])}=\# F^*$, where $j[\phi,U]$ is given by Corollary~\ref{cor-Klein-Gordon-charge-conservation-gauge}.
Let $U_0$ be stationary for the functional $\mathcal{S}[\phi,U]=\epsilon \mathcal{L}[U]$.
By Theorem~\ref{th-Euler-Lagrange-gauge}, $U_0$ satisfies the Yang--Mills equation from Corollary~\ref{cor-gauge} with $j=j[\phi,U_0]$.
Then again by Theorem~\ref{th-Euler-Lagrange-gauge} $U_0$, is stationary for $\mathcal{S}[U]$ from Proposition~\ref{l-Wilson}, where $j=j[\phi,U_0]$ is fixed (i.e., one keeps $j=j[\phi,U_0]$ rather than $j=j[\phi,U]$ under a variation of $U$). Thus by Definition~\ref{def-gauge}, $U_0$ is generated by $j[\phi,U_0]$. The reciprocal assertion is proved analogously.
\end{proof}

\subsection{The Dirac field} \label{ssec-Dirac}

A classical (not quantum) Dirac field
describes the wave function of an electron in quantum-mechanics (not quantum field theory). Our discretization is equivalent to~\cite[(5.19)]{Creuz-70} but not to~\cite[(5.55)]{Creuz-70}.
In this subsection, the ``topological'' notation seems to be less clear than the original ``coordinate'' one \cite{Creuz-70}, but we keep the former for sameness. Corollaries~\ref{cor-probability-conservation}, \ref{cor-Dirac-charge-conservation}, \ref{cor-Dirac-full}, and Proposition~\ref{th-Dirac-approximation} are new.

\subsubsection*{Basic model}

\begin{definition} \label{def-Dirac}
Introduce the \emph{Dirac $\gamma$-matrices} (\emph{generators of the Clifford algebra of $\mathbb{R}^{1,3}$}):
$$
\gamma^0=
\left(\begin{smallmatrix}
  1 & 0 & 0 & 0 \\
  0 & 1 & 0 & 0 \\
  0 & 0 & -1 & 0 \\
  0 & 0 & 0 & -1
\end{smallmatrix}\right),\quad
\gamma^1=
\left(\begin{smallmatrix}
        0 & 0 & 0 & 1 \\
        0 & 0 & 1 & 0 \\
        0 & -1 & 0 & 0 \\
        -1 & 0 & 0 & 0
      \end{smallmatrix}\right),\quad
\gamma^2=
\left(\begin{smallmatrix}
        0 & 0 & 0 & -i \\
        0 & 0 & i & 0 \\
        0 & i & 0 & 0 \\
        -i & 0 & 0 & 0
      \end{smallmatrix}\right),\quad
\gamma^3=
\left(\begin{smallmatrix}
        0 & 0 & 1 & 0 \\
        0 & 0 & 0 & -1 \\
        -1 & 0 & 0 & 0 \\
        0 & 1 & 0 & 0
      \end{smallmatrix}\right).\quad
$$

The \emph{Dirac chain} $\gamma\in C_1(I^4_N;\mathbb{C}^{4\times 4})$ is given by $\gamma(e)=\gamma^k$ for each edge $e\parallel \mathrm{e}_k$, where $k=0,1,2,3$.

A function $\psi\in C^0(I^4_N;\mathbb{C}^{4\times 1})$ is a \emph{Dirac field} of mass $m$, if the following equation holds apart $\partial I^4$:
\begin{itemize}
  \item \emph{the Dirac equation}: $i\gamma\frown \delta\psi+i\gamma\spleen \delta\psi-2m\psi=0$.
\end{itemize}
\end{definition}

Such form of the equation, with the Dirac chain appearing twice,
is forced by the following variational principle and
is a manifestation of  lattice \emph{fermion doubling} phenomenon. Set $\bar\psi:=\psi^*\gamma^0$.

\begin{corollary} 
  \label{cor-Dirac}
  A function $\psi\in C^0(I^4_N;\mathbb{C}^{4\times 1})$ is a Dirac field, if and only if
  among all the fiunctions with the same values on $\partial I^4_N$,
  the function $\psi$ is stationary for the functional $\mathcal{S}[\psi]=\epsilon\mathcal{L}[\psi]$, where
  $$
  \mathcal{L}[\psi]=\mathrm{Re}\left[\bar\psi\frown (i\gamma\frown\delta\psi-m\psi)\right].
  $$
\end{corollary}

  Using column-vectors $\psi$ rather than row-vectors is essential to make the expression meaningful.

The \emph{doubling} of the $d$-dimensional grid $I^d_N$ is defined analogously to Definition~\ref{def-doubling}.

\begin{proposition}\label{prop-fermion-doubling}
Consider a Dirac field on the doubling of $I^4_N$. Then the restriction of the field to the initial grid $I^4_N$ besides the boundary satisfies the Klein--Gordon equation with twice larger mass.
\end{proposition}

The Lagrangian $\mathcal{L}[\psi]$ is globally gauge invariant: $\mathcal{L}[\psi g]=\mathcal{L}[\psi]$ for each
$g\in\mathbb{C}$ with $|g|=1$. In the case $m=0$ there is also a symmetry $\mathcal{L}[e^{i\gamma^5t}\psi]=\mathcal{L}[\psi]$ for each $t\in\mathbb{R}$, where $\gamma^5:=i\gamma^0\gamma^1\gamma^2\gamma^3$.


\begin{corollary}[Current, chiral current, energy, momentum conservation]\label{cor-probability-conservation}
For a Dirac field $\psi$
the following current and tensor are conserved apart $\partial I^4_N$:
$$j[\psi]=\mathrm{Re}\left[\bar\psi\smile \gamma\smile\psi\right]
\quad\mbox{ and }\quad
T[\psi]=\mathrm{Re}[(\bar\psi\spleen i\gamma)\times \delta\psi-(\delta\bar\psi\frown i\gamma+2m\bar\psi)\times \psi].$$
In the case when $m=0$ the current $j^5[\psi]=\mathrm{Re}\left[\bar\psi\smile \gamma^5\gamma\smile\psi\right]$ is also conserved apart $\partial I^4_N$.
\end{corollary}

\begin{remark} Unlike continuum theory, $j[\psi](e)$ is not necessarily positive on edges $e\parallel (1,0,0,0)$ (because $\psi$ and $\bar\psi$ are evaluated at distinct endpoints of $e$) and thus cannot be interpreted as probability.

The tensor $T[\psi]$ is not partially symmetric. Thus we know no integral form of its conservation.
\end{remark}


\subsubsection*{Approximation}

The resulting current and energy-momentum tensor indeed approximate the continuum ones.

In continuum theory, $\psiup\colon \mathrm{I}^4\to\mathbb{C}^4$ is a smooth function. The \emph{current} and the (\emph{canonical}) \emph{energy-momentum tensor} of $\psiup$ are the vector and matrix fields
\begin{equation*}
\mathrm{j}^l=\mathrm{Re}\left[\bar\psiup\,\gamma^l \psiup\right]
\quad\mbox{ and }\quad
\mathrm{T}_k{}^l=\mathrm{Re}\left[i\bar\psiup\gamma^l \partial_k \psiup
-\deltaup^l_k\left(i\bar\psiup\gamma^n \partial_n \psiup-m\bar\psiup\psiup\right)\right],
\end{equation*}
where summation over $n$ is understood. In what follows analogues of Remarks~\ref{rem-renormalization} and~\ref{rem-weekness} apply.

\begin{proposition}[Approximation property] \label{th-Dirac-approximation}
Let $\psiup\colon \mathrm{I}^4\to\mathbb{C}^4$ be a smooth function. Dissect $\mathrm{I}^4$ into $N^4$ equal hypercubes and define the discrete field $\psi_N(v):=\psiup(v)$ on the vertices $v$ of the resulting grid. Let $\mathrm{j}^l$, $\mathrm{T}_k{}^l$ be the continuous current and energy-momentum tensor. Define $j_N=j[\phi_N]$, $T_N=T[\phi_N]$ by the same formulae as in Corollary~\ref{cor-probability-conservation} except that $m$ is replaced by $m/N$. Take $0\le k,l<4$.
Then on the set of all edges $e\parallel \mathrm{e}_l$ and hyperfaces $h\perp \mathrm{e}_l$ not intersecting $\partial \mathrm{I}^4$, we have
\begin{equation*}
j_N(e)\rightrightarrows \mathrm{j}^l(\max e)
\quad\mbox{ and }\quad
(-1)^l N \langle T_N,h\rangle_k\rightrightarrows \mathrm{T}_k{}^l(\max h) \qquad\mbox{as }N\to\infty.
\end{equation*}
\end{proposition}


\subsubsection*{Coupling to a gauge field}

\begin{definition} Let 
$U\in C^1(I^4_N;G)$ be a gauge group field. Assume that $n\ne 4$ to avoid notational conflict. The \emph{covariant coboundary} $D_A\psi$ of $\psi\in C^k(I^4_N;\mathbb{C}^{4\times n})$ is defined literally as for $\psi\in C^k(I^4_N;\mathbb{C}^{1\times n})$. Set
\begin{equation}\label{eq-covar-coboundary-reversed}
  \bar D_A\psi=(\delta \psi^*+A\smile \psi^*)^*.
\end{equation}
A function $\psi\in C^0(I^4_N;\mathbb{C}^{4\times n})$ is a \emph{Dirac field interacting with the gauge field}, if apart $\partial I^4_N$ we have
\begin{itemize}
  \item \emph{the Dirac equation in a gauge field}: $i\gamma\frown D_A\psi+i\gamma\spleen  \bar D_A\psi-2m\psi=0$.
\end{itemize}
\end{definition}


\begin{corollary} 
  \label{cor-Dirac-gauge}
  A function $\psi\in C^0(I^4_N;\mathbb{C}^{4\times n})$ is a Dirac field interacting with a gauge group field $U\in C^1(I^4_N;G)$, if and only if among all functions with the same values on $\partial I^4_N$, the  function $\psi$ is stationary for the functional $\mathcal{S}[\psi,U]=\epsilon\mathcal{L}[\psi,U]$ for fixed $U$, where
  $$
  \mathcal{L}[\psi,U]=\mathrm{Re}\,\mathrm{Tr}\,\left[\bar\psi\frown (i\gamma\frown D_A\psi-m\psi)-\tfrac{1}{2}\#F^*\frown F\right].
  $$
\end{corollary}

\begin{corollary}[Gauge invariance] \label{prop-gauge-invariance-Dirac}
The Lagrangian $\mathcal{L}[\psi,U]$ in Corollary~\ref{cor-Dirac-gauge} is gauge invariant.
\end{corollary}

\begin{corollary}[Charge conservation]
\label{cor-Dirac-charge-conservation}
For a Dirac field $\psi$ interacting with a gauge field $U$, the covariant current $j[\psi]=-\bar\psi\smile i\gamma\smile\psi\in C^1(I^4_N;\mathbb{C}^{n\times n})$ is conserved, i.e., $D^*_A\mathrm{Pr}_{T_U G} j[\psi]\!\!=\!\!0$ apart~$\partial I^4_N$.
In particular, its edgewise product with $iU$ is conserved, i.e., $\partial\langle j[\psi],iU\rangle=0$.
\end{corollary}

\begin{corollary}
\label{cor-Dirac-full}
A gauge group field $U$ is stationary for the functional
$\mathcal{S}[\psi,U]$ from Corollary~\ref{cor-Dirac-gauge}
for fixed $\psi\in C^0(I^4_N;\mathbb{C}^{4\times n})$,
if and only if $U$ is generated by
the covariant current from Corollary~\ref{cor-Dirac-charge-conservation}.
\end{corollary}





\subsubsection*{Proofs}

All the results of \S\ref{ssec-general-connections} remain true for fields $\psi\in C^0({I}^4_N;\mathbb{C}^{4\times n})$, with analogous definitions and proofs. The following lemma is proved by direct checking analogously to Proposition~\ref{l-differentiate-Lagrangian}.

\begin{lemma}\label{l-Dirac-derivatives} Drop the last term (not depending on $\psi$) from the Lagrangian $\mathcal{L}[\psi,U]$ in Corollary~\ref{cor-Dirac-gauge}. The resulting Lagrangian is local and has partial derivatives
$$
\left(\frac{\partial\mathcal{L}}{\partial\psi}\right)^*
= i\gamma^0\gamma\frown D_A\psi - 2m\gamma^0\psi
\qquad\text{and}\qquad
\left(\frac{\partial\mathcal{L}}{\partial(D_A\psi)}\right)^*
=-i\gamma^0\gamma\smile \psi.
$$
\end{lemma}

%

\begin{proof}[Proof of Corollaries~\ref{cor-Dirac} and~\ref{cor-Dirac-gauge}]
Let us prove Corollary~\ref{cor-Dirac-gauge}; \ref{cor-Dirac} is a particular case. 
By a version of Theorem~\ref{th-Euler-Lagrange-covar}, a field $\psi\in C^0(I^4_N;\mathbb{C}^{4\times n})$ is stationary for 
$\mathcal{S}[\psi]$, if and only if 
\begin{multline*}
0=D^*_A\left(\tfrac{\partial\mathcal{L}}{\partial(D_A\psi)}\right)^*+
\left(\tfrac{\partial\mathcal{L}}{\partial\psi}\right)^*
=D^*_A(-i\gamma^0\gamma\smile \psi)+i\gamma^0\gamma\frown D_A\psi-2m\gamma^0\psi=\\
=i\gamma^0\gamma\spleen \bar D_A\psi+i\gamma^0\gamma\frown D_A\psi-2m\gamma^0\psi.
\end{multline*}
Left-multiplying by $(\gamma^0)^{-1}$, we get the Dirac equation in a gauge field. Here the 1st equality 
is obtained by Lemma~\ref{l-Dirac-derivatives} 
and the 2nd one follows from
\begin{multline*}
(D^*_A(\gamma\smile\psi))^*
=\partial(\psi^*\spleen\gamma^*)
- A\spleen (\psi^*\spleen\gamma^*)
=\\=\psi^*\spleen\partial\gamma^*-\delta\psi^*\spleen\gamma^*
- (A\smile\psi^*)\spleen\gamma^*
=-(\gamma\spleen\bar D_A\psi)^*,
\end{multline*}
where we used the obvious identity $\partial\gamma^*=0$, equations \eqref{eq-covar-boundary}--\eqref{eq-covar-coboundary-reversed},
and Lemma~\ref{l-identities}.
\end{proof}

\begin{proof}[Proof of Proposition~\ref{prop-fermion-doubling}]
Let the \emph{Dirac operator} on the doubling act by $\not\!\!\partial\psi=\gamma \frown\delta\psi+\gamma \spleen \delta\psi$ for each $\mathbb{C}^{4\times 1}$-valued field $\psi$ on the vertices of the doubling.
Then the Dirac equation is $i\!\!\not\!\!\partial \psi-2m\psi=0$. Applying the operator $i\!\!\not\!\!\partial+2m$ to the left-hand side and canceling the $\pm im\!\!\not\!\!\partial$-terms we get
 $\not\!\!\partial\!\!\not\!\!\partial\psi+4m^2\psi\!=\!0$. It remains to prove the identity  $\not\!\!\partial\!\!\not\!\!\partial=-\partial_{\text{initial}}\#\delta_{\text{initial}}$, where $\partial_{\text{initial}}$ and $\delta_{\text{initial}}$ are the boundary and coboundary operators respectively on the initial grid $I^4_N$.

Take a non-boundary vertex $v$ of $I^4_N$. By the identity $\gamma^k\gamma^l+\gamma^l\gamma^k=2g^{kl}:=\deltaup^{kl}(4\deltaup^{k0}-2)$
we get
\vspace{-0.3cm}
\begin{align*}
  [\not\!\!\partial\!\!\not\!\!\partial\psi](v)
  &= \sum\nolimits_{k=0}^3 \gamma^{k} [\not\!\!\partial\psi(v+\mathrm{e}_k)-
  \not\!\!\partial\psi(v-\mathrm{e}_k)]\\
  &= \sum\nolimits_{k,l=0}^3 \gamma^{k}\gamma^{l}
  [\psi(v+\mathrm{e}_k+\mathrm{e}_l)-
  \psi(v+\mathrm{e}_k-\mathrm{e}_l)-
  \psi(v-\mathrm{e}_k+\mathrm{e}_l)+
  \psi(v-\mathrm{e}_k-\mathrm{e}_l)]
  \\
  &=
  \sum\nolimits_{k=0}^3 g^{kk}[\psi(v+2\mathrm{e}_k)-2\psi(v)
  +\psi(v-2\mathrm{e}_k)]
  =-[\partial_{\text{initial}}\#\delta_{\text{initial}}\psi](v).\\[-1.6cm]
\end{align*}
\end{proof}

\vspace{-0.2cm}
\begin{remark}
  The actions from Corollaries~\ref{cor-Dirac} and~\ref{cor-Dirac-gauge} can be written as $\mathcal{S}[\psi,U]=\epsilon\mathcal{L}'[\psi,U]$,
  where $$\mathcal{L}'[\psi,U]=\mathrm{Re}\,\mathrm{Tr}\left[\bar\psi\frown (i\!\not \!\!D_A\psi-m\psi)-\tfrac{1}{2}\#F^*\frown F\right]
  $$
  and $\not \!\!D_A\psi:=\gamma \frown D_A\psi+\gamma \spleen \bar D_A\psi$ so that $\not \!\!D_0=\not\!\!\partial$. But in contrast to $\mathcal{L}[\psi,U]$, the Lagrangian $\mathcal{L}'[\psi,U]$ is nonlocal with respect to the gauge group field $U$.
\end{remark}

\begin{proof}[Proof of Corollary~\ref{cor-probability-conservation}]
Since $\mathcal{L}[e^{-it}\psi]=\mathcal{L}[\psi]$ for each $t\in\mathbb{R}$, we have~\eqref{eq-invariance} with $\Delta=-i\psi$. Then by a version of Theorem~\ref{th-Noether} for complex-valued fields,
Lemma~\ref{l-Dirac-derivatives}
for $A=0$, and the identity  $\psi^*\frown\phi^*=(\phi\smile\psi)^*$ from Lemma~\ref{l-identities}, we have the conserved current
\begin{multline*}
j[\psi]=
\mathrm{Re}\left[\tfrac{\partial\mathcal{L}[\psi]}{\partial(\delta\psi)}
\frown \Delta\right]=
\mathrm{Re}\left[
(-i\gamma^0\gamma\smile \psi)^*\frown (-i\psi)
\right]=\\=
\mathrm{Re}\left[\left(i\psi^*\smile(-i\gamma^0\gamma\smile \psi)\right)^*\right]=
\mathrm{Re}\left[\bar\psi\smile\gamma\smile \psi\right].
\end{multline*}
The conservation of $j^5[\psi]$ is proved analogously, only take $\Delta=i\gamma^5\psi$ and apply the identities $(\gamma^5)^*=\gamma^5$ and $\gamma^5\gamma^0=-\gamma^0\gamma^5$. The conservation of $T[\psi]$ follows from Theorem~\ref{th-energy-conservation}, 
Lemma~\ref{l-Dirac-derivatives},
and the identity $(\gamma^0\gamma)^*=\gamma^0\gamma$.
\end{proof}


\begin{proof}[Proof of Proposition~\ref{th-Dirac-approximation}]
As $\psiup$ is $C^1$, we get $\psiup(\min e)\!\rightrightarrows\!\psiup(\max e)$, $N\delta\psi_N(e)\!\rightrightarrows\!\partial_l\psiup(\max e)$, and
\begin{multline*}
j_N(e)=\mathrm{Re}[\bar\psi_N\smile\gamma\smile \psi_N](e)
= \mathrm{Re}\left[
\bar\psiup(\min e)\gamma^l\psiup(\max e)
\right]\\
\rightrightarrows \mathrm{Re}\left[
\bar\psiup(\max e)\gamma^l\psiup(\max e)
\right]
= \mathrm{j}^l(\max e).
\end{multline*}

For $v\!=\!\max h$, by a version of Proposition~\ref{prop-quadratic-lagrangian} for vector-valued fields and 
Lemma~\ref{l-Dirac-derivatives},
\begin{align*}
  (-1)^lN\langle T_N,h\rangle_k
  &=
  \tfrac{N}{2}\mathrm{Re}\left[
  ([-i\gamma_0\gamma\smile\psi_N]^*(v+\mathrm{e}_l)+
  [-i\gamma_0\gamma\smile\psi_N]^*(v+\mathrm{e}_l-2\mathrm{e}_k))\right]
  \delta\psi_N(v-\mathrm{e}_k)
    \\&-
  N\deltaup^l_k\mathrm{Re}\left[
  \bar\psi_N\frown \left(i\gamma\frown\delta\psi_N-\tfrac{m}{N}\psi_N\right)
  \right](v)\\
  &\rightrightarrows \mathrm{Re}\left[i\bar\psiup\gamma^l \partial_k \psiup
  -\deltaup^l_k\left(i\bar\psiup\gamma^n \partial_n \psiup-m\bar\psiup\psiup\right)\right](v)
  =\mathrm{T}_k{}^l(v).\\[-1.4cm]
\end{align*}
\end{proof}

\begin{proof}[Proof of Corollary~\ref{prop-gauge-invariance-Dirac}]
This follows from Proposition~\ref{l-gauge-covariance}.
\end{proof}

\begin{proof}[Proof of Corollary~\ref{cor-Dirac-charge-conservation}]
By Corollary~\ref{prop-gauge-invariance-Dirac} and Theorem~\ref{th-charge-conservation}
we get the conserved covariant current
$$
j[\psi]=
\left(\tfrac{\partial\mathcal{L}[\psi,U]}{\partial(D_A\psi)}
\frown \psi\right)^*=
\left(
(-i\gamma^0\gamma\smile \psi)^*\frown \psi
\right)^*=
-\bar\psi\smile i\gamma\smile \psi.
$$
By Lemma~\ref{l-identities} we get $\mathrm{Pr}_{T_{1}G}D^*_A j[\psi]
=D^*_A \mathrm{Pr}_{T_{U}G} j[\psi]=0$. Take $G=U(n)$ so that $-i\cdot 1\in T_{1}G$. Then
$
\partial\langle j[\psi],iU\rangle
=\mathrm{Re}\mathrm{Tr}[(-i)D^*_A j[\psi]]=0.
$
\end{proof}

\begin{proof}[Proof of Corollary~\ref{cor-Dirac-full}]
For fixed $\psi\in C^0({I}^d_N;\mathbb{C}^{4\times n})$ the Lagrangian $\mathcal{L}[\psi,U]=:\mathcal{L}[U]$ from Corollary~\ref{cor-Dirac-gauge} is local with respect to $U$. By Lemma~\ref{l-lagrangian-functional-derivative-charge} and row~6 of Table~\ref{tab-derivatives} we get
$\frac{\partial\mathcal{L}[U]}{\partial U}=j[\psi]^*$ and $\frac{\partial\mathcal{L}[U]}{\partial (F[U])}=\# F^*$, where $j[\psi]$ is given by Corollary~\ref{cor-Dirac-charge-conservation}.
Let $U$ be stationary for $\mathcal{S}[\psi,U]=\epsilon \mathcal{L}[U]$.
By Theorem~\ref{th-Euler-Lagrange-gauge} $U$ satisfies the Yang--Mills equation from Corollary~\ref{cor-gauge} with $j=j[\psi]$.
Again by Theorem~\ref{th-Euler-Lagrange-gauge} $U$ is stationary for $\mathcal{S}[U]$ from Proposition~\ref{l-Wilson}, i.e., $U$ is generated by $j=j[\psi]$. The reciprocal assertion is proved analogously.
\end{proof}

\endarxiv

\small
\vspace{-0.2cm}
\noindent
\textsc{Mikhail Skopenkov\\
HSE University and
King Abdullah University of Science and Technology}
\\
\texttt{mikhail.skopenkov\,@\,gmail$\cdot $com} \quad \url{https://users.mccme.ru/mskopenkov/}

\end{document}

\newpage

\section*{Appendix. Omitted results.}

\subsection{Feynmann checkerboard}

Let us conclude this section by a \emph{nonexample}: a well-known model \emph{not} obtained by the discretization algorithm from Section~\ref{ssec-main-tools} but
still fitting the general formalism of Section~\ref{ssec-statements}. This will be also an example of \emph{nonconvergence} to the continuum limit.

\subsubsection*{Checkerboard photon}

\mscomm{
A toy model using checkerboard paths.

Equation of motion: $2\phi(x,t)+\phi(x,t+2)-\phi(x-1,t+1)-\phi(x+1,t+1)=0$.

Lagrangian $\mathcal{L}[\phi]=\#\delta\phi\frown\delta\phi+
(\phi-\tau\frown\delta\phi)\frown\phi$, where $\tau(e)=\begin{cases}
                 1, & \mbox{if } e\parallel Ot \\
                 0, & \mbox{if } e\perp Ot.
               \end{cases}$

A convergence theorem for $x=o(t)$.

A nonconvergence theorem for large $x$.

NB! The Lagrangian describes the initial Feynmann's function $\phi(x,t)$, without the normalization $2^{-t/2}\phi(x,t)$. A modification is probably required!
}

\subsubsection*{Checkerboard electron}

\mscomm{
A toy model using checkerboard paths.
}

For $\psi\in C_k({M};\mathbb{C}^{n\times 1})$ denote
\begin{equation}\label{eq-covar-boundarydual}
  \bar D \psi=(D^*\psi^*)^*=\partial\psi+\psi\frown A
\end{equation}

\begin{lemma}\label{l-cap-delta}
 For each $\phi\in C^k({M};\mathbb{C}^{n\times 1})$, $\psi\in C_l({M};\mathbb{C}^{1\times n})$ and $U\in C^1({M};\mathbb{C}^{n\times n})$ we have
 \begin{equation}\label{eq-cap-delta}
  \psi\frown D \phi=(\bar D \psi)\frown \phi-(-1)^{k}\partial(\psi\frown \phi).
 \end{equation}
\end{lemma}

\begin{proof}
This is obtained by adding the identities $\psi\frown(A\smile \phi)=(\psi\frown A)\frown \phi$ and $\psi\frown \delta \phi=\partial \psi\frown \phi-(-1)^{k}\partial(\psi\frown \phi)$ using~\eqref{eq-covar-coboundary} and \eqref{eq-covar-boundarydual}.

\mscomm{(Beware: Identity for $\partial (\psi\frown\phi)$ taken from Hatcher, the signs in Fomenko--Fuchs are different!)}
\end{proof}

\begin{lemma}[Lagrangian functional derivative]\label{l-Lagrangian-functional-derivative}
For a local Lagrangian $\mathcal{L}\colon C^k({M};\mathbb{C}^n)\times C^1({M};\mathbb{C}^{n\times n})\to C_0({M};\mathbb{R})$
and fields $\phi,\Delta\in C^k({M};\mathbb{C}^n)$ we have
$$
 \left.\frac{\partial \mathcal{L}[\phi+ t\Delta]}{\partial t}\right|_{ t=0}=
 \mathrm{Re}\,\left[\left(\frac{\partial\mathcal{L}[\phi]}{\partial\phi}+
 \bar D\frac{\partial\mathcal{L}[\phi]}{\partial(D\phi)}\right)
 \frown \Delta
 -(-1)^k\partial\left(
 \frac{\partial\mathcal{L}[\phi]}{\partial(D\phi)}
 \frown\Delta\right)\right].
$$
\end{lemma}

\begin{proof} Take a vertex $v\in {M}$. By \eqref{eq-local}, \eqref{eq-Lagrangian-derivatives}, and~\eqref{eq-cap-delta} we have
\begin{align*}
\left.\frac{\partial \mathcal{L}[\phi+ t\Delta]}{\partial t}(v)\right|_{ t=0}
&=
\left.\frac{\partial}{\partial t} L_v([\phi+ t\Delta](\mathrm{e}_{v,k}),[D\phi
+D\,t\Delta](\mathrm{e}_{v,k+1}))\right|_{ t=0}
\\&=
\mathrm{Re}\sum_{j=1}^{p(v,k)}
\frac{\partial }{\partial\phi_j}
L_v(\phi(\mathrm{e}_{v,k}),D\phi(\mathrm{e}_{v,k+1}))
\left.\frac{\partial}{\partial t}
[\phi+ t\Delta](\mathrm{e}_{v,k,j})\right|_{ t=0}
\\&+
\mathrm{Re}\sum_{j=1}^{p(v,k+1)}
\frac{\partial }{\partial\phi'_{j}}
L_v(\phi(\mathrm{e}_{v,k}),D\phi(\mathrm{e}_{v,k+1}))
\left.\frac{\partial}{\partial t}
[D\phi+D\,t\Delta](\mathrm{e}_{v,k+1,j})
\right|_{ t=0}
\\&=
\mathrm{Re}\sum_{j=1}^{p(v,k)}
\frac{\partial \mathcal{L}[\phi]}{\partial\phi}(\mathrm{e}_{v,k,j})
\Delta(\mathrm{e}_{v,k,j})+
\mathrm{Re}\sum_{j=1}^{p(v,k+1)}
\frac{\partial \mathcal{L}[\phi]}{\partial(D\phi)}(\mathrm{e}_{v,k+1,j})
D\Delta(\mathrm{e}_{v,k+1,j})
\\&=
\mathrm{Re}\left[\frac{\partial \mathcal{L}[\phi]}{\partial\phi}\frown
\Delta+
\frac{\partial \mathcal{L}[\phi]}{\partial(D\phi)}\frown
D\Delta\right](v)
\\&=
\mathrm{Re}\left[
  \left(\frac{\partial\mathcal{L}[\phi]}{\partial\phi}+
  \bar D \frac{\partial\mathcal{L}[\phi]}{\partial(D\phi)}\right)
  \frown \Delta
  -(-1)^k\partial\left(
  \frac{\partial\mathcal{L}[\phi]}{\partial(D\phi)}
  \frown\Delta\right)\right](v).
\end{align*}
\end{proof}

\begin{lemma}\label{l-nondegeneracy}
  Let $\phi\in C^k({M};\mathbb{C}^{n\times 1})$. If $\epsilon\mathrm{Re}[\psi\frown \phi]=0$ for each $\psi\in C_k({M};\mathbb{C}^{1\times n})$, then $\phi=0$.
\end{lemma}

\begin{proof} Take $\psi=\phi^*$. Then $0=\epsilon\mathrm{Re}[\phi^*\frown \phi]=\sum_{f}|\phi(f)|^2$, where the sum is over all the $k$-dimensional faces $f$ of ${M}$. Thus $\phi=0$.
\end{proof}

\begin{proof}[Proof of Theorems~\ref{th-Euler-Lagrange} and~\ref{th-Euler-Lagrange-covar}]
A field $\phi$ is \edit{R1P4}{an extremal}, if and only if for each field $\Delta$ 
we have
\begin{multline*}
0=\left.\frac{\partial \mathcal{S}[\phi+ t\Delta]}{\partial t}\right|_{ t=0}=
\epsilon \left.\frac{\partial \mathcal{L}[\phi+ t\Delta]}{\partial t}\right|_{ t=0}=
\\
\epsilon \mathrm{Re}\left[\left(\frac{\partial \mathcal{L}[\phi]}{\partial\phi} +\bar D\frac{\partial\mathcal{L}[\phi]}{\partial(D\phi)}\right)\frown \Delta\right]
 -(-1)^k\epsilon \partial\mathrm{Re}\left[
  \frac{\partial\mathcal{L}[\phi]}{\partial(D\phi)}
  \frown\Delta\right]
=
\epsilon \mathrm{Re}\left[\left(\frac{\partial \mathcal{L}[\phi]}{\partial\phi} +\bar D\frac{\partial\mathcal{L}[\phi]}{\partial(D\phi)}\right)
\frown \Delta\right]
\end{multline*}
The latter two equalities follow from Lemma~\ref{l-Lagrangian-functional-derivative} and the identity $\epsilon\partial=0$.
Since $\Delta$ is arbitrary, by Lemma~\ref{l-nondegeneracy} and \eqref{eq-covar-boundarydual} the resulting equation is equivalent to~\eqref{eq-Euler-Lagrange-covar}. In the particular case of the unit gauge group field  we get~\eqref{eq-Euler-Lagrange}.
\end{proof}

\begin{theorem}[Discrete Noether theorem] 
\label{th-Noether-covar}
   Assume that a local Lagrangian $\mathcal{L}$ is invariant under an infinitesimal transformation $\phi\mapsto\phi+ t\Delta[\phi]$ up to complete divergence, i.e.,
  \begin{equation}\label{eq-invariance-divergence}
    \left.\frac{\partial\mathcal{L}[\phi+ t\Delta[\phi]]}
    {\partial t}\right|_{ t=0}=\partial \mathcal{J}[\phi]
  \end{equation}
  for some $\Delta[\phi]\in C^k(I^d_N;V)$ and $\mathcal{J}[\phi]\in C^1(I^d_N;\mathbb{R})$.
  Then
  $$
  j[\phi]:=(-1)^{k}\frac{\partial\mathcal{L}[\phi]}{\partial(D\phi)}\frown \Delta[\phi]+\mathcal{J}[\phi]
  $$
  is a conserved current, i.e., for each classical field $\phi\in C^k(I^d_N;V)$ we have $\partial j[\phi]=0$.
\end{theorem}

\begin{proof}[Proof of Discrete Noether's Theorem~\ref{th-Noether}]
By the invariance, Lemma~\ref{l-Lagrangian-functional-derivative}, and Theorem~\ref{th-Euler-Lagrange} we get
\begin{multline*}
  0=\left.\frac{\partial \mathcal{L}[\phi+ t\Delta[\phi]]}{\partial t}\right|_{ t=0}
  -\partial\mathcal{J}=
  \left(\frac{\partial\mathcal{L}[\phi]}{\partial\phi}
  +D^*\frac{\partial\mathcal{L}[\phi]}{\partial(D\phi)}
  \right)
  \frown \Delta[\phi]
  -\partial\left(
  (-1)^{k}\frac{\partial\mathcal{L}[\phi]}{\partial(D\phi)}
  \frown\Delta[\phi]+\mathcal{J}\right)=
  -\partial j[\phi].
\end{multline*}
\end{proof}

Finally, let us present a self-contained proof of the formula for
$\partial(\phi\spleen  \psi)$ (alternatively, it can be deduced from the formula for $\delta(\phi\smile\psi)$ by pairing with an arbitrary field $\Delta$ and applying Lemma~\ref{l-nondegeneracy-general} and the identities from the previous paragraph).
By Definitions~\ref{def-boundary} and~\ref{def-cap-general} it follows that
\begin{align*}
  [\delta\phi](a\dots c) &=  \sum_{b:\dim(a\dots b)=1,\dim(b\dots c)=k}\langle a,b,c\rangle\phi(b\dots c)
  - (-1)^{k}\sum_{b:\dim(a\dots b)=k,\dim(b\dots c)=1}\langle a,b,c\rangle\phi(a\dots b);\\
  [\partial\phi](b\dots c) &=\sum_{a:\dim(a\dots b)=1,\dim(a\dots c)=k}\langle a,b,c\rangle\phi(a\dots c)
  - (-1)^{k}\sum_{d:\dim(b\dots d)=k,\dim(c\dots d)=1}\langle b,c,d\rangle\phi(b\dots d),
\end{align*}
where the sums are over all the vertices such that there exist faces $a\dots b,b\dots c\subset a\dots c$ or $b\dots c,c\dots d\subset b\dots d$.
Thus
\begin{align*}
  [\partial(\phi\spleen  \psi)](b\dots c) &=\sum_{a:\dim(a\dots b)=1,\dim(a\dots c)=k-l}\langle a,b,c\rangle[\phi\spleen  \psi](a\dots c)\\
  &- \sum_{d:\dim(b\dots d)=k-l,\dim(c\dots d)=1}(-1)^{k-l}\langle b,c,d\rangle[\phi\spleen  \psi](b\dots d)\\
  &=\sum_{a,d:\dim(a\dots b)=1,\dim(a\dots d)=k,\dim(c\dots d)=l}\langle a,b,c\rangle\langle a,c,d\rangle\phi(a\dots d)\psi(c\dots d)\\
  &- \sum_{d,e:\dim(b\dots e)=k,\dim(c\dots d)=1,\dim(d\dots e)=l}(-1)^{k-l}\langle b,c,d\rangle\langle b,d,e\rangle\phi(b\dots e)\psi(d\dots e)\\
  &=\sum_{a,d:\dim(a\dots b)=1,\dim(a\dots d)=k,\dim(c\dots d)=l}  \langle b,c,d\rangle  \langle a,b,d\rangle  \phi(a\dots d)\psi(c\dots d)\\
  &-\sum_{d,e:\dim(b\dots e)=k,\dim(c\dots d)=l,\dim(d\dots e)=1}  (-1)^{k}\langle b,c,d\rangle  \langle b,d,e\rangle  \phi(b\dots e)\psi(c\dots d)\\
  &+\sum_{d,e:\dim(b\dots e)=k,\dim(c\dots d)=l,\dim(d\dots e)=1}  (-1)^{k}\langle b,c,d\rangle  \langle b,d,e\rangle  \phi(b\dots e)\psi(c\dots d)\\
  &- \sum_{d,e:\dim(b\dots e)=k,\dim(c\dots d)=1,\dim(d\dots e)=l}  (-1)^{k-l}\langle b,c,e\rangle  \langle c,d,e\rangle  \phi(b\dots e)\psi(d\dots e)\\
  &=  [\partial \phi\spleen  \psi+(-1)^{k}\phi\spleen  \delta \psi](b\dots c).
\end{align*}

Now we proceed to the proof of the results from Section~\ref{sec-general}.

\begin{lemma}[Action functional derivative]\label{l-action-functional-derivative-simple}
For a local Lagrangian $\mathcal{L}\colon C^k({M};\mathbb{C}^{1\times n})\times C^1({M};\mathbb{C}^{n\times n})\to C_0({M};\mathbb{R})$
and fields $\phi\in C^1({M};\mathbb{C}^{1\times n})$, $\Delta\in C^1({M};\mathbb{C}^{1\times n})$ we have
$$
 \left.\frac{\partial \mathcal{S}[\phi+ t\Delta,U]}{\partial t}\right|_{ t=0}=
 \left\langle
 \left(\frac{\partial\mathcal{L}[\phi,U]}{\partial \phi}\right)^*+
 D^*_A\left(\frac{\partial\mathcal{L}[\phi,U]}{\partial(D_A\phi)}\right)^*,
 \Delta\right\rangle.
$$
\end{lemma}

\begin{proof}
  Literally as in the proof of Lemma~\ref{l-Lagrangian-functional-derivative} with $\delta\phi$ replaced by $D_A\phi$ and the operator $\mathrm{Re}\,\mathrm{Tr}$ applied to each summand we get
  \begin{equation*}
  \left.\frac{\partial \mathcal{L}[\phi+ t\Delta,U]}{\partial t}
  \right|_{ t=0}
  =
  \mathrm{Re}\,\mathrm{Tr}\,
  \left[
    \frac{\partial \mathcal{L}[\phi,U]}{\partial \phi}\frown\Delta+
 \frac{\partial \mathcal{L}[\phi,U]}{\partial(D_A\phi)}\frown D_A\Delta
 \right].
  \end{equation*}
  Applying the functional $\epsilon$ and
  using Lemma~\ref{l-cap-delta-covar} we get
  $$
  \left.\frac{\partial \mathcal{S}[\phi+ t\Delta,U]}{\partial t}
  \right|_{ t=0}
  =
 \left\langle
    \left(\frac{\partial \mathcal{L}[\phi,U]}{\partial \phi}\right)^*,\Delta\right\rangle+
 \left\langle
 \left(\frac{\partial \mathcal{L}[\phi,U]}{\partial(D_A\phi)}\right)^*, D_A\Delta
 \right\rangle
 =
 \left\langle
 \left(\frac{\partial\mathcal{L}[\phi, U]}{\partial \phi}\right)^*+
 D^*_A\left(\frac{\partial\mathcal{L}[\phi,U]}{\partial(D_A\phi)}\right)^*,
 \Delta\right\rangle.
$$
\end{proof}

\begin{proof}[Proof of Theorem~\ref{th-Euler-Lagrange-covar}]
  A  field $\phi$ is \edit{R1P4}{an extremal}, if and only if $\left.\frac{\partial \mathcal{S}[\phi+ t\Delta,U]}{\partial t}
  \right|_{ t=0}
  =0$ for each $\Delta\in C^k({M},\mathbb{C}^{1\times n})$.
  By Lemmas~\ref{l-action-functional-derivative-simple} and~\ref{l-nondegeneracy-general} this is equivalent to~\eqref{eq-Euler-Lagrange-covar}.
\end{proof}

Denote by
$C^1({M};T_U G)$ the set of all $\Delta\in C^1({M};\mathbb{C}^{n\times n})$ such that $\Delta(e)$ belongs to the tangent space $T_{U(e)}G$ for each edge $e$. The following lemma is proved completely analogously to Lemma~\ref{l-action-functional-derivative} with $\phi$ and $D\phi$ replaced by $U$ and $F[U]$ respectively, using that
\begin{align*}
  \left.\frac{\partial}{\partial t} F[U+ t\Delta]
  \right|_{ t=0}&=\left.\frac{\partial}{\partial t} [\delta (U+t\Delta-1)+(U+t\Delta-1)\smile(U+t\Delta-1)]
  \right|_{ t=0}\\
  &=\delta\Delta+(U-1)\smile \Delta+\Delta\smile(U-1)=D\Delta.
\end{align*}


\begin{lemma}[Action functional derivative]\label{l-action-functional-derivative}
For a local Lagrangian $\mathcal{L}\colon C^1({M};\mathbb{C}^{n\times n})\to C_0({M};\mathbb{R})$
and fields $U\in C^1({M};G)$, $\Delta\in C^1({M};T_U G)$ we have
$$
 \left.\frac{\partial \mathcal{S}[U+ t\Delta]}{\partial t}\right|_{ t=0}=
 \left\langle
 \left(\frac{\partial\mathcal{L}[U]}{\partial U}\right)^*+
 D^*_A\left(\frac{\partial\mathcal{L}[U]}{\partial(F[U])}\right)^*,
 \Delta\right\rangle.
$$
\end{lemma}

\begin{lemma}\label{l-nondegeneracy-gauge}
  Let $\phi\in C_1({M};\mathbb{C}^{n\times n})$. If $\langle \phi,\Delta\rangle=0$ for each $\Delta\in C^1({M};T_U G)$, then $\mathrm{Pr}_{T_U G}\phi=0$.
\end{lemma}

\begin{proof} Take $\Delta=\mathrm{Pr}_{T_U G}\phi$. Then $0=\langle \phi,\mathrm{Pr}_{T_U G}\phi\rangle=
\sum_{e}\langle \phi(e),\mathrm{Pr}_{U(e)}\phi(e)\rangle =
\sum_{e}\langle\mathrm{Pr}_{U(e)}\phi(e),\mathrm{Pr}_{U(e)}\phi(e)\rangle $, where the sums are over all edges $e$,
because $\mathrm{Pr}_{U(e)}\colon \mathbb{C}^{n\times n}\to T_{U(e)}G$ is an orthogonal projection. Since the pairing $\langle \cdot,\cdot\rangle$ on $\mathbb{C}^{n\times n}$ is nondegenerate, it follows that $\mathrm{Pr}_{T_U G}\phi=0$.
\end{proof}

\begin{proof}[Proof of Theorem~\ref{th-Euler-Lagrange-gauge}]
  A gauge group field $U$ is \edit{R1P4}{an extremal}, if and only if $\left.\frac{\partial \mathcal{S}[U+ t\Delta]}{\partial t}
  \right|_{ t=0}
  =0$ for each $\Delta\in C^1({M},T_{U}G)$.
  By Lemmas~\ref{l-action-functional-derivative} and~\ref{l-nondegeneracy-gauge} this is equivalent to~\eqref{eq-Euler-Lagrange-gauge}.
\end{proof}


\begin{proof}[Proof of Proposition~\ref{prop-global-momentum-conservation}]
\mscomm{Rewrite completely!!!}
It suffices to prove that the flux of the $k$-th component across the boundary of a $d$-dimensional face vanishes. For that it suffices to define the fluxes through the boundaries of the \emph{bottom} and \emph {top halves} of the face, and prove that they vanish. We do it for the bottom half only.

The \emph{flux} of the $0$-th component of a partially symmetric type $(1,1)$ tensor $T$ \emph{across the bottom half} of a face $f\parallel \mathrm{e}_0$ is
$$
\sum\limits_{e:e\subset f,e\ni\min f,e\parallel \mathrm{e}_k}T(e\times e_{-,l}).
$$
Denote by $f+\vec{ t$ the parallel translation of a face $f$ along a vector $\vec{ t$. The \emph{flux across the bissector} $f+\mathrm{e}_0/2$, where the faces $f\perp \mathrm{e}_0$, is
$$
\sum\limits_{e:e\subset f,e\ni\min f,e\perp \mathrm{e}_0}\left[T(e\times [e+\mathrm{e}_0])-T(e_{+,0}\times e_{+,0})\right].
$$

Now take a $d$-dimensional cell $g_{+,0}$ with bottom face $g$.
Since the faces of $g_{+,0}$ parallel to $\mathrm{e}_0$ are in bijection with the faces of $g$, the total flux across the bottom halves of vertical faces of $\partial g_{+,0}$ equals
\begin{align*}
  &\sum_{\substack{f:f\subset \partial g_{+,0},\\ \dim f=d-1,\\f\parallel \mathrm{e}_0}}
  \sum_{\substack{e:e\subset f,\\ e\ni\min f,\\ e\parallel \mathrm{e}_0}} T(e\times e_{-,l})
  =
  \sum_{\substack{e:e\subset g,\\e\ni\min g,\\e\perp \mathrm{e}_0}}
  \sum_{\substack{l\ne 0:\\ \mathrm{e}_l\perp e}} \left[T([e_{+,0}+\mathrm{e}_l]\times [e_{+,0}+\mathrm{e}_l]_{-,l})+T(e_{+,0}\times (e_{+,0})_{-,l})\right]\\
  &=\sum_{\substack{e:\min g\in e\subset g}}
   \sum_{\substack{f:e\subset f\subset g,\\\dim f=\dim e+1}}
      \langle e,\partial f\rangle \left[T([e_{+,0}+\mathrm{e}_l]\times f)-T(e_{+,0}\times [f-\mathrm{e}_l])\right],\mscomm{Check signs!!!}\\
\intertext{where in the second line we take $l$ such that $\mathrm{e}_l\parallel f$ and $\mathrm{e}_l\perp e$. Add and subtract the expression $\langle e,\partial f\rangle T(e_{+,0}\times f)$ from each summand, and add the fluxes across the face $g$ and the bissector $g+\mathrm{e}_0/2$}
   &\sum_{\substack{e:\min g\in e\subset g}} \left[T(e\times e)-T(e_{+,0}\times e_{-,0})+T([e+\mathrm{e}_0]\times e)-T(e_{+,0}\times e_{+,0})\right].\\
\intertext{We get the total flux across the boundary of the bottom half of $g_{+,0}$:}
   &\sum_{\substack{e:\min g\in e\subset g}}
   [T(\partial e_{+,0}\times e)-T(e_{+,0}\times \delta e)]
   = \sum_{\substack{e:\min g\in e\subset g}}
   [\partial T](e_{+,0}\times e)=0.
\end{align*}
Here $\partial T=0$ because the tensor is conserved apart the boundary.
\end{proof}

\section{Integral charge conservation}

First state integral charge conservation. Let us define current flux across a hypersurface not transversal to edges. Informally, we just shift the hypersurface slightly in the $(-1,\dots,-1)$-direction.

\begin{definition} A \emph{hyperface} is a $(d\!\!-\!\!1)$-dimensional face of~$I^d_N$. The \emph{flux} of a current $j$ across a hyperface $h\perp\mathrm{e}_l$ \emph{in a positive normal direction} is  $(-1)^l j(h-\mathrm{e}_0-\dots-\mathrm{e}_{d-1})$. The flux is a well-defined function on hyperfaces not contained in the coordinate hyperplanes. Usually it is denoted by $[I^d_N]\frown j$. 

Assume that $d\ge 2$. Let $\pi$ be an oriented piecewise-linear hypersurface consisting of hyperfaces disjoint with the coordinate hyperplanes.
For each hyperface $h\subset \pi$ denote
\begin{equation}\label{eq-orientation}
\langle h,\pi\rangle=\begin{cases}+1, &\text{if the orientations of $\pi$ and $h$ agree},\\
-1, &\text{if the orientations of $\pi$ and $h$ are opposite}.
\end{cases}
\end{equation}
The latter notation is also used, if $\pi$ and $h$ have arbitrary dimension $k>0$.
The \emph{flux} of $j$ across $\pi$ is the sum of the fluxes over all the hyperfaces $h$ of $\pi$ with the coefficients $\langle h,\pi\rangle$.

A current $j$ is \emph{conserved apart $\partial{I}^d_N$}, if $\partial j(v)=0$ for each vertex $v\notin\partial{I}^d_N$.
\end{definition}

\begin{proposition}[Integral charge conservation] \label{prop-global-charge-conservation}
  A current is conserved apart the boundary of the grid ${I}^d_N$, where $d\ge 2$, if and only if the flux of the current across each closed oriented hypersurface consisting of hyperfaces disjoint with the coordinate hyperplanes vanishes. 
\end{proposition}

\begin{proof}[Proof of Proposition~\ref{prop-global-charge-conservation}] This is sufficient to prove in the case when the hypersurface is the boundary of a $d$-dimensional face $h$ disjoint with the coordinate hyperplanes. Then the flux of $j$ across $\partial h$ equals $[\partial j](\min h)$. The latter expression vanishes for each $h$, if and only if $j$ is conserved apart~$\partial I^d_N$.
\end{proof}

\begin{proof}[Second proof of Theorem~\ref{th-charge-conservation}]
  Take arbitrary $\Delta\in C^0({M},T_1 G)$.
  By Lemma~\ref{l-infinitesimal-gauge} it follows that
  $$
  0=\left.\frac{\partial}{\partial t} \mathcal{S}[\phi+ t \phi\smile\Delta,U+ t D_A\Delta]
  \right|_{ t=0}
  =
  \left.\frac{\partial}{\partial t} \mathcal{S}[\phi+ t \phi\smile\Delta,U]
  \right|_{ t=0}+
  \left.\frac{\partial}{\partial t} \mathcal{S}[\phi,U+ t D_A\Delta]
  \right|_{ t=0}.
  $$
  Here the first, and hence the second, summand vanishes for each \edit{R1P4}{extremal} $\phi$ because $\phi\smile\Delta$ is a possible variation of $\phi$. On the other hand, by Lemmas~\ref{l-lagrangian-functional-derivative-charge} and~\ref{l-cap-delta-covar} we have
  $$
  \left.\frac{\partial}{\partial t} \mathcal{S}[\phi,U+ t D_A\Delta]
  \right|_{ t=0}
  =
  \epsilon\,\mathrm{Re}\,\mathrm{Tr}\left[
  \left(\frac{\partial\mathcal{L}[\phi,U]}{\partial(D_A\phi)}
  \frown \phi\right)\frown D_A \Delta\right]
  =
  \langle j[\phi,U],D_A \Delta\rangle=\langle D^*_A j[\phi,U],\Delta\rangle.
  $$
Thus $\langle D^*_A j[\phi,U],\Delta\rangle=0$ for each
$\Delta\in C^0({M},T_1 G)$. By Lemmas~\ref{l-nondegeneracy-gauge} and~\ref{l-identities} we get $D^*_A\mathrm{Pr}_{T_U G} j[\phi,U]=\mathrm{Pr}_{T_1 G} D^*_A j[\phi,U]=0$, as required.
\end{proof}

\end{document}